\documentclass{article}

\usepackage{arxiv}

\usepackage[utf8]{inputenc} % allow utf-8 input
\usepackage[T1]{fontenc}    % use 8-bit T1 fonts
\usepackage{hyperref}       % hyperlinks
\usepackage{url}            % simple URL typesetting
\usepackage{booktabs}       % professional-quality tables
\usepackage{amsfonts}       % blackboard math symbols
\usepackage{nicefrac}       % compact symbols for 1/2, etc.
\usepackage{microtype}      % microtypography
\usepackage{lipsum}
\usepackage{graphicx}

\newcommand{\T}{\top}

\graphicspath{{images/}}

\usepackage{amssymb}
\usepackage{color, graphics, graphicx}
\usepackage{amsthm}
\usepackage{mathrsfs}
\usepackage{fontenc}
\usepackage{lscape}
\usepackage{natbib}
\usepackage{url}
\usepackage{bbm}
\usepackage{dsfont}
\usepackage{lscape}
\usepackage{setspace}
\usepackage{epstopdf}
\usepackage{enumerate}
\usepackage{bigints}
\usepackage{rotating}
\usepackage{booktabs}
\usepackage{multirow}
\usepackage{enumerate}
\usepackage{comment}

\newcommand{\beq}{\begin{equation}}
\newcommand{\eeq}{\end{equation}}
\newcommand{\bfth}{\boldsymbol{\theta}}

\DeclareMathOperator{\var}{var} %var(x)
 
\DeclareMathOperator*{\argmin}{arg\,min}

\newtheorem{theorem}{Theorem}

\newtheorem{corollary}[theorem]{Corollary}

\newtheorem{definition}[theorem]{Definition}

\newtheorem{lemma}[theorem]{Lemma}

\newtheorem{remark}[theorem]{Remark}

\input{AnnasCommands.tex}

\title{Flexible latent variable models on graphs: Laplace approximated inference for multiview network data}

\author{
 Anna van Es \\
  Research Institute for Statistics and Information Science, GSEM\\
  University of Geneva\\
  \texttt{anna.vanes@unige.ch} \\
   \And
 Eva Cantoni \\
  Research Institute for Statistics and Information Science, GSEM\\
  University of Geneva\\
  \texttt{eva.cantoni@unige.ch} \\
  \And
 Davide La Vecchia \\
  Research Institute for Statistics and Information Science, GSEM\\
  University of Geneva\\
  \texttt{davide.lavecchia@unige.ch} \\
}

\begin{document}
\maketitle
\begin{abstract}
We propose a novel and flexible nonlinear approach for dimensionality reduction of large-scale multiview network data and derive its theory. The (linear) predictor incorporates observed covariates (edge-specific, layer-specific, and global) and Gaussian latent factors. Inference is conducted via the graph Laplace approximated maximum likelihood estimator. Letting $K$ denote the number of network layers and $n_V$ the number of nodes, we derive asymptotic theory under two regimes: (i) $K \to \infty$ with fixed $n_V$, and (ii) double asymptotics $K, n_V \to \infty$, establishing consistency and asymptotic normality for local and global parameters, with distinct convergence rates. In an application to the gravity model for commodity trades, we use a zero-adjusted Gamma distribution with latent factors and observable covariates (e.g.\ distance, tariffs, common language) to capture excess zeros, skewness, and unobserved heterogeneity. Synthetic and real-data exercises show that our approach outperforms the routinely applied Poisson pseudo-maximum likelihood estimator with fixed effects. We complement our theoretical and empirical contributions with open-source R/C++ routines and a novel strategy for starting values selection.
\end{abstract}

% keywords
\keywords{Asymptotic theory \and Laplace approximation \and Latent variables \and $M$-estimation \and Multiview networks}

\section{Introduction}
\subsection{Motivation: gravity model for  trading data} \label{Sec:Motiv}

International trade data constitute a high-dimensional, heterogeneous system of exchanges across countries and commodities. Trade flows are networked,   asymmetric, economically informative, and display patterns that vary substantially across countries and products.

The statistical analysis of such network data can be formulated as a learning problem on a random field over a graph, where countries are vertices and bilateral trade flows define weighted and directed edges. Since we consider multiple layers of interactions among the same nodes, we speak about multiview networks.

Analogous network structures arise across several disciplines—for example, in academic coauthorship networks, social media interactions, in protein-protein interactions, and communication systems—leading to a common set of challenges. To illustrate some of  them in the setting of trade flows, we consider the 2022 World Trade Organization (WTO) dataset, a new, previously unanalyzed, dataset obtained directly from the WTO in Geneva, Switzerland. It covers 45 countries, 1980 directed country pairs, and 72 product layers. For each country pair and commodity, the response variable measures
the traded amount (volume) in U.S. dollars; see Section \ref{Sec: RealData} for details. Figure~\ref{fig:EUASIA} displays the trade networks for clocks and watches in Europe and for live animals in Asia. These plots motivate representing countries as nodes and trade flows as edges, justifying our idea of modelling the data as  random fields over graphs (one for each commodity, namely one for each a view of the network). Moreover, they illustrate that the edges
change in number and intensity across regions and products, yielding networks with different degrees of sparsity and heterogeneity. 

\begin{figure}[t]
  \centering

  \begin{minipage}[t]{0.49\textwidth}
    \vspace{0pt}
    \centering

    \begin{minipage}[t][0.15\textheight][t]{\linewidth}
      \centering
      \includegraphics[
        width=0.95\linewidth,
        height=0.55\textheight,
        keepaspectratio
      ]{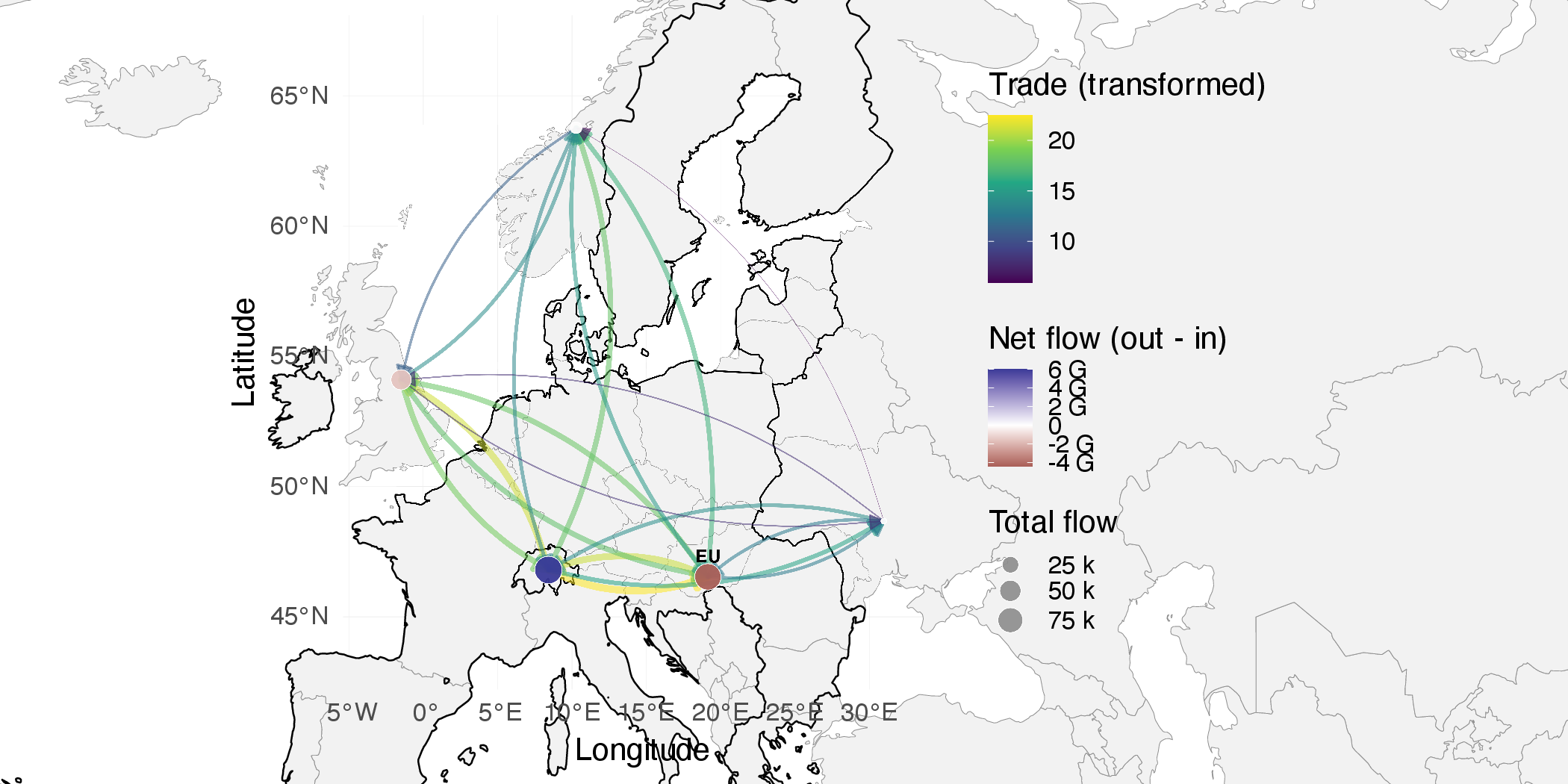}
    \end{minipage}

    \vspace{0.5em}

    {\small \begin{flushleft}
      (a)~Clocks and watches in Europe.
    \end{flushleft}}
  \end{minipage}
  \hfill
  \begin{minipage}[t]{0.49\textwidth}
    \vspace{0pt}
    \centering

    \begin{minipage}[t][0.15\textheight][t]{\linewidth}
      \centering
      \includegraphics[
        width=0.95\linewidth,
        height=0.4\textheight,
        keepaspectratio
      ]{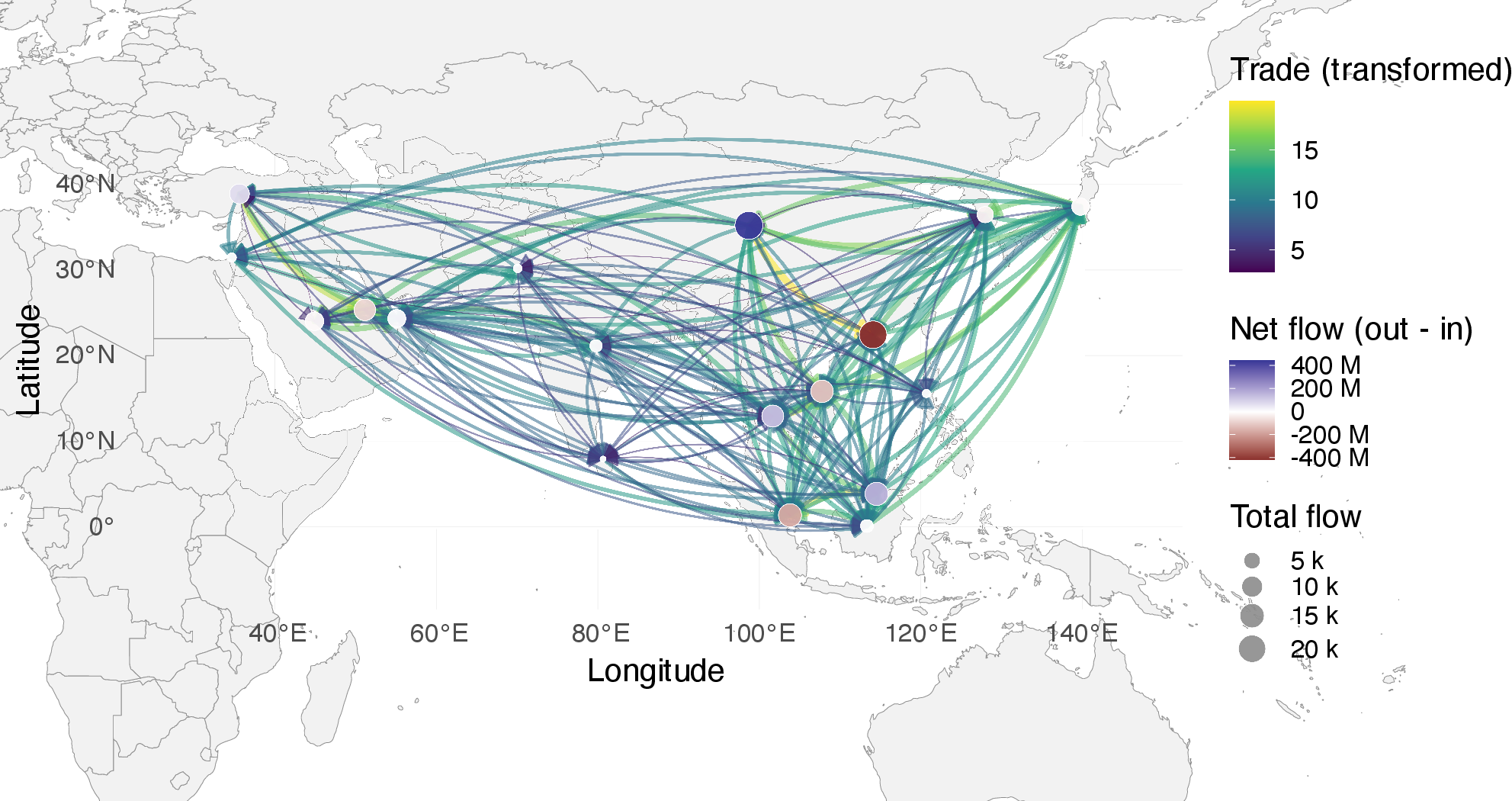}
    \end{minipage}

    {\small \begin{flushleft}
      (b)~Live animals in Asia.
    \end{flushleft} }
  \end{minipage}

  \caption{WTO data: network flows for two commodities in Europe (left) and Asia (right). "G" stands for billions, "M" for millions.}
  \label{fig:EUASIA}
\end{figure}

The economic workhorse framework for such data is the gravity model \citep{AvW03}, where trade depends on economic size, trade frictions, and multilateral resistance. The focus is on the parameters ($\boldsymbol\beta$ and $\bg$, see next section) multiplying covariates, which captures trade elasticities and are central for economic decision makers. 

Despite its theoretical foundation, observed data deviate from gravity model predictions, exhibiting excess zeros, skewness, and overdispersion. This motivates alternative models aimed at capturing the random behaviour of trade flows—see \cite{Yotov2016} or  Section~\ref{Sec: RealData}. 

In this context, \citet{GS22} establish a connection between economic theory and entropic optimal transport (E-OT): equilibrium flows arise from entropy-regularized matching between importers and exporters; see also \citet{GH26}. Their approach introduces randomness via perturbations of economic surplus, thereby generating unobserved heterogeneity and linking economic theory to $M$-estimation. Under specific assumptions, the E-OT constraints yield equations that coincide with the PPMLE first-order conditions. 

Relaxing these assumptions leads to a more general economic framework (still related to E-OT), which in turn motivates flexible statistical models. This provides the starting point for our investigation: we develop a novel inference approach for flexible multiview network models.

In the setting of commodity trading, we let the E-OT equilibrium characterize the conditional mean of trades, while allowing for a general density (either a probability mass function, pmf, or a probability density function, pdf) $\mathsf{f}$ depending on covariates and latent factors (e.g., multilateral resistance).

We consider response distributions whose $\mathsf{f}$ satisfies the usual regularity conditions required for Laplace approximation; see e.g. \citet[Ch. 6]{S10} for a book-length presentation. This class includes the exponential family as a special case, but is not restricted to it. For instance, it encompasses smooth models like zero-inflated Poisson (ZIP), hurdle, and zero-adjusted Gamma (ZAGA) distributions. The precise regularity conditions for  $\mathsf{f}$ are stated in Section \ref{sec:model1_double} and underpin the asymptotic results.  
The choice of $\mathsf{f}$ involves trade-offs: Poisson aligns with E-OT but fits poorly; log-normal ignores zeros; more flexible models  capture stylized facts (e.g. excess zeros and skewness), they preserve interpretability, but their economic link to E-OT is not  studied. From a statistical standpoint, the main challenge of flexible models is likelihood-based inference in multiview networks with high-dimensional latent variables and intractable integrals.

\subsection{Main contributions, with a preview of some empirical results}

The main theoretical and methodological contributions of this paper are as follows.

\textit{(i) Modelling, identification, and estimation.} 
We extend Graph Generalized Linear Latent Variable Models  (GGLLVM, see \cite{JLVR24}) beyond the exponential family while preserving tractable Laplace approximation based inference. The use of flexible models allows accommodating excess zeros and skewness, including covariates and latent factors. This  yields new estimating equations which define the graph Laplace approximated maximum likelihood estimator (GLAMLE), with tailored identification constraints. As an example, we apply our general methodology to WTO data, using the ZAGA: we illustrate that the resulting ZAGA-GLAMLE approach outperforms PPMLE with fixed effects, matching the observed trade volumes  better.
In Figure~\ref{fig:wto-glamle-vs-ppml-cottoncars} we provide a preview of our results. Figure~\ref{fig:wto-glamle-vs-ppml-cottoncars}(a) displays the zero Brier scores for both approaches: PPMLE fails to capture the zero inflation probability, whereas our ZAGA-GLAMLE yields systematically lower errors. Our approach also fits positive trade volumes better, as illustrated by the QQ-plot in Figure~\ref{fig:wto-glamle-vs-ppml-cottoncars}(b), which compares observed versus predicted values for live animals and clearly demonstrates superior accuracy over PPMLE. We refer to Section \ref{Sec: RealData} for more comments on the WTO data analysis. 

Beyond this example on trading data, our methodology applies broadly to other types of high-dimensional data in other scientific areas. For instance, it is useful in disciplines such as psychometrics and ecology, where generalized linear latent variable models (GLLVM) are applied using distributions from the exponential family or the zero-inflated Poisson (ZIP) model; see \citet{HRVF04, BKM11, N17}.

\begin{figure}[t]
  \centering

  \begin{minipage}[t]{0.49\textwidth}
    \vspace{0pt}
    \centering

    \begin{minipage}[t][0.25\textheight][t]{\linewidth}
      \centering
      \includegraphics[
        width=0.95\linewidth,
        height=0.35\textheight,
        keepaspectratio
      ]{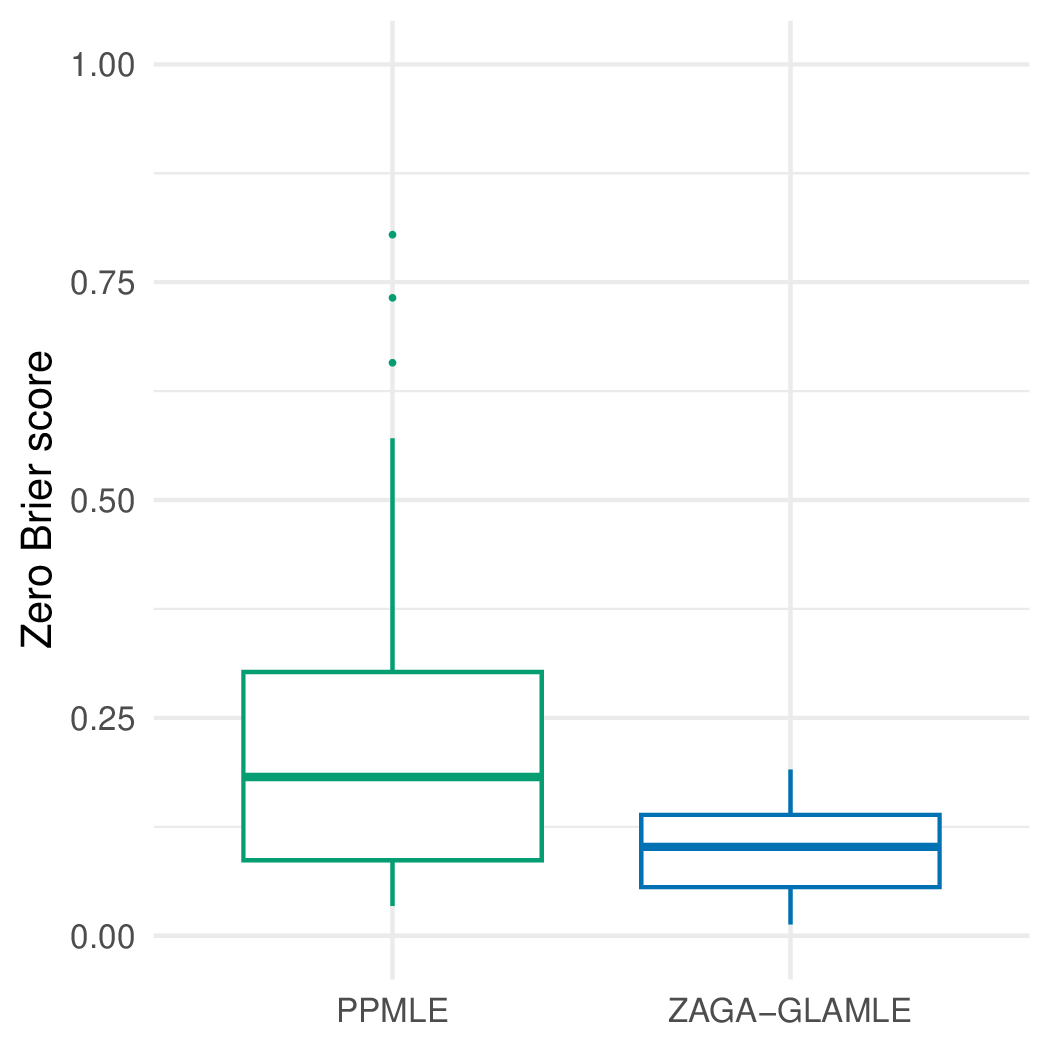}
    \end{minipage}

    \vspace{0.5em}

    {\small \begin{flushleft}
      (a)~Zero Brier scores across all commodities.
    \end{flushleft}}
  \end{minipage}
  \hfill
  \begin{minipage}[t]{0.49\textwidth}
    \vspace{0pt}
    \centering

    \begin{minipage}[t][0.25\textheight][t]{\linewidth}
      \centering
      \includegraphics[
        width=0.95\linewidth,
        height=0.45\textheight,
        keepaspectratio
      ]{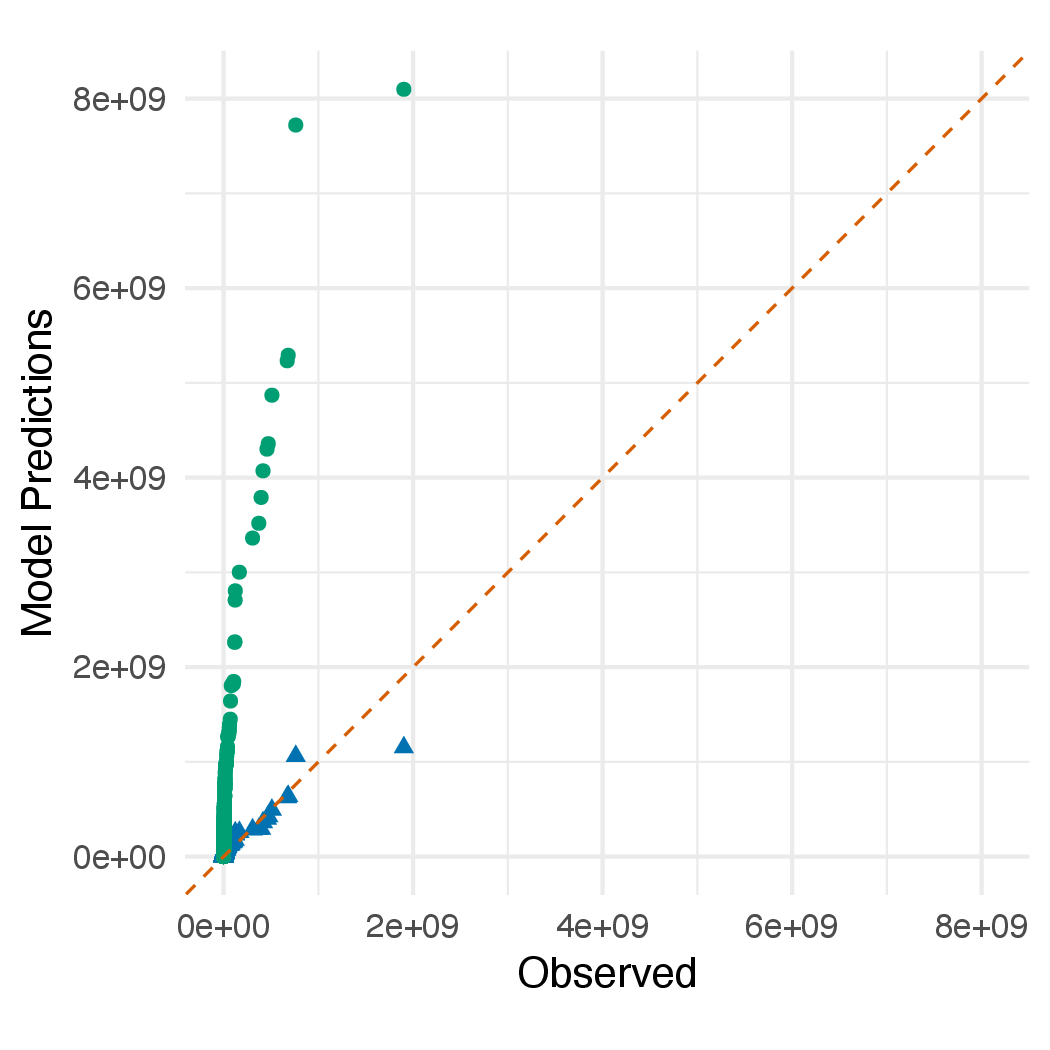}
    \end{minipage}

    \vspace{0.5em}

    {\small \begin{flushleft}
      (b)~QQ-plot of observed (x-axis) versus model predicted (y-axis)
    trade volumes for live animals: GLAMLE (blue triangles) and PPMLE
    (green dots).
    \end{flushleft} }
  \end{minipage}

  \caption{WTO data: comparison between ZAGA-GLAMLE and PPMLE approaches.}
  \label{fig:wto-glamle-vs-ppml-cottoncars}
\end{figure}

\textit{(ii) Asymptotic theory.} We develop a general inferential framework for Laplace-approximated likelihoods in latent-variable multiview network models under increasing dimension. In contrast to the results of \citet{JLVR24}—which are restricted to exponential-family models, fixed network size, and do not incorporate covariates—we establish consistency (with rates) and asymptotic normality under joint growth of the number of layers and the network dimension. This framework explicitly accounts for the interaction between approximation error and sampling variability. This result is important because it provides theoretical justification for the use of Laplace approximations in large scale network settings, where their validity has so far remained unclear. Letting $K$ denote the number of layers (commodities) and $m$ the number of edges  (related to $n_V$ the number of nodes, namely the countries), we study two asymptotic regimes: (i) $K \to \infty$ with $m$ fixed; and (ii) $K, m \to \infty$ jointly. In the latter case, we show that the Laplace approximation remains valid in high dimensions and derive $M$-estimator convergence rates that depend on the joint growth of  $m$ and  $K$. A key feature of our results is that different convergence rates arise for local and global parameters; see Theorem~\ref{thm:model1_double}. This  distinguishes our results from existing Laplace-based inference, such as the one available in \citet{SMcC95}, \citet{Bianconcini2014}, and \citet{ogdenErrorLaplaceApproximations2021}. In particular, it explains why different parameter types may exhibit distinct convergence behaviour.

\textit{(iii) Algorithms and software.} We extend the \texttt{R} implementation of GLAMLE \citep{JLVR24} by incorporating covariates and flexible distributions via Laplace approximation using the Template Model Builder  (TMB, see \citet{Kr15TMB}) framework. We implement the ZAGA specification with various types of covariates, benchmark our model against the PPMLE in the \texttt{gravity} package, and resolve practical numerical issues related to centring/scaling, back-transformation, and initialization. These computational and methodological contributions are entirely new to the literature. The code to replicate our results is hosted in a public \texttt{GitHub} repository\footnote{https://github.com/annavanes/glamle-ggllvm}, which contains also the WTO dataset. Beyond this economic network, we anticipate these algorithms will benefit a broad research community, including researchers working on Laplace approximated inference in ecology and social sciences \citep{N17,Hui22}.

Taken together, our results provide a general and flexible framework for parametric inference in multiview network models, offering both theoretical guarantees and practical tools for high-dimensional data analysis.

\subsection{Related work and structure of the paper}

GLLVM (\citet{BKM11}) extend generalized linear models (GLM) by introducing latent variables and are widely used across disciplines; see, for instance, \citet{HRVF04,O16,N17}. \citet{JLVR24} extend this framework to graph generalized linear latent variable models (GGLLVM) for multiview networks under exponential-family assumptions, but do not address observable covariates or asymptotics as the network size grows. 
Our approach builds on and extends this literature. It connects to inference for mixed-effects models via Laplace approximation \citep{V96,RVL09}, approximate likelihood methods for GLLVM \citep{Bianconcini2014}, and the broader GAMLSS framework \citep{R_etal_19}. Relative to existing (G)GLLVM contributions, we go beyond exponential-family assumptions in a multiview network setting, incorporate both layer-dependent and layer-independent covariates, and develop a double asymptotic theory for all model parameters, including those associated with latent structure and
covariates.

Our work is also related to INLA \citep{rue2009approximate} and Bayesian approaches to multilayer networks \citep{Gollini2016246,Salter-Townshend20171217}. In contrast to these methods, we adopt a frequentist perspective, provide asymptotic guarantees, and mention that our estimation procedure admits an interpretation akin to an expectation-maximization (EM) algorithm. 

The use of latent variables further links our framework to the manifold hypothesis (see \citealp{WGRD26} and references therein), suggesting that high-dimensional network data often lie on low-dimensional structures.

Finally, our work connects to the literature on gravity models, typically estimated via PPMLE on aggregate data \citep{silva2006log,WAR13}. When disaggregated, such data naturally form multiview (one view for each commodity) networks exhibiting excess zeros and overdispersion—features that PPMLE cannot capture by design. While alternative pseudo-likelihood approaches (e.g. negative binomial) handle overdispersion, they do not account for latent heterogeneity. Related extensions with latent factors include \citet{chen2021nonlinear} and \citet{JLVR24}. Compared to these approaches, we incorporate flexible distributions compatible with Laplace approximation, accommodate multiple types of covariates, and study the asymptotic theory of the resulting $M$-estimator.

Due to the complexity and to the flexibility of the considered approach, the theoretical development is technically involved; to maintain readability, we present the main ideas in the body of the paper and defer technical details to the \textcolor{blue}{Supplementary Material} (henceforth \textcolor{blue}{SM}). The remainder of the paper is organized as follows. Section~\ref{Sec: frame} introduces the modelling framework, Section~\ref{Sec: Infe} presents the GLAMLE, Section~\ref{Ex.ZAGA} discusses ZAGA models, and Section~\ref{Sec: uniq} addresses identifiability. Section~\ref{Sec_Asym} develops the asymptotic theory, Section~\ref{Sec: MC_gllvm} reports Monte Carlo results, and Section~\ref{Sec: RealData} presents the empirical application. The main methodological components are contained in Sections~\ref{Sec: frame}, \ref{Ex.ZAGA}, \ref{Sec: uniq}, \ref{Sec: MC_gllvm}, and \ref{Sec: RealData}, while the remaining Sections provide detailed theoretical analysis. Proofs, additional simulation and data analysis results are collected in the \textcolor{blue}{SM}.

\section{Modelling framework} \label{Sec: frame}

We consider a network $\mathcal{G}=(V,E)$, where $V=\{1,\ldots,n_V\}$ is the set of nodes and $E$ the set of edges. For directed networks, each  pair $(i,j)$ with $i\neq j$ defines a dyad. Throughout the paper, the composite index $ij$ is used as a single dyad index, corresponding to the directed node pair $(i,j) \in E$, rather than the $(i,j)$ entry of a matrix.

We model the collection of   random variables $\{Y_{ij}\}$, defined on the edges of $\mathcal{G}$.  In our motivating example about trades, $Y_{ij} \in \mathbb{R_+}$ represents flows from country $i$ to country $j$, so have a directed edge $i \to j$. Since we do not consider self-edges among the nodes (no trades in the same country),  $m= \vert E \vert$ denotes the number of dyads in the network: it is given by $m= {n_{V}(n_{V}-1)}/2 $  for undirected relations and  $m={n_{V}(n_{V}-1)}$ for directed relations.
We use  $i \sim j$, for undirected relations and $i \to j$, in the case of directed relation from $i$ to $j$. 
To capture unobserved heterogeneity and dependence among  edge-associated random variables, we introduce latent variables  $\mathbf{Z}=(1,Z_1,\cdots, Z_q)^\top=(1,\mathbf{Z}_{(2)})^\top \in \mathbb{R}^{q+1}$, for $\mathbf{Z}_{(2)} \in \mathbb{R}^q$, where the integer $q \ll n_V$, 
and assume:  \\

\textbf{A1.} \label{ass:lv-st-norm} \textit{The $\mathbf{Z}_{(2)}$ are $\mathcal{N}(0,\boldsymbol{I}_q)$ distributed and edges are conditionally independent given the latent variables, such that}
$
P_{\bt}(\mathbf{Y}\mid \mathbf{Z}=\mathbf{z})
=
\prod_{i\neq j} P_{\bt}(Y_{ij}\mid \mathbf{z}),
$
\textit{where ${\bt}$ is an unknown (column) parameter vector whose dimension depends on the selected model; see Models 1-6 hereunder for details.} \\

The assumption on the Gaussian distribution of latent variables is common in the (G)GLLVM literature and also the conditional independence; see e.g. \citet{HRVF04,JLVR24}. The main implication of \textbf{A1} is that dependence across edges is induced by the shared latent factors. 
Differently from the existing GGLLVM and latent position models (\citet{RFR16}), our formulation extends them in two directions. First, we allow $\mathsf{f}$ to be any distribution for which likelihood inference is tractable via Laplace approximation (e.g. ZIP and ZAGA), while existing models restrict it
to the exponential family. 
Second, we allow the conditional mean to depend on both latent and observable variables.

Consider also the intercepts $\baa_0 \in \mathbb{R}^m$ and factor loadings $\baa_{(2)} \in \mathbb{R}^{q \times m}$ with $\baa_{ij}=(
  \alpha_{0,ij}, 
  [\baa_{(2)}^\top]_{\cdot, ij})^\top \in \mathbb{R}^{q+1}$, where  $[\baa_{\boldsymbol{(2)}}]_{\cdot, ij} $ represents the $ij$-th column of $\baa_{\boldsymbol{(2)}}$, and $\baa = \left( \baa_{0}^\top, \baa_{(2)}\right)$. 
We consider a multiview network setting:  we are given random views $\bY^{(1)}, \ldots, \bY^{(K)}$ of a network, representing $K$ different types of relational ties among the $n_V$ actors---that is, $K$ network views over the same set of nodes. \textbf{A1} holds for each layer $k=1,\ldots,K$, so 
\begin{equation}
\Yijk \mid \bzk, \bxx^{(k)}_{ij}, \bw_{ij}
\sim
\mathsf{f}_{\bt} (\cdot\mid \eta^{(k)}_{ij}),
\label{Eq. genf}
\end{equation}
where the (linear) predictor $\eta^{(k)}_{ij}$ characterizes the (conditional mean of the) conditional distribution 
and may depend on covariates and latent variables, each with its own coefficient. 
The quantity $\eta_{ij}^{(k)}$ represents the linear predictor associated with edge $(i,j)$ in the $k$-th view.
We consider  $\bxx_{ij}^{(k)} \in \mathbb{R}^{L_x}$, edge- and layer-dependent, and $\bw_{ij} \in \mathbb{R}^{L_w}$, layer-independent covariates. The model can include any type of covariates, e.g.\ continuous, discrete, categorical.  
We also introduce the index $\bt$ to highlight the dependence of $\mathsf{f}$ on the unknown parameter. Thus,
in Equation~\eqref{Eq. genf}, if we set $\mathsf{f}_{\bt}$ as the Gaussian pdf we recover linear latent variable models, while for $\mathsf{f}_{\bt}$ belonging to the exponential family, we have a class of nonlinear models similar to the ones in \cite{chen2021nonlinear,JLVR24}. Beside the selection of $\mathsf{f}_{\bt}$,
the specification of $\eta^{(k)}_{ij}$  yields different models. Hereunder,  we itemize some leading examples that we are going to discuss more in details in the next pages. Other specifications of $\eta^{(k)}_{ij}$ can be considered working on the structure of covariates.  \\

\noindent
-\textit{Model 1 (no covariates)}: 
$\eta^{(k)}_{ij}= \baa_{ij}^{\top}\bzk$; \label{model-lv} \\ 
-\textit{Model 2  (layer-specific covariates, edge-specific coefficients)}:
$\eta^{(k)}_{ij} = \baa_{ij}^{\top}\bzk + \bb_{ij}^{\top}\bxx^{(k)}$; \\
-\textit{Model 3  (covariates vary across edges and layers, common coefficients)}: 
$\eta^{(k)}_{ij} = \baa_{ij}^{\top}\bzk + \bb^{\top}\bxx^{(k)}_{ij}$;\\ 
-\textit{Model 4 (edge-specific covariates and coefficients)}:  
$\eta^{(k)}_{ij} =  \baa_{ij}^{\top}\bzk + \bb_{ij}^{\top}\bxx^{(k)}_{ij}$; \\ 
-\textit{Model 5 (Model 4 with additional global parameters)}: 
$\eta^{(k)}_{ij} = \baa_{ij}^{\top}\bzk + \bb_{ij}^{\top}\bxx^{(k)}_{ij} + \boldsymbol{\gamma}^{\top}\bw_{ij} $; \\
-\textit{Model 6 (Model 3 with additional global parameters)}: 
$\eta^{(k)}_{ij} = \baa_{ij}^{\top}\bzk + \bb^{\top}\bxx^{(k)}_{ij} + \boldsymbol{\gamma}^{\top}\bw_{ij}$. \\

Let $\boldsymbol{\varphi}$ contain all $\varphi_{ij}$ representing other parameters which characterize the underlying distribution (e.g. the overdispersion
and scale that may differ for the edge).
Then the parameter $\bt$ always includes $\mathrm{vec}(\baa)^{\top}$, in addition to $\mathrm{vec}(\boldsymbol{\varphi})^{\top}$,  either $\bb$ or $\mathrm{vec}(\bb)^{\top}$ (where the latter contains all $ \bb_{ij}$), and $\boldsymbol{\gamma}$, when it applies.

These model specifications encode how covariates and latent factors determine  $\eta^{(k)}_{ij}$. 
All models have edge-dependent loadings $\baa_{ij}$ (constant across layers) to account for unobserved heterogeneity. Edge-specific coefficients connect to (non)linear factor models, while common coefficients are closer to classical GLM or mixed-effects formulations. 

\section{$M$-Estimation} \label{Sec: Infe} 
For a selected $\mathsf{f}_{\bt}(\cdot\vert \eta^{(k)}_{ij})$, latent variables are unobserved and need to be integrated out. Thus, the marginal likelihood becomes:
\begin{equation}
\ell_K(\bt) =\sk \ln \l[ \int \l\{ \prod_{i\ne j}^{n_V} \mathsf{f}_{\bt} \l( \Yijk \vert \eta^{(k)}_{ij} \r) \r\}h\l( \bztk \r) d\,\bztk \r].
\label{Eq.loglik}
\end{equation} 
In principle, in (\ref{Eq.loglik}) one may use $\mathsf{f}_{\bt,ij}$, where the additional index $(ij)$ highlights that edge $(i, j)$ (either directed or indirected)  may have its own pdf/pmf which may not coincide with the one of other nodes. This is an additional flexibility feature of our model, which can deal with different types of random fields on edges (e.g. counts, dichotomous, continuous over the entire real line) and  implies that the likelihood is obtained by the combination of different $\mathsf{f}_{\bt,ij}$  (e.g. Poisson, Bernoulli, Gaussian). To lighten the notation,  in the rest of the paper, we assume all edges share the same $\mathsf{f}_{\bt}$. All formulae can be generalized replacing $\mathsf{f}_{\bt}$ by $\mathsf{f}_{\bt,ij}$.

To estimate $\bt$ we should apply the maximum likelihood method, maximizing $\ell_K(\bt)$. However, the integral over the $q$-dimensional space of factors makes the expression in Equation~\eqref{Eq.loglik} analytically intractable, even under the Gaussian assumption in \textbf{A1}. As a consequence, the exact likelihood is not available in closed form and its optimization is numerically challenging, if not impossible. To cope with this issue, we select  $\mathsf{f}_{\bt}$ to belong to the class of pdf/pmf for which the integral in Equation~\eqref{Eq.loglik} admits a Laplace approximation. This  yields our GLAMLE.  In the next subsection we provide the key details.

\subsection{Estimating equations} \label{GLAMLE_key}

Consider the case where,  for each edge $i \to j$ and for each $k$-th layer, the $Y_{ij}$s are real scalars.
 Making use of \textbf{A1}, we assume that the latent variables $\bZt$ have standard normal distributions and that they are independent. 
Then, we rewrite the marginal density function as
\beq
\label{eq:f-mQ}
f_{\bt}(\by)= \int \exp\l\{ m Q\l(\bt, \bxx,  \bw, \bz, \by \r)\r\} d\bz_{(2)}.
\eeq
The exact functional form of $Q\l(\bt, \bxx,  \bw, \bz, \boldsymbol y\r)$ depends on the selected $\mathsf{f}_{\bt}$
and on the specification of $\eta^{(k)}_{ij}$.  In the most general form, we have
\begin{eqnarray} 
mQ\l(\bt, \bxx, \bw,  \bz,\boldsymbol y\r) &=& 
  \snv 
    \ell_{ij} \l(\bt; y_{ij}, \bxx, \bw, \bz\r)
  -
  \frac{\bz^\top_{(2)}\bz_{(2)}}{2}-\frac{q}{2}\ln(2\pi),
\label{Eq: Q_gen}
\end{eqnarray} 
where $\ell_{ij} $  is the log-likelihood associated with $\mathsf{f}_{\bt}$ 
and  
$\sum_{i \ne j}^{n_V} = \underset{i \ne j}{\sum_{i=1}^{n_V}\sum_{j=1}^{n_V}}$. 
For the ease of notation, we let $D_{\bz}^rQ$ denote the tensor of $r$-th derivatives of function $Q$ w.r.t. $\bz_{(2)}$ and evaluated at $\hat{\mathbf{z}}$ (and similarly for 
function $\ell_{ij}$): e.g.,
for $r=2$ we have  the Hessian matrix
\begin{equation}
 \partial_{ \bzt^\top} \partial_{ \bzt}
    Q
    \l(
      \bt, \bxx, \bw, \bz,\by
    \r)\Big\vert_{\bz=\hat{\mathbf{z}}}
    =
    D_{\bz}^2 Q
    \l(
      \bt, \bxx, \bw, \hat{\mathbf{z}},\by
    \r).
    \label{Eq:D2Q}
\end{equation}    
    So, we obtain the Laplace approximated density function
$$ 
\tilde{f}_{\bt}(\by) 
= 
\l(
  \frac{2\pi}{m}\r)^{q/2}
  \det
  \l\{ 
    -U
    \l(
      \bt,   \hat{\mathbf{z}}
    \r)
  \r\}^{-1/2}
  \exp
    \l\{
      mQ
      \l(
        \bt, \bxx, \bw,  \hat{\mathbf{z}},\by
      \r)
  \r\},
$$
where, making use of Equation~\eqref{Eq:D2Q},
$$
  U
  \l( 
    \bt,  \hat{\mathbf{z}}    
  \r)
  =
    D_{\bz}^2 Q
    \l(
      \bt, \bxx, \bw, \hat{\mathbf{z}},\by
    \r)
    =
    -m^{-1}\Gamma
    \l(
      \bt, \bxx, \bw, \hat{\mathbf{z}}
    \r),
$$
and
$$
\Gamma
\l(
  \bt,\bxx, \bw, \hat{\mathbf{z}}
\r)
=
- \snv 
D_\bz^2 \ell_{ij}
  \l(
    \bt; y_{ij}, \bxx, \bw, \bzh 
  \r)+ \boldsymbol{I}_q.
  $$                       
  
  The $\hat{\bz}=(1, (\hat{\bz}_{(2)})^\top)^\top$ maximises $Q\l(\bt, \bxx, \bw,\bz, \by \r)$, therefore is the solution to 
\begin{equation}
\partial_{\bz} Q\l(\bt, \bxx, \bw,\bz, \by \r) =\mathbf 0,
\label{Eq. FOC_Q}
\end{equation}
 and $\hat{\bz}_{(2)}$ is defined through the fixed point equation:
 
\beq \label{Eq: Est_z}
\bzth =
\bzth
\l(
  \bt,\bxx, \bw, \by
\r)
=
\snv 
  D_{\bz}
  \ell_{ij} (\bt; y_{ij}, \bxx, \bw, \bzh)
\eeq

Thanks to the Laplace approximation, we have (see \cite{JLVR24}):
\begin{eqnarray}
f_{\bt}(\boldsymbol{y}) &=& \tilde{f}_{\bt}(\boldsymbol{y})\l\{1+O\l(m^{-1}\r)\r\}, \label{Eq. Lapl}
\end{eqnarray}
which illustrates that the accuracy of $\tilde{f}_{\bt}$ increases as $m$ diverges. The above derivation is for a fixed layer.  For a random sample $(\bY^{(1)}, \bY^{(2)}, \dots, \bY^{(K)})$ containing $n_V$ nodes in each one of the $K$ layers,  Equations Equation~\eqref{Eq: Q_gen}-Equation~\eqref{Eq: Est_z} yield the 
Laplace-approximated likelihood
\begin{eqnarray} 
\tilde{\ell}_K(\bt) 
&=&
\sk 
\l[
  -\frac{1}{2}\ln
  \l\{ 
    \det
    \l\{
      \Gamma
      \l(
        \bt,\bxx, \bw, \zhk
      \r)      
    \r\}
  \r\} 
+
 m Q 
  \l( 
    \bt, \bxx, \bw,\zhk, \boldsymbol{Y}^{(k)}
  \r)
+
\frac{q}{2} \ln 2 \pi 
\r]. \label{Eq: LA_gen}
\end{eqnarray} 

Equipped with $\tilde{\ell}_K(\bt)$ as in Equation~\eqref{Eq: LA_gen}, estimates of the model parameter $\boldsymbol\theta$ are obtained solving the equations yielded by the first order conditions (FOC), where derivatives w.r.t.\ all elements of $\bt$ are set  to zero (or equivalently the approximated log-likelihood is optimized).  

\subsection{Examples}

We consider two examples: when the conditional pmf is the Poisson and when conditional pdf is the ZAGA.  We use them in the numerical experiments; see Section \ref{Sec: MC_gllvm} and Section \ref{Sec: RealData}.

\subsubsection{Poisson} \label{Ex.Poi}

 For the sake of illustration, we select a specific model from the taxonomy in Section \ref{Sec: frame}; similar equations with obvious changes hold for other models. So, we consider Model 2, with Poisson pmf $\mathsf{f}_{\bt}$ and  $\eta^{(k)}_{ij}=\baa_{ij}^{\top}\bzk
+
\bb_{ij}^{\top}\bxx^{(k)}$. 
This is the gravity model with layer-dependent  covariates and latent factors to account for unobserved heterogeneity, it completes the results in \cite{JLVR24} for the pure factor model, to which we refer for the basic equations. 

The Laplace approximated log-likelihood is
\begin{eqnarray*} 
\tilde{\ell}_K(\bt) 
&=&\sk \l(-\frac{1}{2}\ln\l[\det\l\{\Gamma\l(\bt,\bxx, \mathbf{0}, \zhk\r)      \r\}\r]\r) 
\nonumber \\
& + & \sk \l[
\sum_{i\ne j}^{n_V}\l\{Y_{ij}^{(k)} \baa_{ij}^\top \zhk + Y_{ij}^{(k)} \boldsymbol{ \bb}^\top_{ij} \bxx^{(k)} - \exp\l(\baa_{ij}^\top \zhk + \boldsymbol{ \bb}^\top_{ij} \bxx^{(k)} \r) - \ln Y_{ij}^{(k)} !\r\}-
\frac{\hat{\bz}_{\mathbf{( 2 )}}^{(k)^{\top}} \hat{\bz}_{\mathbf{( 2 )}}^{(k)} }{2} \r], \nonumber
\end{eqnarray*} 
where $\hat{\mathbf{z}}^{(k)}$ is the
root of ${\partial}_\bz Q\l(\bt,\bxx ,\mathbf{0}, \mathbf{z},\boldsymbol{Y}^{(k)}\r) = \mathbf 0$, with functions
\begin{equation*} 
m Q\l(\bt,\bxx ,\mathbf{0}, \mathbf{z},\boldsymbol y^{(k)}\r)= \l[  \sum_{i\ne j}^{n_V}\l\{ y_{ij}^{(k)} \eta_{ij}^{(k)}  -
\exp\l(\eta^{(k)}_{ij}\r) - \ln y_{ij}^{(k)}! \r\}
-
\frac{
\hat{\bz}^{\top}_{(2)} \hat{\bz}_{(2)}}{2}-\frac{q}{2}\log(2\pi)\r]
\end{equation*}

and $$\Gamma(\bt,\bxx,\mathbf{0},\bzk)= \sum_{i\ne j}^{n_V} \left\{  \exp\left(\baa_{ij}^\top \bz^{(k)} + \boldsymbol{ \bb}^\top_{ij} \bxx^{(k)}\right) \right\}  \aat + \boldsymbol{I}_q.$$
 
 For each layer, $\Yijk$ are count data. In the presence of excess of zeros one could suggest a ZIP model, which we discuss in Appendix~\ref{AppZIP} of the \textcolor{blue} {SM}. 

\subsubsection{ZAGA} \label{Ex.ZAGA}

Let $Y_{ij} \vert \bz, \bxx, \bw \sim \text{ ZAGA}$, characterized by $\pij, \mij$ and  $\sigma \in \mathbb{R}_+$, as defined in \citet{R_etal_19}, Ch. 9 (see also Appendix~\ref{AppZAGA} of the \textcolor{blue} {SM}). So, for each $k$-th layer, and dyad
\begin{equation}
P_{\bt}\left(Y_{ij} = y_{ij} \mid \bz, \bxx, \bw\right) =  \left\{
    \begin{array}{ll}
        \pij &  \ \text{if} \ y_{ij}=0\\
            \frac{(1-\pij)}{\Gamma \left( 1/\sigma^2 \right)(\sigma^2 \mij)^{1/\sigma^2} } y_{ij}^{1/\sigma^2 - 1} \exp \left(-y_{ij}/\sigma^2 \mij\right) & \ \text{if} \  y_{ij}>0
    \end{array}.
    \right.  \label{ZIGpmf}
\end{equation}
 Now, set for each edge $i \to j$ and for each $k$-th layer, $\ln \mijk =  \eta^{(k)}_{ij}$, where $\eta^{(k)}_{ij}$ can be any model in the taxonomy. For  $\pij=0$, the ZAGA corresponds to a standard Gamma model, whilst values $\pij\in (0,1)$ allow for modelling the excess of zeros.  
We have 
\begin{equation*}
  \small
    \begin{aligned}
        \elk (\bt; \yijk, \bxx, \bw, \bzk) & = \ln \mathsf{f}_{\bt} \left( \yijk \vert \bzk, \bxx, \bw, \bt \right)  
         = \underbrace{\left\{ {\ind{=}} \ln \pij + {\ind{>}} \ln (1-\pij) \right\}}_{:= \llp} \\
        & \quad + \underbrace{{\ind{>}} \left\{  (1/\sigma^2 -1) \ln \yijk  - \ln \Gamma (1/\sigma^2) - 1/\sigma^2 \left(2 \ln \sigma + \ln \mijk \right)   - \frac{\yijk}{\sigma^2 \mijk} \right\}}_{:= \llg (\bzk)},
    \end{aligned}
\end{equation*}
where $\eta^{(k)}_{ij}$ characterizes the conditional mean of the $k$-th layer.  So,
 Equation~\eqref{Eq: LA_gen} becomes
\begin{equation*}
\begin{aligned}
    \tilde{\ell}_{K} (\bt) 
    & = \sk \Biggl(-\frac{1}{2} \ln \left[\operatorname{det}\left\{\Gamma\left(\bt, \bxx, \bw , \hat{\bz}^{(k)}\right)\right\}\right]
    -
    \frac{\hat{\bz}_{\mathbf{( 2 )}}^{(k)^{\top}} \hat{\bz}_{\mathbf{( 2 )}}^{(k)}}{2}\\
    & + \snv \Biggl\{ {\ind{>}} \left( (1/\sigma^2 -1) \ln \yijk  - \ln \Gamma (1/\sigma^2) - 1/\sigma^2 \left(2 \ln \sigma + \ln \mijk \right)   - \frac{\yijk}{\sigma^2 \mijk} \right)  \Biggr\} \Biggr),\\
\end{aligned}
\end{equation*}
where
$\hat{\mathbf{z}}^{(k)}$ is the
root of $\partial_\bz Q\l(\bt,\bxx ,\bw, \mathbf{z},\boldsymbol{Y}^{(k)}\r) = \mathbf 0$, with function $m Q (\bt, \bxx ,\bw, \mathbf{z},\boldsymbol{y}^{(k)})$ equal to  
\begin{equation*}
    \begin{aligned}
        &  
         \Biggl[\snv {\ind{>}} \Biggl\{ \left(\frac{1}{\sigma^2} -1\right) \ln \yijk  - \ln \Gamma \left(\frac{1}{\sigma^2}\right) - \frac{1}{\sigma^2} \left(2 \ln \sigma + \ln \mijk \right)   - \frac{\yijk}{\sigma^2 \mijk} \Biggr\} \\
        & \quad - \frac{q}{2}\ln (2 \pi) - \frac{\bztkt \bztk}{2} \Biggr]
    \end{aligned}
\end{equation*}
and $\Gamma\left(\bt, \bxx, \bw, \hat{\bz}^{(k)}\right) 
=  \snv {\ind{>}} \l( {\yijk}/{\sigma^2 \mijk} \r) \aat + \boldsymbol{I}_q.
$

\subsection{Solving the estimating equations}

Although our $M$-estimation approach is methodologically clear,  a number of theoretical aspects need special care. 

In the Laplace approximation, the latent factors $\bz$ are implicitly treated as  parameters, which are needed to approximate the integrals characterizing the marginal likelihood. Due to maximization of function $Q$ in Equation~\eqref{Eq: Q_gen} (which yields Equation~\eqref{Eq. FOC_Q}), $\hat{\bz}_{(2)}^{(k)}$ can be formally interpreted as the maximum likelihood estimates of the latent factors in the $k$-th network view: each $\hat{\bz}_{(2)}^{(k)}$ depends on the model parameter $\bfth$, the covariates,  and on the observation $\boldsymbol{y}^{(k)}$. Once the latent variables are estimated, the GLAMLE $\boldsymbol{\hat{\bfth}}$ is interpretable as an $M$-estimator; see among the others \citet{H81} and  \citet{vdW98} for book-length introduction.  To elaborate further in the setting of this paper, let us define $\boldsymbol{\tilde{\mathcal S}}_K(\bfth)= \partial{\tilde\ell_K(\bfth)}/\partial {\bfth}$, whose $k$-th component (for the $k$-th layer) is denoted by 
$\boldsymbol{\tilde{\mathcal S}}^{(k)}({\bfth}) = 
\partial \ln \tilde{f}_{\boldsymbol{\bfth}}(\boldsymbol{y}^{(k)})/\partial {\bfth} =   \partial_{\bfth}  \ln \tilde{f}_{\boldsymbol{\bfth}}(\boldsymbol{y}^{(k)}) $.  By definition, the GLAMLE ${\hat{\bfth}}$ 
solves the FOC:
\begin{equation}
\boldsymbol{\tilde{\mathcal S}}_K(\bfth)= \sk \boldsymbol{\tilde{\mathcal S}}^{(k)}({{\bfth}}) = \boldsymbol{0}. \label{Eq. GLAMLEshort}
\end{equation}
In principle, we should write ${\hat{\bfth}}_{K,n_V}$ instead of ${\hat{\bfth}}$, to emphasize the dependence of the $M$-estimator on the $K$ views of the relationships among the $n_V$ nodes. However, for the ease of notation, in what follows we prefer to drop the $K,n_V$.

We notice that the GLAMLE solution to Equation~\eqref{Eq. GLAMLEshort} relates to the adaptive Gauss-Hermite MLE of \cite{Bianconcini2014} and shares an EM interpretation: the E-step uses Laplace approximation, and the M-step maximises the approximated expected score. See \cite{RVL09} for analogous considerations.

\subsection{Identification} \label{Sec: uniq}
A well-known problem in the  literature on factor models is that latent factors and their loadings are not unique; see e.g. \cite{HRVF04, BL12, N17, JLVR24}.  
The same issue appears also for the solution of the estimating equations in Equation~\eqref{Eq. GLAMLEshort}. To see this, let us consider that $\eta_{ij}$ contains 
$\baa_{ij}^\top \bZ = \alpha_{0,ij} + [\baa_{\boldsymbol{(2)}}]_{\cdot, ij} \bZt$, where $\bZt$ are column vector having $q$-variate normal distribution. Let $\boldsymbol{O}_q$ be an orthogonal square matrix of dimension $q$ and $\boldsymbol{O}_{q+1}$ an orthogonal matrix with block $\boldsymbol{O}_q$ such that
$$\boldsymbol{O}_{q+1} = \left [ \begin{smallmatrix} 1 & \boldsymbol{0}_q^\top \\ \boldsymbol{0}_q & \boldsymbol{O}_q \end{smallmatrix} \right],$$
where $\boldsymbol{0}_q$ is a vector of zeroes of length $q$. It is possible to rotate the matrix $\baa = \{\baa_{ij}\}$ premultiplying it by $\boldsymbol{O}_{q+1}$ and thus obtaining a new matrix of parameters 
$
\boldsymbol{\tilde\alpha} = \boldsymbol{O}_{q+1} \baa $.

Similarly, we may consider $\mathbf{\tilde{Z}} = \boldsymbol{O}_{q+1} \bZ = ( \boldsymbol{1} \quad \boldsymbol{O}_{q}\bZt)^\top = ( \boldsymbol{1} \quad \mathbf{\tilde{Z}}_{(2)})^\top$. 
Clearly,  $\eta_{ij}$ does not change:  
 the original and the rotated solutions are observationally equivalent and only the linear span of the factors can be estimated.  To tackle this issue, we introduce \\

\textbf{A2} (i) \textit{The $q$ rows and first $q$ columns of $\baath$ form an upper triangular block: for rows $r = 1, \ldots, q$ and columns $c = 1, \ldots, m$, $[\hat{\baa}_{(2)}]_{r,c} = 0$ for $1 \le c < r \le q$. Furthermore, the diagonal entries satisfy $[\hat{\baa}_{(2)}]_{r,r} = 1$. This could be loosened by allowing $[\hat{\baa}_{(2)}]_{r,r} > 0$}. \\

Thanks to \textbf{A2} (i), the factor loadings matrix is uniquely determined, so it is guaranteed that we obtain a unique solution to (\ref{Eq. GLAMLEshort}); see Proposition 1 in \citet{HRVF04}. This results in $q(q-1)/2 + q$ constraints, with the first term removing the continuous orthogonal indeterminacy and the second removing the residual column-sign ambiguity. We refer to \cite{BL12}, pp.$442-443$ for a similar discussion.

Differently from existing results on GGLLVM (\cite{JLVR24}) and GLLVM (\cite{HRVF04}), the use of covariates entails the need for additional restrictions, which change accordingly to the  type of considered variables. To begin with, let us consider the models with layer-dependent covariates, i.e. Models~2-4. The three models have an edge-specific intercept plus a covariate term. We assume the following  \\

\textbf{A2} (ii) \textit{Consider $L_x$ layer-dependent covariates and their associated design matrices $\boldsymbol{X}_{ij} \in \mathbb{R}^{K \times L_x}$. For all edges $(i , j) \in E$, design matrix $\boldsymbol{X}_{ij}$ must be of full rank.} \\

Finally, in the presence of $L_w$ layer-independent covariates, i.e. for Models~5-6, let $\bW  \in \mathbb{R}^{m \times L_w}$ be the design matrix associated with covariates $\bw$. After inspecting the expression for the linear predictor $\eta^{(k)}_{ij}$, the layer-invariant part of the model is $\bs :=\baa_0+ \bW \bg \in \mathbb{R}^m$. For any $\boldsymbol{\delta} \in \mathbb{R}^{L_{w}}$, define $\tilde{\baa}_0:=\baa_0+ \bW \boldsymbol{\delta}$, $\tilde{\boldsymbol{ \bg}}:=\bg-\boldsymbol{\delta}$. Then $\tilde{\baa}_0+ \bW \tilde{\boldsymbol{\bg}}= \baa_0 + \bW \bg = \bs$. Hence, $(\baa_{0}, \bg) $ is not identified:  the solution set is the whole $L_{w}$-dimensional affine space $\left\{\left(\baa_0+ W\boldsymbol{\delta}, \bg -\boldsymbol{\delta}\right): \boldsymbol{\delta} \in \mathbb{R}^{L_{w}}\right\}$. The following assumption resolves the issue: \\

\textbf{A2} (iii) \textit{ In the presence of layer-independent covariates, set $\bW^\top \baa_{0} = \boldsymbol{0}$.}  \\

Assumption~\textbf{A2} (iii) allows separating the sum $\bs$ in a way such that the edge-specific intercept $\baa_0$ could not contain any component that is explainable by the covariates $\bW$. In other words, it forces $\baa_0$ to be the part of the baseline signal that is orthogonal to the covariate space spanned by the columns of $\bW$. Imposing a linear constraint $\bW^\top \baa_{0} = \boldsymbol{0}$ implies that every valid $\baa_0$ must lie in the null space (kernel) of $\bW$. 
Taken together,  \textbf{A2} (i)-(iii) identify the model parameters. They are easily implemented in software and, as shown in Section~\ref{Sec_Asym}, simplify the asymptotic theory while guaranteeing convergence of the estimated loadings without sign indeterminacy. If one wishes to interpret the estimated factor loadings, the constraints in  \textbf{A2} imply that the order of the edges matters. However, estimation of $\pij$ in ZAGA is unaffected, as rotated versions of factors and loadings are observationally equivalent. Appendix~\ref{SecIdentif} goes into more details on the constraints and their numerical implementation. 

\section{Asymptotic theory} \label{Sec_Asym}

We study the asymptotics of GLAMLE under two regimes. Regime (i): $n_V$ fixed, $K \to \infty$; the GLAMLE is a misspecified MLE converging to a pseudo-true value. Regime (ii): $m = \Theta(K^{\varrho})$ with\footnote{We recall that for generic functions $f$ and $g$, we write $f = \Theta(g)$ iff $f = \bigo{g}$ and $g = \bigo{f}$.}  $\varrho > 0$, $K \to \infty$, so $m \to \infty$ as well; the GLAMLE converges to the true parameter and is asymptotically normal, with different convergence rates for edge-specific and global parameters.

\subsection{Asymptotic regime (i):  $K \to \infty$} \label{Sec: Inter}

This asymptotic regime is already discussed in  \cite{JLVR24}, to which we refer for details. Here we briefly recall 
the central aspects and we elaborate more on testing procedures. 

We denote  as $\text{int}(\bt)$ the interior of 
the parameter space $\bt$
and  define $\boldsymbol{\mathcal{S}}_K(\bfth) = \partial {\ell _K(\bfth)}/\partial {\bfth}$. In analogy with (\ref{Eq. GLAMLEshort}), 
the contribution of the $k$-th layer to the exact likelihood score is $\boldsymbol{\mathcal{S}}^{(k)}({\bfth})=\partial  \ln {f}_{\boldsymbol{\bfth}}(\boldsymbol{y}^{(k)})/\partial {\bfth} = \partial_ {\bfth}  \ln {f}_{\boldsymbol{\bfth}}(\boldsymbol{y}^{(k)}) $, so $\boldsymbol{\mathcal{S}}_K(\bfth) = \sk \boldsymbol{\mathcal{S}}^{(k)}({\bfth}) $. An $M$-estimator is obtained solving  $\boldsymbol{\mathcal{S}}_K(\bfth) = \mathbf 0$: this is the exact MLE whose asymptotic theory can be derived following \cite{vdW98}. However,  the GLAMLE is obtained solving  $\boldsymbol{\tilde{\mathcal S}}({\bfth}) = \boldsymbol{0}$, which makes use of the pseudo-likelihood $\tilde\ell_K$, and not of the exact $\ell_K$. This implies that $\hat{\bfth}$ is a consistent estimator of the pseudo-true value $\tilde{\bfth}= \arg \min \mathbb{E}_0[\tilde{f}_{\bt}(\by)]$, where $\mathbb{E}_0[\cdot]$ indicates that the expected value is taken with respect to the true (unknown) measure. Thus, $\hat{\bfth}\overset{\mathcal{P}}{\rightarrow}\tilde{\bfth}$, as $K \to \infty$.  Moreover,
under standard regularity conditions, we have
$\sqrt{K} (\hat{\bfth}- \tilde{\bfth})\overset{\mathcal{D}}{\Rightarrow}\mathcal{N}(\boldsymbol 0, \boldsymbol V(\tilde{\bfth})),$
where $$\boldsymbol V(\tilde{\bfth}) = \boldsymbol B(\tilde{\bfth})^{-1} \boldsymbol A(\tilde{\bfth})  [\boldsymbol B(\tilde{\bfth})^{-1}]^\top,$$ and
$
\boldsymbol A(\tilde{\bfth})=\mathbb{E}_0 [( \partial_{\bfth} \ln \tilde{f}_{\boldsymbol{\bfth}}(\boldsymbol{y})\vert_{\bfth=\tilde{\bfth}}) (\partial_{\bfth} \ln \tilde{f}_{\boldsymbol{\bfth}}(\boldsymbol{y})\vert_{\bfth=\tilde{\bfth}})^\top],$  $ \boldsymbol B(\tilde{\bfth})= -\mathbb{E}_0[ \partial_{\bfth^\top} (\partial_{\bfth} \ln \tilde{f}_{\boldsymbol{\bfth}}(\boldsymbol{y}))\vert_{\bfth=\tilde{\bfth}}].$

 These results 
are useful to construct asymptotic confidence intervals and for hypothesis testing on $\tilde{\bfth}$.  For instance, assume that we want  to test 
$
\mathcal{H}_0: h(\tilde{\bfth}) = \boldsymbol 0$  vs $
\mathcal{H}_1 : h(\tilde{\bfth}) \neq \boldsymbol 0, 
$
where $h : \boldsymbol \Theta \to \mathbb{R}^r$ is a continuous  function of ${\bfth}$ such that its Jacobian at the pseudo-true value is finite with full row rank $r< \text{dim}\boldsymbol\Theta$---this is to test if some parameter in the misspecified likelihood can be set to zero. Consider also $\hat{\boldsymbol {V}}_K(\hat{\bfth})$ and $\hat{\boldsymbol {B}}_K(\hat{\bfth})$, where the estimated matrices replace the expectation with an empirical average. Then,  
$$\mathcal{W}_K(\hat{\bfth}) = K h(\hat{\bfth})^\top \l[ \partial_{\bfth} h(\hat{\bfth}) \hat{\boldsymbol {V}}_K(\hat{\bfth}) \partial_{\bfth^\top} h(\hat{\bfth})  \r]^{-1} h(\hat{\bfth})$$
is distributed 
under $\mathcal{H}_0$ as a central chi-square with $r$ degrees-of-freedom, $\chi^2_r$. 
Additionally,  let $\bar{\boldsymbol \theta} = \arg\max_{\bfth \in \Theta} \tilde{\ell}_K(\bfth)$, {subject to} $ h(\bfth) =0$, 
and define the Lagrange Multiplier (LM) test statistic 
$$\mathcal{LM}_K (\bar{\bfth}) = \left[\boldsymbol{\tilde{\mathcal S}}_K(\bar{\bfth})  \hat{\boldsymbol {B}}_K(\bar{\bfth})^{-1} \partial_{\bfth^\top} h(\bar{\bfth})\right]
\l[ \partial_{\bfth} h(\bar{\bfth}) \hat{\boldsymbol {V}}_K(\bar{\bfth})  \partial_{\bfth^\top} h(\bar{\bfth}) \r]^{-1} \nonumber  
  \left[ \boldsymbol{\tilde{\mathcal S}}_K(\bar{\bfth})  \hat{\boldsymbol {B}}_K(\bar{\bfth})^{-1} \partial_{\bfth^T} h(\bar{\bfth})\right]^\top, $$
also distributed as $\chi_r^2$, under the null.
One can prove 
that $\mathcal{LM}_K \overset{p}{\to} \mathcal{W}_K(\hat{\bfth})$, as $K \to \infty$.  This result establishes  the usual asymptotic equivalence of the Wald and LM test statistics also in our graph setting, when estimation is conducted via GLAMLE.  However, for testing the hypothesis $\tilde{\bfth}=\bfth^{(0)}$,  the  Likelihood Ratio (LR) test statistic $\mathcal{LR}_K (\hat{\bfth}) = -2 K \ln {\boldsymbol{\tilde{\mathcal S}_K}({\bfth}^{(0)})}/{\boldsymbol{\tilde{\mathcal S}_K}(\bar{\bfth})}$
 is not asymptotically equivalent to $\mathcal{W}_K$. This is because the likelihood is approximated and the second Bartlett equality does not hold.

\subsection{Asymptotic regime (ii):  $K\to \infty$ and $m \to \infty$}
\label{sec:model1_double}

\subsubsection{Consistency  and asymptotic normality}

Due to the order of the error of the Laplace approximation, see Equation~\eqref{Eq. Lapl}, one should let $m \to \infty$.  Therefore, proving consistency for the whole vector of parameter estimates requires {double asymptotics}, where $K\to\infty$ {and} ($n_V\to\infty$, {so }) $m \to\infty$, jointly. 
To have an intuition of the mechanism behind this, consider that the Laplace approximation induces an approximation error of order $O(m^{-1})$ per layer; see \cite{JLVR24} and \citet[Prop.~3.2.1]{Bianconcini2014} for a related discussion. Hence, letting $m\to\infty$ suppresses the Laplace approximation bias and stabilizes the mode and curvature terms. Combining $K\to\infty$ (sampling variability) with $m\to\infty$ (approximation error which entails a bias) gives consistency and asymptotic normality for the full  vector $\boldsymbol{\hat\theta}$, including all edge–specific estimators. In this setting, the GLAMLE converges to the true value, call it $\bt_0$. This result requires that $m$ diverges with $K$ at a certain pace, which depends on  the nature of the considered parameters (either edge-specific or global parameters), reflecting the information content (essentially the number of observations) that sample has for each parameter. Thus, we partition $\bt=(\bt^\top_E, \bt^\top_G)^\top$, where we stack the edge-dependent parameters $\bt_E $ whose dimension depends on $m$  and the fixed-dimensional and global parameter $\bt_G$. For instance, for Model 5 with ZAGA, 
we have $\bt_E = \left(\baa, \bb\right)^\top$ and $\bt_G=(\bg, \sigma)^\top$. Moreover,  
let us partition
$ \baa_0 = (\baa_{0, \mathcal{C}}, 
    \baa_{0, \mathcal{F}})^\top$ using constraints in \textbf{A2} (i), stack all constrained and unconstrained entries of $\baat$ into vectors $\baat[, \mathcal{C}], \baat[, \mathcal{F}]$, respectively, and regroup all free parameters into the vector 
$ \baa_{\mathcal{F}} = (  \baa_{0, \mathcal{F}},  \baat[, \mathcal{F}])^\top$. 
 $\dim \left( \vect{(\baa_{\mathcal{F}})}\right)= m(q+1) - \left(L_w + {q(q+1)}/{2} \right)$; 
 $\dim \left( \vect{(\bb)}\right)= L_x m$ for Models 2, 4, and 5, 
 $\dim \left( \bb\right)= L_x$ for Models~3 and 6; 
$\dim \left( \vect{(\bphi)}\right) = d_{\bphi} \le m$  (it is non-zero only for distributions with overdispersion); and
$\dim \left( \vect{(\bg)}\right) = L_w$.

The total amount of parameters is
   $$\pKMe 
   = 
   m\left(q + 1 + 
   \mathbbm{1}_{\left\{ \dim \left( \operatorname{vec}(\bphi) \right) =  m \right\}} +
   L_x \mathbbm{1}_{\left\{ \dim \left( \operatorname{vec}(\bb) \right) = L_x m \right\}} \right) - \left(L_w + {q(q+1)}/{2} \right),$$ for edge-specific parameters, and $$\pKMg = L_w + L_x \mathbbm{1}_{\left\{ \dim \left( \bb \right) = L_x \right\}} + d_{\bphi} \mathbbm{1}_{\left\{ \dim \left( \operatorname{vec}(\bphi) \right) <  m \right\}}$$ for global parameters. 
The theory developed hereunder applies to both edge-specific and global parameters, but the rates depend on whether the aforementioned parameters are defined as edge-specific or global. 

To derive the asymptotic theory, we need a series of assumptions and lemmas,  which help to  control the size of the terms of the asymptotic expansion of the FOC of the GLAMLE.   The proof, essentially, proceeds in two main steps. We first establish consistency and rates of convergence (Step 1); then, we derive  the asymptotic normality, expressing the asymptotic variance (Step 2). The mathematical derivation blends the arguments in \citet{Bianconcini2014} (derived in the case of $\mathsf{f}_{\bt}$ in the exponential family and for the adaptive Gauss-Hermite (AGH) quadrature), with those in \citet{V96}, \citet{RVL09}, and \citet{ogdenAsymptoticValidityNaive2017a}, adapting all of them to our multiview network, which includes  different types of covariates and considers pdf/pmf that are not necessarily in the exponential family.  

 To illustrate our construction, we provide some heuristics, and we focus on the estimation of $\bt_E$: we call $\bt_0$  (instead of $\bt_{0,E}$) its population value. The full proofs are available in Appendices~\ref{app:tech-det} and ~\ref{proof-thm1} of the \textcolor{blue}{SM}. \\
 
\textit{Step 1: Consistency and rates.} 
Let us define     
\begin{equation}
\epsk[\bt] = \ln \tdens-\ln \dens.
\label{Eq: epsk}
\end{equation}  The Laplace expansion yields the approximate score $\tSco$, which  
satisfies
    $$
 \frac {1}{K} \sk \tSco 
    =
    \frac{1}{K} \sk \Sco
    +  \bigo{m^{-2}},
    $$
since $\tSco = \Sco +  \partial_{\bt} \epsk[\bt]$ and 
 $ \partial_{\bt} \epsk[\bt]= O\!\left(m^{-2}\right) $ uniformly in $\bt$.  As $m\to\infty$, the Laplace bias vanishes and averaging across layers takes care of  sampling variability. Indeed, a mean‑value expansion of the FOC in Equation~\eqref{Eq. GLAMLEshort}  around $\bt_0$ yields
$$
\hbt - \bt_0
=
-\,\Bigg[
\frac{1}{K}\sum_{k=1}^K \partial_{\bt^\top} \partial_{\bt}  \ln \tilde f_{\bt}(\boldsymbol y^{(k)})\bigg|_{\bt=\boldsymbol{\bar\theta}}
\Bigg]^{-1}
\Bigg[
\frac{1}{K}\sum_{k=1}^K \tSco[\bt_0]
\Bigg],
$$
with $\boldsymbol{\bar\theta}$ between $\hbt$ and $\bt_0$. Now analyse the two components. We start from the score term 
\[
\frac{1}{K}\sum_{k=1}^K \tSco [\bt_0] 
= \underbrace{\frac{1}{K} \sum_{k=1}^K \Sco [\bt_0] }_{O_p(K^{-1/2}) \text{ by CLT}} 
+ \underbrace{\frac{1}{K}\sum_{k=1}^K  \partial_{\bt} \epsk[\bt]}_{O(m^{-2})},
\]
which can be rewritten as $({1}/{K})\sum_{k=1}^K \tSco [\bt_0] = O_p\!\big(\max\{K^{-1/2}, m^{-2}\}\big).$ 
Next we look at the {Hessian term}. Under regularity conditions, the Hessian matrix converges by the weak law of large numbers to its limit in probability: 
${1}/{K}\sum_{k=1}^K \partial_{\bt^\top} \partial_{\bt}   \ln \tilde f_{\bt}(\boldsymbol{y}^{(k)})\big|_{\bt=\bt_0} \xrightarrow{p} -\boldsymbol{B}_{\boldsymbol{\vt}} (\bt_0),$
 which (by assumption) is non-singular and its inverse 
 is  still $O_p(1)$. Combining these results yields
$ \bigl\| \boldsymbol{\hat\theta} - \bt_0 \bigr\| =
O_p\!\big(\max\{K^{-1/2}, m^{-2}\}\big)$, which 
illustrates the consistency as $K^{-1/2} \to 0$ and $m^{-2} \to 0$.

 \emph{Step 2: Asymptotic normality.} To characterize the limiting distribution, 
 let us assume $m = \Theta(K^{\varrho})$ for some $\varrho > 0$. 
 For the estimation of $\bt_E$, when $\varrho > 1/4$, we have $K^{-\varrho} = o(K^{-1/4})$, so the bias term is asymptotically negligible relative to the stochastic term at the $\sqrt{K}$ scale. Consequently, 
\[
\sqrt{K}\,(\boldsymbol{\hat\theta} - \bt_0)
=
-\Bigg[
\frac{1}{K}\sum_{k=1}^K \partial_{\bt^\top} \partial_{\bt}  \ln \tilde f_{\bt}(\boldsymbol{y}^{(k)})\big|_{\bt=\bt_0}
\Bigg]^{-1}
\frac{1}{K}\sum_{k=1}^K \tSco [\bt_0]
+ o_p(1).
\]

The central limit and Slutsky's theorems yield: 
$$\sqrt{K}\,(\boldsymbol{\hat\theta} - \bt_0) \overset{\mathcal{D}}{\Rightarrow}\mathcal{N}  (\boldsymbol 0;  \boldsymbol{B}_{\boldsymbol{\vt}} (\bt_0)^{-1}\,\boldsymbol{A}_{\boldsymbol{\vt}} (\bt_0)\,[\boldsymbol{B}_{\boldsymbol{\vt}} (\bt_0)^{-1}]^\top),$$
where $\boldsymbol{A}_{\boldsymbol{\vt}} (\bt_0) = \plim K^{-1}
\sk \tSco[\bt_0]\tSco[\bt_0]^\top$.

Similar calculation holds for $\bt_G$ with two main differences: $ \partial_{\bt} \epsk[\bt]= O\!\left(m^{-1}\right) $ for $\varrho > 1/3$, and the scaling factor
$1/K$ becomes $1/(Km)$. To formalize the above heuristics and understand why $\varrho$ takes on different values for edge-specific and global parameters, we need to introduce assumptions \textbf{A3} and \textbf{A4}, which allow to control all the terms in the expansions.  \\

\textbf{A3}  \textit{Let $m \to \infty$ and $K\to\infty$ jointly. 
Consider
$\bt=\big(\{\baa_{ij}\}_{i\neq j},\{\bb_{ij}\}_{i\neq j},\bg,\{\bphi_{ij}\}_{i\neq j} \big)^\top
\in \boldsymbol{\ThKM}$, and set $\pKM=\dim(\boldsymbol{\ThKM}).$ There exists a true parameter $\bt_0\in\mathrm{int}(\boldsymbol{\ThKM})$ and a constant $M$
independent of $(K,m)$ such that $\sup_{\tl\in\boldsymbol{\ThKM}} \mnorm{\tl} \le M < \infty$. Moreover, $\boldsymbol{\ThKM}$ is compact for each fixed pair $(K,m)$ and let $\boldsymbol{\Theta_0}\subset \boldsymbol{\ThKM}$ be compact.}  \\

\textbf{A3} specifies the parameter space and its behaviour in the asymptotic regime, where $\pKM \to \infty$. This is similar to the technical assumptions in \citet[p. 1525]{Bianconcini2014} for the AGH estimator for latent variable models. 
However, for our general model and differently from the assumptions available in the literature, we need also to specify the behaviour of edge and layer specific parameters.  To proceed further, fix $\theta_l$  in $\bt, l = 1, \ldots, \pKM$ and set $\tl = \theta_l$, to lighten the notation. Making use of Equation~\eqref{Eq:D2Q}, for a fixed integer $r\ge 2$, we define 
   $$\boldsymbol{\hat{Q}}_{r}^{(k)} (\tl) = - m \, D^r_{\bz} Q\l(\bt^{l}, \bxx,  \bw,  \mathbf{\hat{z}}^{(k)}, \boldsymbol y^{(k)}\r),$$
where we let $\bt^{l}$ denote the fact that we fix $\theta_l = \theta$ in $\bt$.  
  For notational brevity, we henceforth write $\boldsymbol{\hat{Q}}_{r}(\tl)$ in place of $\boldsymbol{\hat{Q}}_{r}^{(k)}(\tl)$, while emphasizing that this quantity retains its dependence on the layer $k$.
  Note that when $r=2$, $\boldsymbol{\hat{Q}}_2(\tl) = -m U (\bt^{l}, \bzh)$, the negative Hessian matrix (scaled by $m$) evaluated at the mode $\bztkh$. For $r\ge 3$, $\boldsymbol{\hat{Q}}_r(\tl)$ is an $r$-dimensional tensor of dimension $q^r$.
 
 Before moving forward, we introduce the following notation; we refer to Appendix~\ref{app:tech-det} of \textcolor{blue}{SM} for further details. For a tensor $\boldsymbol{A}_r$ with dimension $d^r$, define the componentwise sup-norm (max entry norm) $
    \|\boldsymbol{A}_r\|_{\max}
    :=
    \max \{ 
    |
    (\boldsymbol{A}_r)_{a_1\cdots a_r}
    | \vert, \ 
    1\le a_1,\ldots,a_r\le d
    \} 
    =
    \max_{1\le a_1,\dots,a_r\le d}
    |
    (\boldsymbol{A}_r)_{a_1 \ldots a_r}
    |.
    $
    For a matrix $\boldsymbol{A}$, the spectral norm $\snorm{\cdot}$ is defined as  $
    \snorm{A}
    :=
    \sigma_1(\boldsymbol{A}),$ which is  the largest singular value of $\boldsymbol{A}$. For any vector $\boldsymbol{x}$, $\mnorm{\boldsymbol{x}}=\max_i |x_i|$.
For any matrix $\boldsymbol{A}$, $\mnorm{\boldsymbol{A}}=\max_i \sum_j |A_{ij}|$ (the largest absolute row sum). 
The norm $\mnorm{\cdot}$ will be written as $\norm{\cdot}$ for simplicity.   Equipped with these norms,  we introduce: \\

\textbf{A4.}     \label{ass:reg}
\begingroup
\itshape
    Fix a layer $k$ and let $m \to\infty$ with latent dimension $q$ fixed. 
For each $\bt \in \boldsymbol{\Theta_0}$, whose $l$-th component is equal to $\theta$:  
    \begin{enumerate}[(i)]
    \item \label{ass:unique-optimizer}
    (Unique interior mode)
    Assume $Q$ is continuous in $\bz$ and,  for each layer $k$, it 
    has a unique maximiser $\bztkh(\tl)$. Moreover, there exists $t>0$ (not depending on $k$ or $\theta$) such that
    $Q$ is well-defined on $\mathcal B_k(\bztkh(\tl),t)$, where
    $
    \mathcal B_k(\bztkh(\tl),t)
    :=\{\bzt \in\mathbb{R}^q: \norm{\bzt -\bztkh(\tl)} <t\}.
    $
    \item \label{ass:smooth-z} (Uniform smoothness in the latent variables) 
    Fix an integer $r_{\max}\ge 4$. For each integer $r$ with $2\le r\le r_{\max}$,
    there exist constants $\epsilon_r>0$ and $C_r<\infty$ (both depending on $r$), independent of $(K,m)$, such that
    uniformly over $1\le k\le K$, $\bt  \in \boldsymbol{\Theta_0}$,
    and dyads $i \to j$ (or $i \sim j$ for undirected graph) with $i\neq j$,
    $
    \sup_{\bz \in \Ball{\zhatk(\tl)}{\epsilon}}
    \enorm{D_{\bz}^r\,\ell_{ij}^{(k)}(\tl, \bz)} 
        \le C_r.
    $
    \item \label{ass:mixed-derivs} (Uniform mixed derivatives in $(\bztk,\tl)$)    Fix $r_{\max}\ge 4$ as in Assumption~(\ref{ass:smooth-z}) above.
    For each integer $r$ with $1\le r\le r_{\max}$, there exists $C_{r,1}<\infty$
    independent of $(K,m)$ such that uniformly over $1\le k\le K$, $\tl\in\ThKM,\bt  \in \boldsymbol{\Theta_0}$,
    and dyads $i \to j$ (or $i \sim j$ for undirected graph) with $i\neq j$,
    $
    \sup_{\bz \in \Ball{\zhatk(\tl)}{\epsilon}}
     \enorm{{\partial_{\tl} } D_{\bz}^r\,\ell_{ij}^{(k)}(\tl;\bz)}
    \le C_{r,1}.
    $

    \item \label{ass:curvature-mode} (Curvature) There exist constants $0<c<C<\infty$ such that, uniformly in $k$ and $\tl\in\Theta_0$,
    $
    c \le \lambda_{\min}\! (m^{-1}\boldsymbol{\hat{Q}}_2(\tl))
    \le \lambda_{\max}\! (m^{-1}\boldsymbol{\hat{Q}}_2(\tl))
    \le C,
    $
    where  $\lambda_{\min}\! (m^{-1}\boldsymbol{\hat{Q}}_2(\tl) )$ and  $\lambda_{\max}\! (m^{-1}\boldsymbol{\hat{Q}}_2(\tl) )$ denote the smallest and largest eigenvalue of the matrix $\boldsymbol{\hat{Q}}_2(\tl)$, respectively, which is a symmetric positive definite (hence, invertible) matrix.  \\
    \end{enumerate}
    
\endgroup

Thanks to \textbf{A1-A4} we may state the key asymptotic properties of the GLAMLE. 
In what follows, we consider a subvector $\boldsymbol{\vt}_E$ of $\bt_E$ corresponding to  
edge-specific parameters (for example $\text{vec}(\boldsymbol{\beta})$, containing all dyad specific coefficients $\boldsymbol\beta_{ij}$)  and we label the corresponding true population parameter $\boldsymbol{\vt}_{E,0}$. Otherwise, we consider the  block $\bt_G$ of global parameters and label the corresponding true population parameter $\bt_{G,0}$. 

\begin{theorem}
\label{thm:model1_double}
Suppose Assumptions $\textbf{A1}-\textbf{A4}$ hold. Furthermore, assume that the Laplace approximation is such that the per-layer log-likelihood error in (\ref{Eq: epsk}) is of order $\bigo{m^{-1}}$, uniformly on a neighbourhood of $\bt_0$  (which can be either $\boldsymbol{\vt}_{E,0}$ or  $\bt_{G,0}$). Denoting $\boldsymbol{\hat\theta}$ the GLAMLE estimator of $\bt_0$, we have:

{
\renewcommand{\labelenumi}{(C\arabic{enumi})}
\begin{enumerate} 
    \item \label{thm-c1-cons-rate} \textbf{Consistency and rates:}
    As $K\to\infty$ and $m\to\infty$,
$$\bigl\| \boldsymbol{\hat\theta} - \bt_0 \bigr\|  =   \bigop{\max \left\{ n_0^{-1/2},  m^{-2}\right\}},$$  where $n_0$ is equal to $K$ for edge-specific parameters or equal to $Km$ for global  parameters.
 
    \item \label{thm-c2-as-norm} \textbf{Asymptotic normality:} 
    Assume additionally that $m = \Theta(K^{\varrho})$ for some positive $\varrho$. So, 
    for  $\boldsymbol{\vt}_{E,0}$, if $\varrho > 1/4$, it holds that
                \[
        \sqrt{K} (\boldsymbol{\hat\theta}_E - \boldsymbol{\vt}_{E,0}) \overset{\mathcal{D}}{\Rightarrow}
        \mathcal{N}\left(\boldsymbol 0;\; 
        \boldsymbol{B}_{E}(\bt_0)^{-1}\,\boldsymbol{A}_{E}(\bt_0)\,[\boldsymbol{B}_{E}(\bt_0)^{-1}]^\top
        \right),
        \]
        where $
        \boldsymbol{A}_{E}(\bt_0) = \plim \frac{1}{n_0}\sum_{k=1}^K \boldsymbol{\tilde{S}^{(k)}_E}(\bt_0) \boldsymbol{\tilde{S}^{(k)}_E}(\bt_0)^\top,$  $
        \boldsymbol{B}_{E}(\bt_0) = -\plim \frac{1}{n_0}\sum_{k=1}^K \partial_{\bt_E} \boldsymbol{\tilde{S}^{(k)}_E} (\bt_0),$
        and $\boldsymbol{\tilde{S}^{(k)}_E}(\bt) = \partial_{\bt_E} \ln \tilde{f}_{\bt}(\by^{(k)})$.    For  $\bt_G$, if $\varrho > 1/3$, we have
        \[
        \sqrt{K m} (\boldsymbol{\hat\theta}_G - \bt_{G,0}) \overset{\mathcal{D}}{\Rightarrow}
        \mathcal{N}\left(\boldsymbol 0;\; 
        \boldsymbol{B}_{G}(\bt_0)^{-1}\,\boldsymbol{A}_{G}(\bt_0)\,[\boldsymbol{B}_{G}(\bt_0)^{-1}]^\top
        \right),
        \]
        where $\boldsymbol{A}_{G}(\bt_0)$ and $\boldsymbol{B}_{G}(\bt_0)$ are defined analogously using $\boldsymbol{\tilde{S}^{(k)}_G}(\bt) = \partial_{\bt_G} \ln \tilde{f}_{\bt}(\by^{(k)})$.

    The matrices $\boldsymbol{B}_{E}(\bt_0)$ and $\boldsymbol{B}_{G}(\bt_0)$ are assumed to be nonsingular.
\end{enumerate}
}
\end{theorem}

\subsubsection{Discussion} \label{Sec. Disc} 

\textit{General comments.} 
Theorem~\ref{thm:model1_double} implies that the GLAMLE of $\bb$ and of $\bg$ is consistent and asymptotically normal under double asymptotics. 
 Its convergence rate depends on $K$ and $m$: $O_p(K^{-1/2})$ from $M$-estimation sampling variability, and $O(m^{-2})$ from the Laplace error. The max-rate in (C1) shows the error is constrained by the slower of these two sources. For edge-specific parameters, the stochastic component is of order $O_p(K^{-1/2})$, whereas for global parameters it is of order $O_p((Km)^{-1/2})$, since the latter exploit information from all $Km$ dyadic observations. In both cases, the Laplace approximation contributes an additional bias term whose order is determined by Theorem~\ref{thm:model1_double}. 
 From (C2), we observe that if $m=\Theta(K^{\varrho})$, then the Laplace contribution becomes asymptotically negligible and the GLAMLE attains the same asymptotic distribution as the corresponding infeasible MLE. The required growth condition depends on the parameter type: $\varrho>1/4$ for edge-specific parameters and $\varrho>1/3$ for global parameters.   The stronger condition for global parameters should not be interpreted as these parameters being harder to estimate. On the contrary, global parameters are estimated more efficiently, converging at the faster rate $(Km)^{-1/2}$. The difference arises because global parameters enter all dyadic likelihood contributions simultaneously and are therefore more sensitive to the accumulation of Laplace approximation errors across the network. Consequently, a faster growth of $m$ relative to $K$ is required to render the approximation bias asymptotically negligible. We also notice that $\varrho>1/2$ is a sufficient condition for the GLAMLE to be asymptotically equivalent to the true MLE for all parameter types, mirroring the results in \citet{RVL09}. If $\varrho$ is below the corresponding threshold, the Laplace approximation error is no longer negligible at the asymptotic scale and contributes to the leading-order behaviour of the estimator; see \cite{V96} for a related discussion.
The asymptotic normality result in (C2)  requires $m=\Theta(K^{\varrho})$. This technical condition is discussed in Appendix~\ref{app:tech-det} of the \textcolor{blue}{SM} (see Remark~14 in Appendix~\ref{proof-thm1}) and is related to the admissible growth rate of the number of edge-specific parameters for the validity of the central limit theorem in the high-dimensional regime considered here. 

In the statement of Theorem~\ref{thm:model1_double}, for the sake of simplicity, we assume
a common number of edges $m$ across all $K$ network views. In practice, different layers
may exhibit different numbers of dyads $m_k$, $k=1,\dots,K$. In the WTO dataset, this can be due, for instance,
to non-reporting of trade values for specific commodities, layer-specific trimming of sparsely-traded pairs, 
or missing covariate values for particular dyads. This heterogeneity can be accommodated similarly to 
 \citet{RVL09} (for unbalanced longitudinal designs) and to
\citet{V96} (for the nonlinear mixed-effects setting). Namely, one can let $m:=\min_k(m_k)$:
this does not entail the need for major changes in the existing proofs. Lemmas~\ref{lem:hatQr-order}--\ref{lem:delta-inf} in
Appendix~\ref{app:tech-det} of the \textcolor{blue}{SM} are already stated for a fixed layer $k$, so
considering $m_k$ requires no change to their arguments. The only place where the
common-$m$ assumption is used substantively is in the supremum in Lemma~\ref{lem:delta-inf}, which
controls the Laplace bias uniformly across layers. With heterogeneous $m_k$, this supremum
is instead governed by $\min_k(m_k)$---the supremum over $k$ of a decreasing function
of $m_k$ is attained at the layer with the fewest edges. Consistency (C1) is thus
unaffected. The asymptotic normality result (C2) also continues to hold, but with the
correct normalizing sequence for global parameters given by $\sum_k m_k$, the total number
of dyad--layer observations, rather than $Km$; the growth condition $\varrho>1/3$ applied
to $\min_k(m_k)$ remains sufficient for the Laplace bias to be negligible, irrespective of
how the other $m_k$ grow. When the $m_k$ are, in addition, of the same asymptotic order,
$\sum_k m_k = \Theta(Km)$, so that the normalizing sequence simplifies to $Km$ as in the
homogeneous case.  We refer to the Appendix~\ref{SecNbEdges} in the \textcolor{blue}{SM} for
details.

Finally, we emphasize that although the theory is asymptotic, our numerical experience suggests that the finite-sample accuracy of the Laplace approximation is remarkably good. In particular, the simulations in Section~\ref{Sec: MC_gllvm} show that the approximation bias can already be virtually negligible, for all parameters, in networks with as few as $n_V=10$ nodes.

\textit{Some technical aspects.} The proof of Theorem~\ref{thm:model1_double} requires several intermediate results contained in Lemma~\ref{lem:hatQr-order}--\ref{lem:delta-inf}, available in Appendix~\ref{app:tech-det} of 
the \textcolor{blue}{SM}; the full argument is given in Appendix~\ref{proof-thm1} of 
the \textcolor{blue}{SM}, and we highlight its key steps here. Following \cite{SMcC95}, \citet{Bianconcini2014}, and \citet{ogdenAsymptoticValidityNaive2017a,ogdenErrorLaplaceApproximations2021}, we express the Laplace approximation error using the series expansion as 
in Equation~Equation~\eqref{eq:eps-error},
which depends on the terms $\hat{Q}_2^{-1} (\tl)$ and $\hat{Q}_r (\tl)$, from which we select scalar entries in a particular combination of indices. 
In our development, we use this expansion to determine the asymptotic order of the Laplace approximation error and its derivatives. We begin by showing (Lemmas~\ref{lem:hatQr-order}--\ref{lem:orders-hatQr-vector}) that $\epsilon^{(k)} (\bt) = \bigo{m^{-1}}$. Consequently, the error incurred by replacing the exact log-likelihood with its Laplace approximation is $O(m^{-1})$. The resulting $M$-estimator admits a von Mises expansion with stochastic error $O_p(K^{-1/2})$ and deterministic bias of order $O(m^{-1})$. This aligns with the results in \citet{V96} for mixed-effects models in longitudinal studies and already constitutes a novel contribution in our multiview network setting. 

However, our derivation goes further. As noted by \citet{ogdenAsymptoticValidityNaive2017a}, the above result on $\epsilon^{(k)}(\bt)$  controls the log-likelihood error but is of limited use for controlling the error in the log-likelihood score $\Sco$: analysing the scores requires studying $\partial_{\boldsymbol\theta} \epsilon^{(k)}(\bt)$. 
To address this, and following \citet{ogdenAsymptoticValidityNaive2017a}, we consider the error term as an asymptotic expansion, and study the error of the highest contribution.
More in detail, see Lemma~\ref{lem:delta-inf},  we prove that $\partial_{\boldsymbol\theta} \epsilon (\bt^l) = O(m^{-2})$ for dyad-specific parameters and $O(m^{-1})$ for global parameters. 

We combine all these ingredients to
prove Theorem \ref{thm:model1_double}, which yields conclusions similar to those of \citet{RVL09} and \citet{Bianconcini2014} and characterizes the order of the Laplace approximation in the estimated parameters. In spite of this theoretical similarity with the existing literature, unlike previous work,  our multiview setting distinguishes between edge-specific and global parameters, and explains how to derive the Gaussian asymptotic distribution under joint growth (at rate depending on $\varrho$)  of the network size $m$ and the number of views $K$.

\section{Monte Carlo experiments} \label{Sec: MC_gllvm}

We assess the finite-sample performance of the GLAMLE via Monte Carlo (MC) experiments. Section~\ref{PZ_MC} evaluates the Poisson and ZAGA versions of GLAMLE. Section~\ref{Sec: GLPP_comp} compares GLAMLE with the Poisson PMLE (PPMLE) with fixed effects, the natural off-the-shelf competitor for gravity models.

\subsection{GLAMLE for Poisson and ZAGA models} \label{PZ_MC}

We  study the performance of the GLAMLE under different distributions, with a small number of nodes ($n_V=10$, $m=90$ directed edges $i \to j$). Full design details are in Appendix~\ref{App: MCsettings}  of the \textcolor{blue}{SM}, here we provide the key details. Specifically, we consider the
\textit{Poisson case} (Models~3 and~4 with one latent variable, $q=1$,  one edge-specific covariate, $L_x=1$,  no global covariates, $L_w=0$), and
the \textit{ZAGA case} (Model~5 with one latent variable, $q=1$, one edge-specific covariate, and one global covariate, $L_x=L_w=1$). In both cases,  $K=500$. 

To illustrate the simulated random fields, Figure~\ref{fig:b_xij} shows chord diagrams for Poisson (left) and ZAGA (right) across two layers. The plots illustrate that Poisson generates dense graphs with few zeros, while ZAGA produces sparse networks with clear zero flows and heterogeneous edge weights, mimicking real trade data.

For all models, estimation is based on the Laplace-approximated log-likelihood, implemented in \texttt{R} using the \texttt{RTMB} package, a native \texttt{R} interface to TMB, connecting  C++ and \texttt{R}. The resulting output allows the log-likelihood to be maximised using the \texttt{PORT} routines  in \texttt{nlminb}. 
The estimated parameters were compared with their true values, excluding parameters fixed by the identifiability constraints, see \Autoref{Tfig:bij_xij}(a) and \Autoref{Tfig:bij_xij}(b) for Poisson Models~3  and 4, respectively, and  \Autoref{Tfig:bij-xij-wij-o} 
for ZAGA Model~5. 
The plots illustrate that all GLAMLE estimates 
are essentially unbiased, so the Laplace approximation to the true density works remarkably well even when $n_V=10$. 
For Poisson Model~3,  the precision for the estimator $\boldsymbol{\hat\beta}$  is even larger (Figure \ref{Tfig:bij_xij}(a)): this is a global parameter and this numerical result is in line with Theorem \ref{thm:model1_double}, illustrating that global parameters may converge faster than edge-specific parameters. 

\begin{figure}[t]
  \centering

  \begin{minipage}[t]{0.49\textwidth}
    \vspace{0pt}
    \centering

    \begin{minipage}[t][0.15\textheight][t]{\linewidth}
      \centering
      \includegraphics[
        width=0.95\linewidth,
        height=0.35\textheight,
        keepaspectratio
      ]{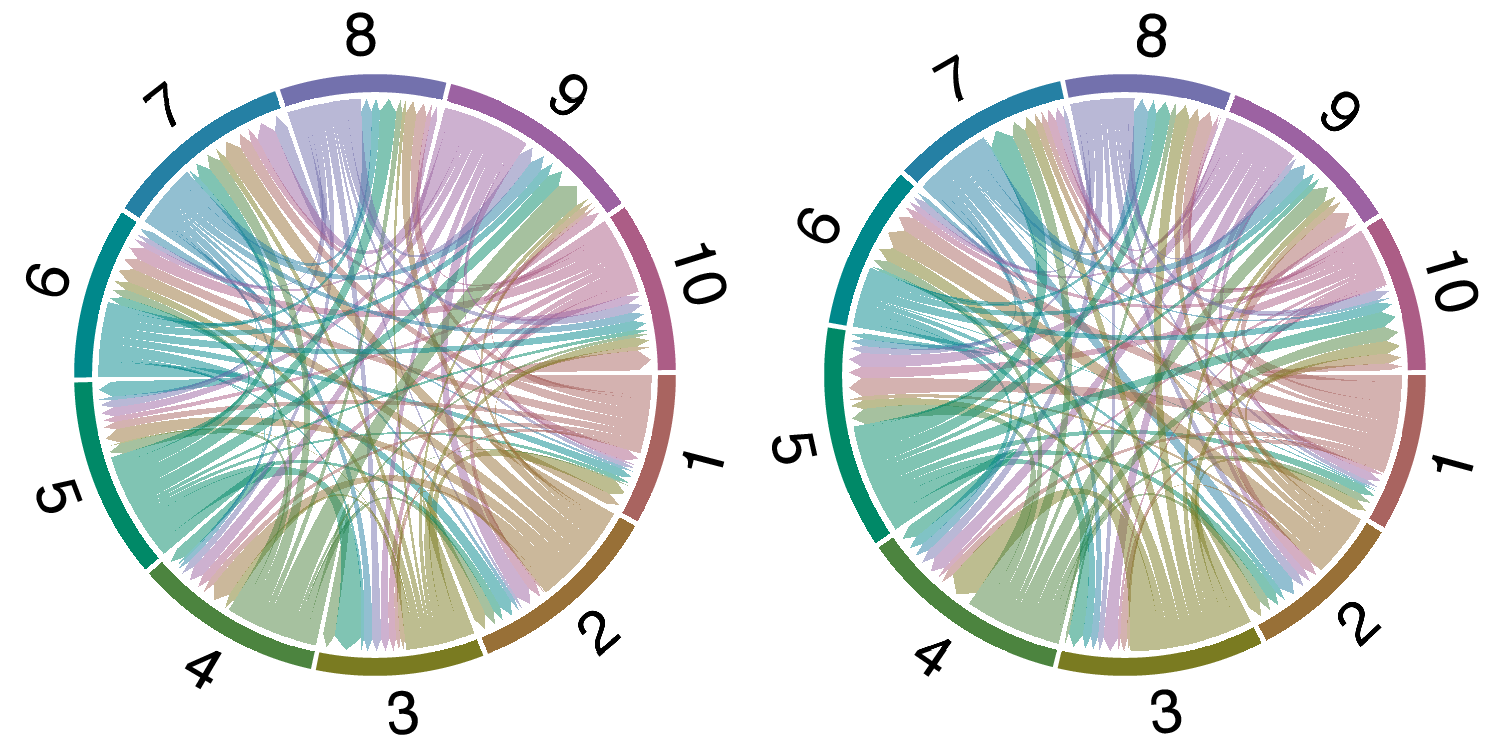}
    \end{minipage}

    \vspace{0.5em}

    {\small (a)~Poisson.}
  \end{minipage}
  \hfill
  \begin{minipage}[t]{0.49\textwidth}
    \vspace{0pt}
    \centering

    \begin{minipage}[t][0.15\textheight][t]{\linewidth}
      \centering
      \includegraphics[
        width=0.95\linewidth,
        height=0.45\textheight,
        keepaspectratio
      ]{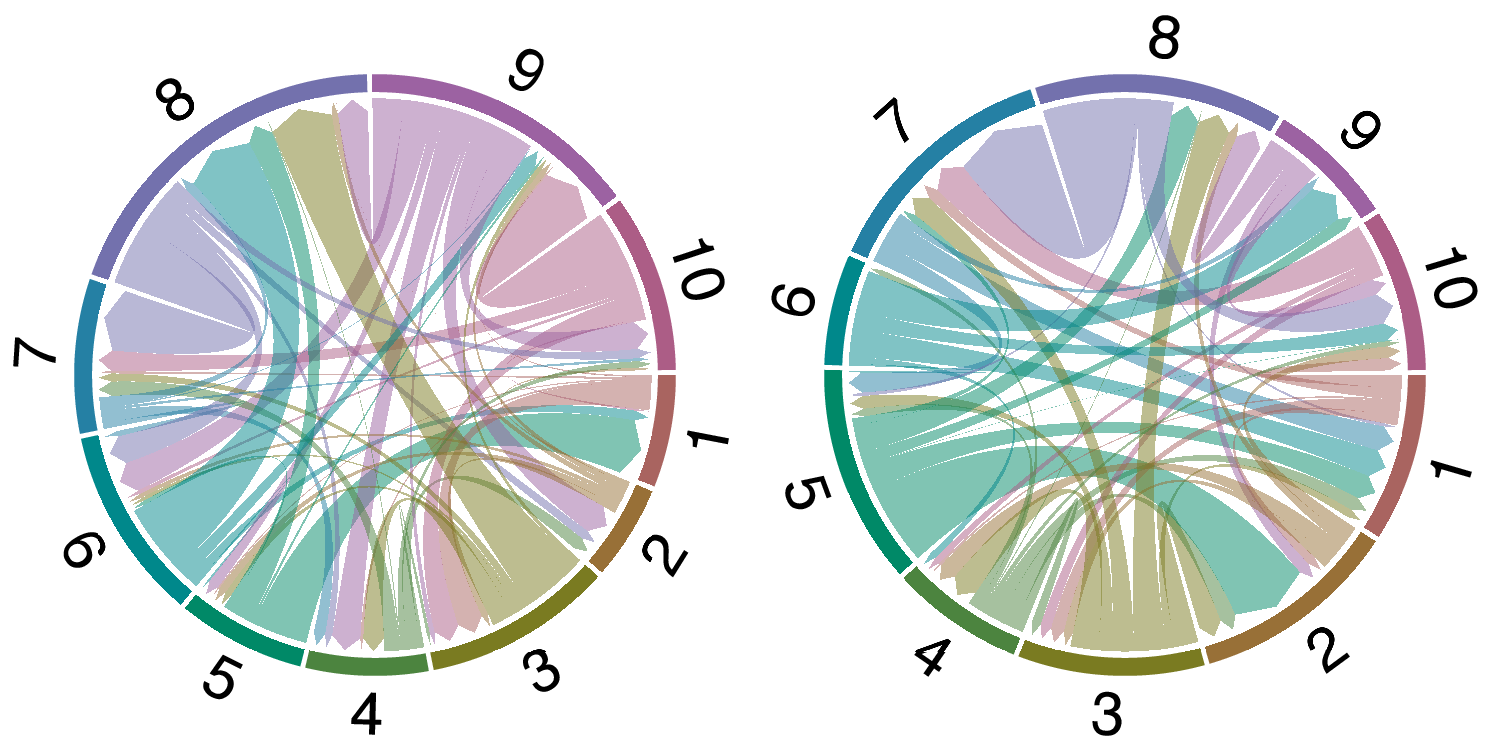}
    \end{minipage}

    \vspace{0.5em}

    {\small (b)~ZAGA.}
  \end{minipage}

  \caption{Network representation via chord plots of Poisson Model 2 
  (left) and ZAGA Model 3 (right) simulated data: connections and counts volume per edge for two layers.}
  \label{fig:b_xij}
\end{figure}

\begin{figure}[t]
  \centering

  \begin{minipage}[t]{0.49\textwidth}
    \vspace{0pt}
    \centering

    \begin{minipage}[t][0.15\textheight][t]{\linewidth}
      \centering
      \includegraphics[
        width=0.95\linewidth,
        height=0.35\textheight,
        keepaspectratio
      ]{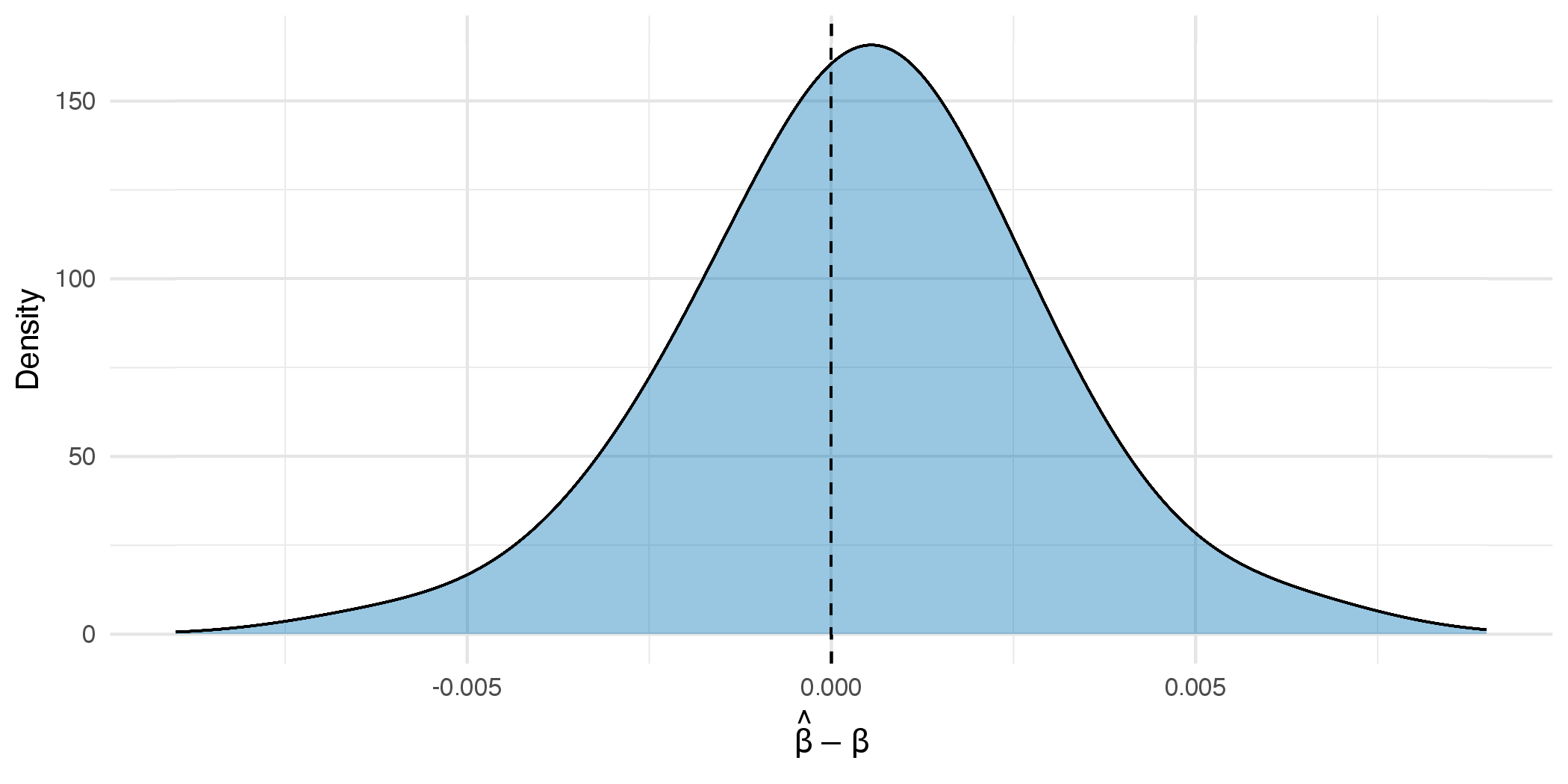}
    \end{minipage}

    \vspace{0.5em}

    {\small \begin{flushleft}
      (a)~Density plot of ($\hat{\beta} - \beta$) for Poisson Model~3.
    \end{flushleft}}
  \end{minipage}
  \hfill
  \begin{minipage}[t]{0.49\textwidth}
    \vspace{0pt}
    \centering

    \begin{minipage}[t][0.15\textheight][t]{\linewidth}
      \centering
      \includegraphics[
        width=0.95\linewidth,
        height=0.45\textheight,
        keepaspectratio
      ]{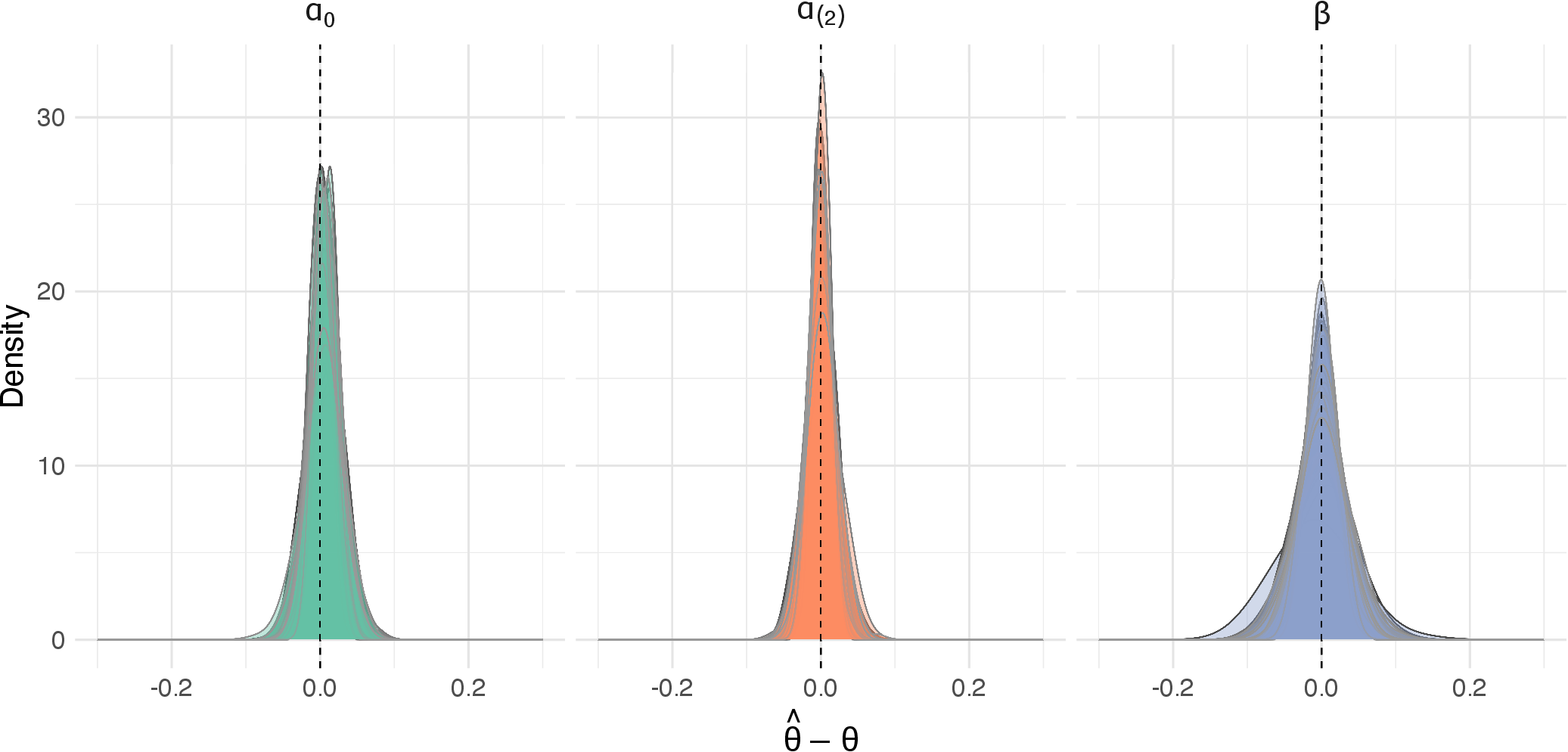}
    \end{minipage}

    \vspace{0.5em}

    {\small  \begin{flushleft}
      (b)~Densities of (estimates - true values) for Poisson Model~4. Densities for all edges are superimposed.
    \end{flushleft}}
  \end{minipage}

  \caption{Poisson Models 3 and 4, 
  with $L=1, K = 500, m = 90$, and 100 MC runs.}
  \label{Tfig:bij_xij}
\end{figure}

\begin{figure}[h!]
  \centering
  \begin{minipage}{0.45\textwidth}
    \centering
    \includegraphics[width=1.1\textwidth, height=0.23\textheight]{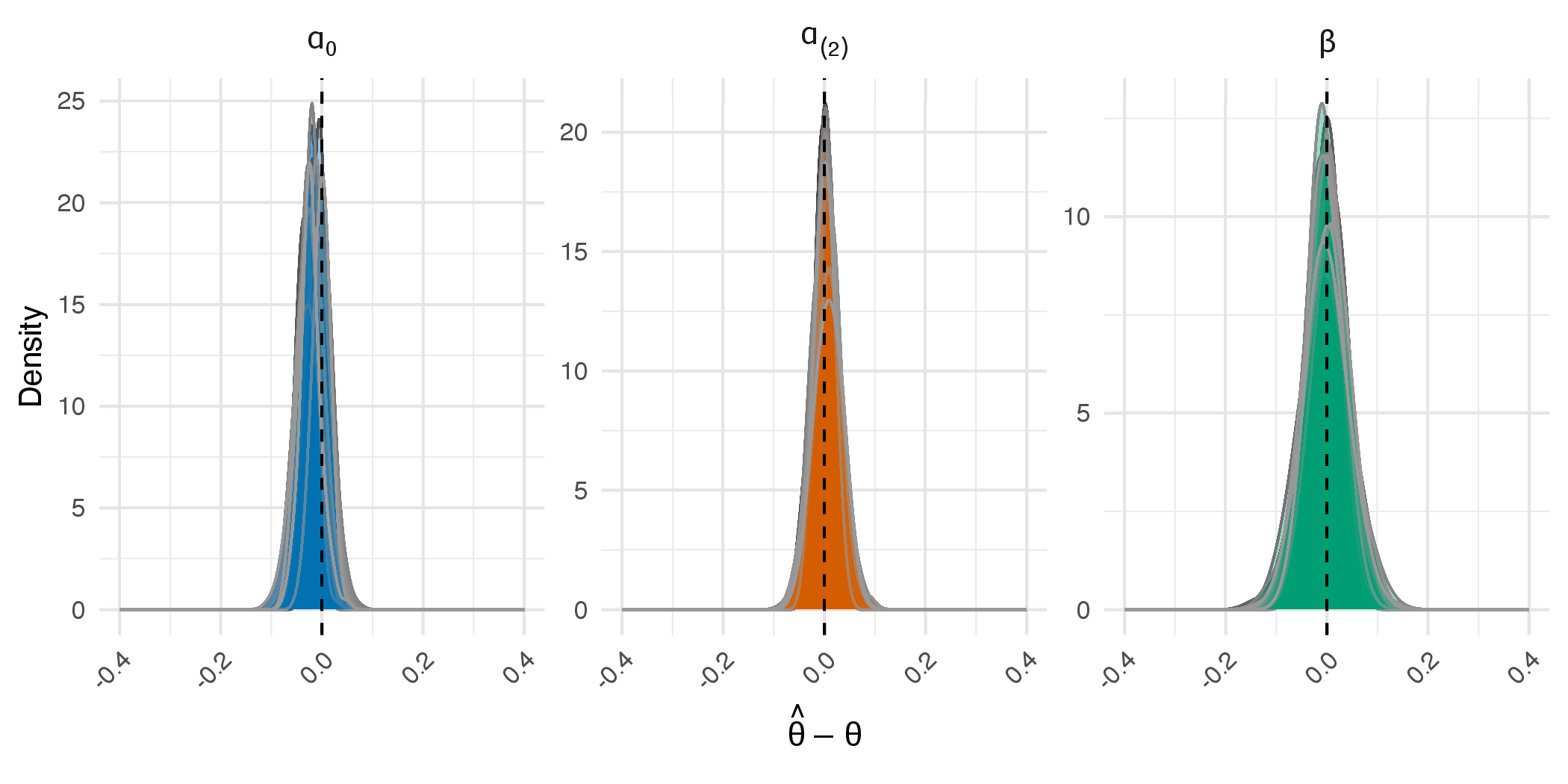}
    \label{Tfig:bij-xij-wij-densities-o}
  \end{minipage}
  \quad
  \begin{minipage}{0.45\textwidth}
    \centering
    \includegraphics[width=1.1\textwidth, height=0.23\textheight]{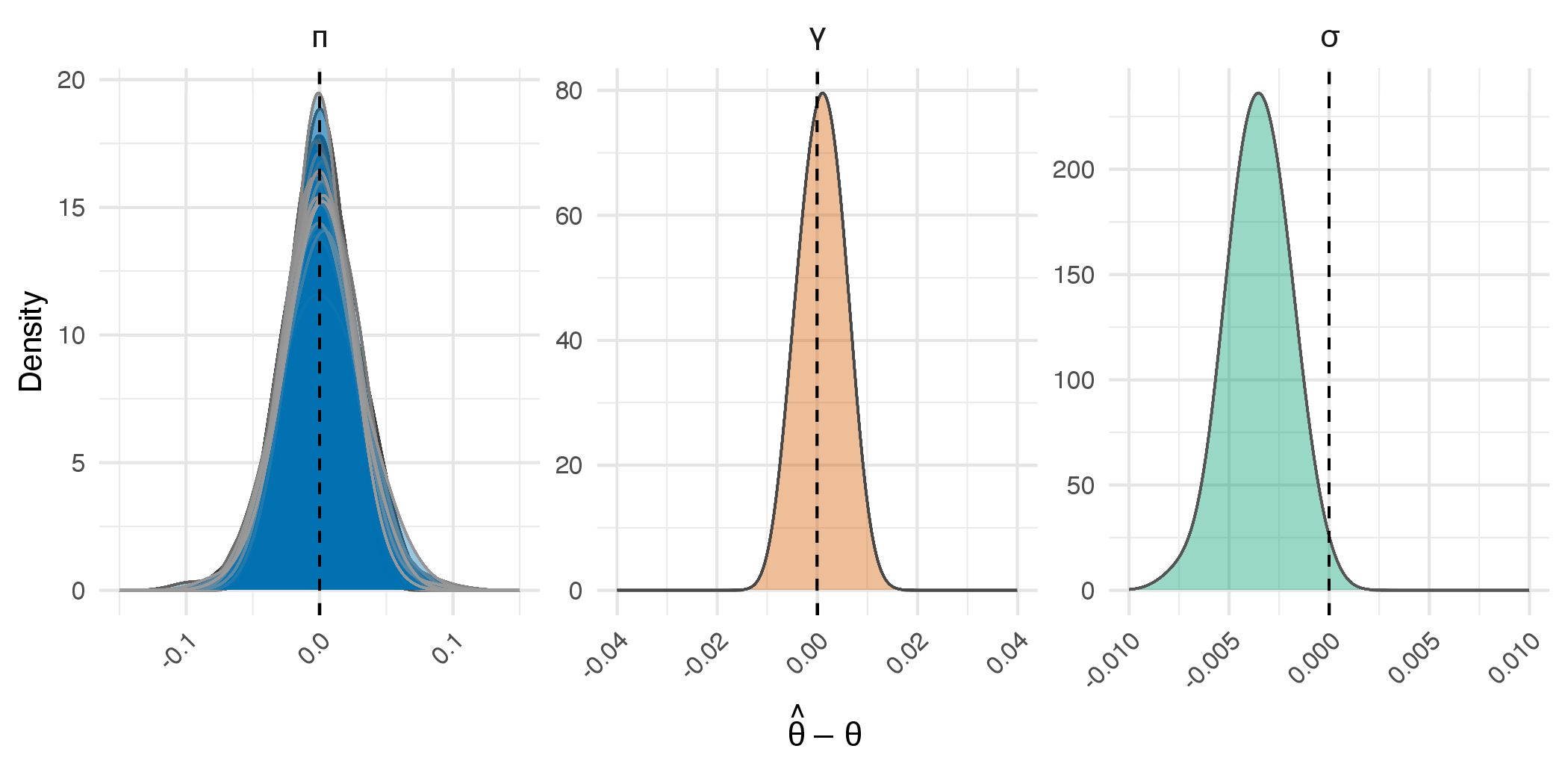}
    \label{Tfig:bij-xij-wij-rmse}
  \end{minipage}
 \caption{Densities of (estimates - true values) for ZAGA Model 5, with $L_x=L_w = q = 1, K = 500, m = 90$, and $100$ MC runs.  Densities for all edges are superimposed. }
 \label{Tfig:bij-xij-wij-o} 
\end{figure}

\subsection{Comparison between GLAMLE  and PPMLE approach}
\label{Sec: GLPP_comp}

To evaluate our  approach under conditions that mirror the stylized
facts of trades data, we simulate from a ZAGA model with latent variables,
producing flows that simultaneously exhibit excess zeros and right-skewed
positive values. 

We compare the resulting ZAGA-GLAMLE fit against the
routinely applied PPMLE fit, as
implemented in the \texttt{gravity} package: 
 $ \Yijk \sim \mathrm{Poi}(\mijk),$ and
 $ \ln \mijk
  =
  \beta_{0}
  +
  \bb^\top \bxx_{ij}^{(k)}
  +
  \bg^\top \bw_{ij}
  +
  \nu_{i}  
  +
  \nu_{j}.$
This benchmark includes observed covariates and exporter/importer fixed
effects ($\nu_{i}$ and $\nu_j$, respectively) for fairer comparison with the latent-variable GLAMLE, but neither a latent factor structure for each edge, nor a separate zero-inflation
component. 

For each MC run $s = 1, \ldots, S=100$, we fit both models and
extract their respective fitted quantities. Under PPMLE approach, the fitted mean is 
$\hat{\mu}^{\mathrm{PPMLE},s}_{ij,k} = \hat{\mathbb{E}}_{\mathrm{PPMLE}}
(\Yijk \mid \bxx_{ij}^{(k)}, \bw_{ij}, \nu_i, \nu_j)$.
Under ZAGA-GLAMLE, the fit yields a mean for the Gamma (modelling positive trades)
$\hat{\mu}^{\mathrm{GLAMLE},s}_{ij,k}$ and a zero-inflation probability
$\hat{\pi}^{\mathrm{GLAMLE},s}_{ij}$, so that the fitted mean, conditional on
the latent and observable variables, is
$\hat{\mathbb{E}}_{\mathrm{GLAMLE}}(\Yijk \mid \bxx_{ij}^{(k)}, \bw_{ij}, \hat{\bz}_{(2),s}^{(k)})
= (1 - \hat{\pi}^{\mathrm{GLAMLE},s}_{ij})\,\hat{\mu}^{\mathrm{GLAMLE},s}_{ij,k}.$

The corresponding true conditional mean 
$(1 - \pi^0_{ij})\,\mu^0_{ij,k}$ is used as benchmark---with obvious notation,
 $\mu^0_{ij,k}$ and $\pi^0_{ij}$ denote the true Gamma mean and
zero-inflation probability from the data-generating mechanism. We assess the two models along three complementary dimensions: the excess of zeros (dimension 1), the modelling of positive trades (dimension 2), and the overall fitting of volume of trades (dimension 3). Below we comment on the outcomes.

\textit{Dimension 1.} For the ZAGA-GLAMLE approach, in each $s$-th MC run, the fitted zero probability is the
zero-inflation component:
$\hat{p}^{\mathrm{GLAMLE},s}_{0,ij,k} = \hat{\pi}^{\mathrm{GLAMLE},s}_{ij}$.
For PPMLE approach, which has no such component, we use the Poisson-implied zero
probability,
$\hat{p}^{\mathrm{PPML},s}_{0,ij,k} = \exp\{-\hat{\mu}^{\mathrm{PPMLE},s}_{ij,k}\}$.
To provide a visual representation of the performance yielded by the two approaches, 
the fitted probability of zeros is plotted against the population counterpart: a well-calibrated model should track the
$45^\circ$ line. As \autoref{fig:zaga-vs-ppml-mc}(a) shows for one layer ($k=1$), our ZAGA-GLAMLE approach
reproduces the zero mechanism closely (all points line up), 
while PPMLE approach fails to identify it (many points are off from the $45^\circ$ line). Other layers yield similar plots. 

\textit{Dimension 2.} Restricting to strictly positive
observations, 
we compare fitted positive means with
observed positive responses. For the ZAGA-GLAMLE approach the relevant quantity is
$\widehat{\mu}^{\mathrm{GLAMLE},s}_{ij,k}$, since the model separates the
zero probability from the magnitude of positive trade. For the PPMLE approach, we have 
$\hat{\mathbb{E}}_{\mathrm{PPMLE}}(\Yijk \mid \Yijk > 0,
\widehat{\mu}^{\mathrm{PPMLE}}_{ij,k})
= \widehat{\mu}^{\mathrm{PPMLE}}_{ij,k} /
(1 - \exp(-\widehat{\mu}^{\mathrm{PPMLE}}_{ij,k})).$
Using the same representation as for Dimension 1, we select one layer (still, $k=1$): \autoref{fig:zaga-vs-ppml-mc}(b) shows that
ZAGA-GLAMLE tracks the positive part of the distribution almost exactly,
whereas the PPMLE approach underestimates the mean of positive trades. 

 \textit{Dimension 3.} Using all observations,
including zeros, we compare the fitted conditional mean with the observed
response. Specifically,  in each $k$-th layer we compute the mean absolute error (MAE). So, for PPMLE, we have
$
\mathrm{MAE}^{\mathrm{PPMLE}}_{k}
=
\sum_{(i,j)\in\mathcal{I}_{k}}
\left|
\Yijk
-
\widehat{\mu}^{\mathrm{PPMLE}}_{ij,k}
\right|/{n_k}$, where $n_k$ 
is the cardinality of the set of dyads retained for each layer $k$. 
We proceed similarly for the ZAGA-GLAMLE.

\autoref{fig:zaga-vs-ppml-mc}(c) confirms the pattern observed for the $k$-th layer with $k=1$: ZAGA-GLAMLE yields an MAE which is smaller
than the one of PPMLE, whose fit is not accurate, especially
at higher mean values---this highlights the inability of the routinely applied
PPMLE approach to capture the observed skewness of trade data.

\begin{figure}[t]
  \centering

  \begin{minipage}[t]{0.31\textwidth}
    \vspace{0pt}
    \centering

    \begin{minipage}[t][0.23\textheight][t]{\linewidth}
      \centering
      \includegraphics[
        width=\linewidth,
        height=0.23\textheight,
        keepaspectratio
      ]{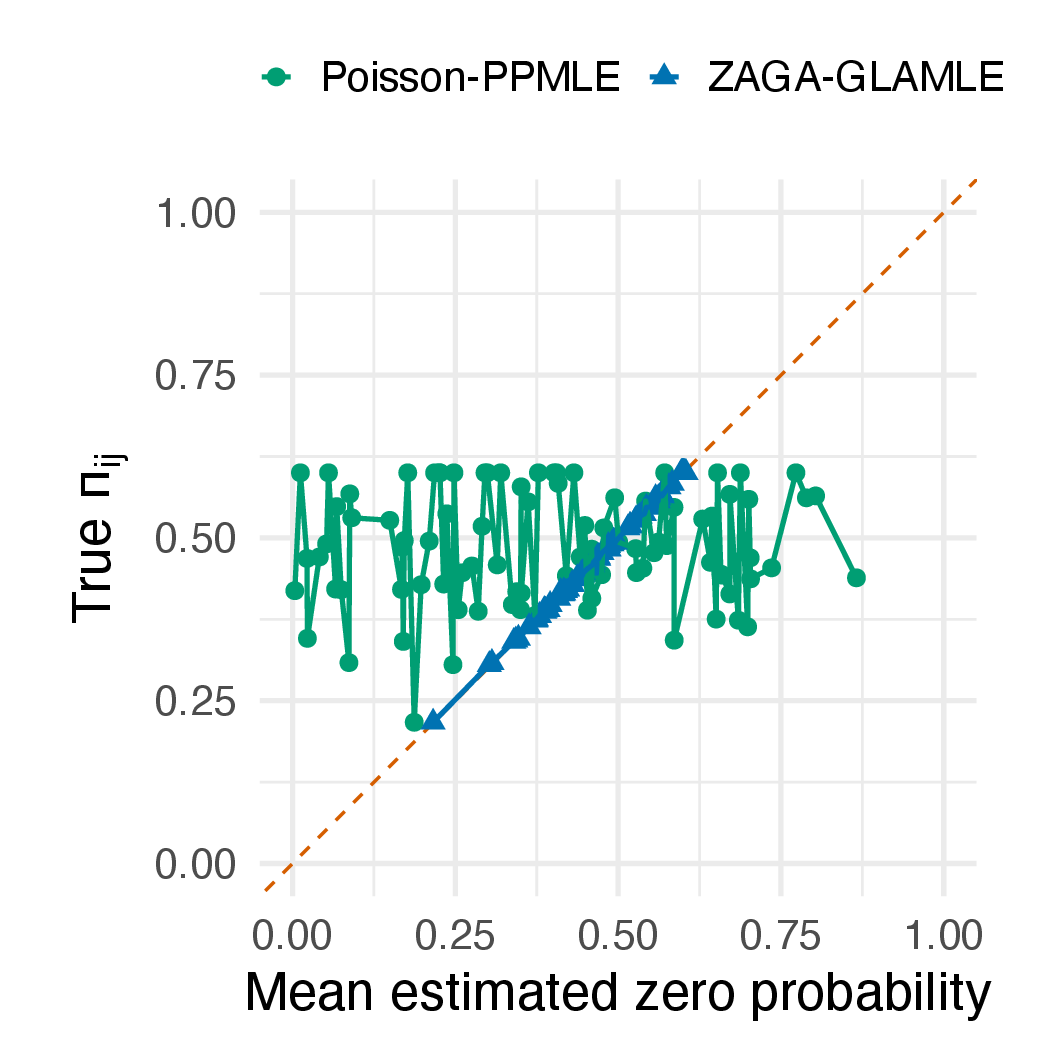}
    \end{minipage}

    \vspace{0.5em}

    {\small (a)~Dimension 1.}
  \end{minipage}
  \hfill
  \begin{minipage}[t]{0.31\textwidth}
    \vspace{0pt}
    \centering

    \begin{minipage}[t][0.23\textheight][t]{\linewidth}
      \centering
      \includegraphics[
        width=\linewidth,
        height=0.23\textheight,
        keepaspectratio
      ]{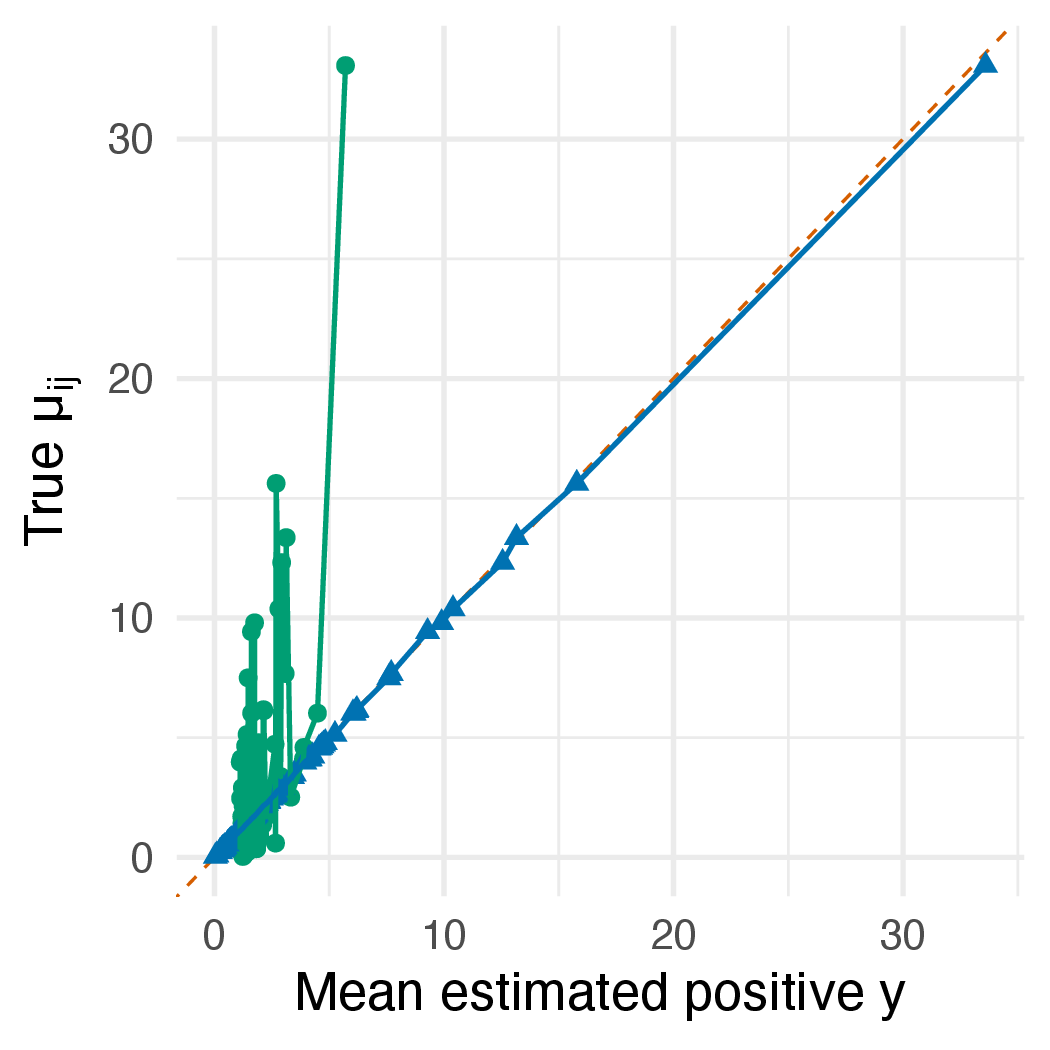}
    \end{minipage}

    \vspace{0.5em}

    {\small (b)~Dimension 2.}
  \end{minipage}
  \hfill
  \begin{minipage}[t]{0.31\textwidth}
    \vspace{0pt}
    \centering

    \begin{minipage}[t][0.23\textheight][t]{\linewidth}
      \centering
      \includegraphics[
        width=\linewidth,
        height=0.23\textheight,
        keepaspectratio
      ]{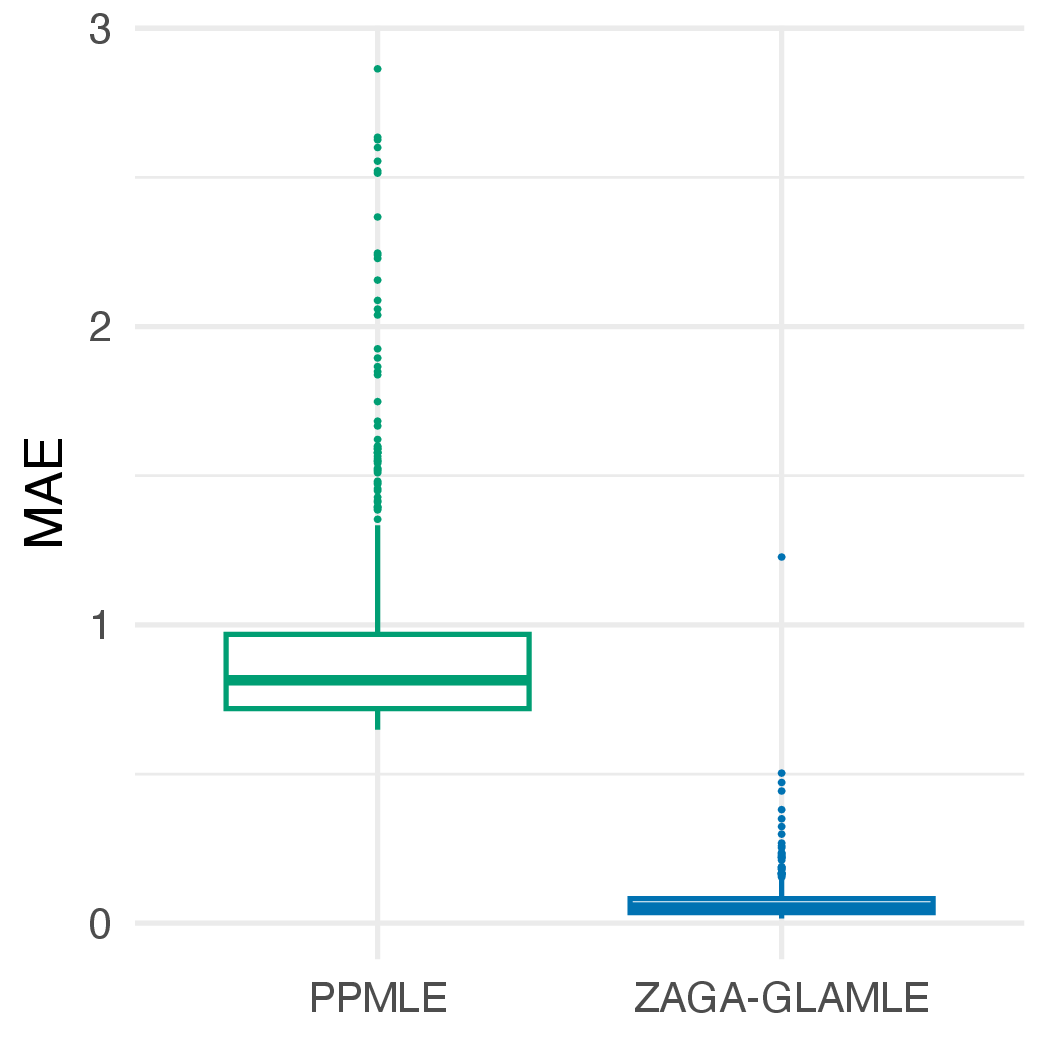}
    \end{minipage}

    \vspace{0.5em}

    {\small (c)~Dimension 3.}
  \end{minipage}

  \caption{Monte Carlo results comparing ZAGA-GLAMLE and PPMLE approaches for all dimensions: $L_x=L_w=q=1$, $K=500$, $m=90$, with $100$ repetitions. Panels (a) and (b) are for layer $k=1$, whereas panel (c) is across all layers.}
  \label{fig:zaga-vs-ppml-mc}
\end{figure}

\section{WTO data analysis} \label{Sec: RealData} 
\textit{Overview.} We apply the ZAGA-GLAMLE approach to
the WTO data introduced in Section~\ref{Sec:Motiv} with the purpose to illustrate the key features of the proposed methodology. In particular, we provide empirical evidence of the flexibility and ability of the ZAGA-GLAMLE framework, which can  jointly model the observed zero inflation, positive trade magnitudes, and latent heterogeneity, thereby capturing key features of the data that are not accommodated by standard approaches such as PPMLE with fixed effects. The empirical patterns that emerge are consistent with our theoretical findings and with the MC evidence, where distinct components of the model contribute separate sources of variation. 

\textit{Data description and modelling.}  For comparison with the gravity model
applied in empirical economics (see e.g. \cite{Yotov2016}), we augmented the initial WTO dataset with additional covariates from the \texttt{cepiigeodist} package in \texttt{R} and from the WTO Data Portal. We centred and scaled the continuous variables. We mention the layer independent variables \texttt{dist} (distance in Km between the most populated cities of two countries), \texttt{\texttt{comlang\_ethno}} (whether the two countries have at least $9\%$ of their population speaking the same language), \texttt{colony} (whether the exporter country was ever a colony of the importer country), \texttt{contig} (whether the two countries share a border) and the layer dependent variable \texttt{best} (bilateral best applied simple average tariff for 2022), that we will discuss more deeply below, and refer to Appendices~\ref{tradesdata} and \ref{tradesdata-estim} of the \textcolor{blue}{SM} for additional details.

For each directed edge $i \to j$ and product layer $k$, the ZAGA model is
characterized by two quantities: $\pij$, the edge-specific probability of an
excess zero, and $\mu_{ij}^{(k)} = \exp(\eta_{ij}^{(k)})$, the conditional mean of positive trade. We use covariates to model both quantities. 
We considered $q=1$ and $q=2$ latent variables in the specification of $\eta_{ij}^{(k)}$ in Model~6: as both
yielded very similar results, we retain the more parsimonious $q=1$
specification, with $q=2$ results available, for completeness, in Appendix~\ref{tradesdata-estim} of 
the \textcolor{blue}{SM}. 

All the $m=1980$ directed country pairs are used for the estimation of  $\pij$, whereas only 1416 pairs (after removal of $47$ pairs that exhibit no trade in any of the $K=72$ product layers, and $517$ pairs trading fewer than $30$ goods) contribute to the estimations of $\mu_{ij}^{(k)}$.

\textit{Starting values.} 
A non-trivial implementation aspect is  the selection of starting values. 
The detailed procedure is available in Appendix~\ref{Sec-implem} of the \textcolor{blue}{SM}, here we briefly describe its two key steps. First, for each edge $i\to j$, we fit a ZAGA model using the \texttt{gamlss} package in \texttt{R} (no latent variables): $\ln \mu_{ij} = \eta_{ij} = s_{ij} + \bb_{ij}^{\top}\bxx_{ij}$. This yields initial estimates $\hat{\eta}_{ij}$, $\hat{s}_{ij}$, $\hat{\bb}_{ij}$, and $\hat{\sigma}$. The edge-specific intercepts $\hat{s}_{ij}$ are then decomposed into a component orthogonal to the design data matrix $\bW$ and a component lying in its span; see Section \ref{Sec: uniq}. Using the projection matrices $\mathbf{P}_w = \bW(\bW^{\top}\bW)^{-1}\bW^{\top}$ (onto the column space) and $\mathbf{M}_w = \mathbf{I}_m - \mathbf{P}_w$ (onto the orthogonal complement), we obtain the unique decomposition $\mathbf{s} = \baa_0 + \bW\boldsymbol{\gamma}$ with $\baa_0 = \mathbf{M}_w \mathbf{s}$ and $\bW\boldsymbol{\gamma} = \mathbf{P}_w \mathbf{s}$. The global covariate coefficients are then $\hat{\boldsymbol{\gamma}} = (\bW^{\top}\bW)^{-1}\bW^{\top}\hat{\mathbf{s}}$, which can be efficiently computed via a QR decomposition as $\hat{\boldsymbol{\gamma}} = \mathbf{R}_1^{-1}\mathbf{Q}_1^{\top}\hat{\mathbf{s}}$. Finally, starting values for the factor loadings $\baa$ and the latent variables $\mathbf{z}$ are obtained by performing factor analysis on the Dunn--Smyth residuals from this initial fit. 
The resulting estimates are then used to set up the starting values for the likelihood optimization that defines the GLAMLE.

\textit{Estimation results.} Table~\ref{tab:wto-glamle-estimates-colony-contig} reports some of the global parameter estimates with sandwich standard errors for both parts of the model. These regression coefficients act on the positive Gamma mean $\hat{\mu}_{ij}^k$, capturing the effect of covariates on the intensity of positive trade flows. However, the fitted mean of the observed response depends jointly on $\hat{\mu}_{ij}^{(k)}$ and the estimated zero probability $\hat{\pi}_{ij}$. A dyad may therefore have a large fitted positive-trade mean but a small unconditional mean when the estimated zero probability is high. Thus, the ZAGA-GLAMLE approach allows us to disentangle, for each covariate, its contribution to the probability of trade occurrence and the magnitude of trade given occurrence.

We see that \texttt{best} increases the probability of zero trade (see \autoref{fig:wto-best-dist-pi}(a) for a visual representation), but its coefficient is only significant on the zero part\footnote{This covariate is of particular interest to WTO statistical officers, as it highlights the role of tariff in the commodity trade network.}. Hence, it decreases the probability the trade happens, but if it does, it has no influence on the magnitude of the value traded.  
Overall, these patterns suggest that once trade is established, changes in tariffs do not translate immediately into reduced trade volumes: the intensity of positive trade flows appears governed by historical, institutional, and long-term (e.g. related to production plans) factors. 

\begin{footnotesize}

\begin{table}[t]
    \centering
    \caption{GLAMLE global parameter estimates for the zero and positive parts for the WTO trade data, 2022.  Continuous covariates are centred and scaled, their estimates are transformed back and presented on the original scale. All values reported are on the $10^3$ scale except $\sigma$. We only present the estimates of global parameters.}
    \label{tab:wto-glamle-estimates-colony-contig}
    \begin{tabular}{l|rr|rr}
     \toprule 
  & \multicolumn{2}{c|}{Zero part} & \multicolumn{2}{c}{Positive part} \\
 \midrule 
  Parameter & Estimate & SE  & Estimate & SE \\
 \midrule 
     Intercept &-2553.29&106.703 &  &   \\
     \texttt{best} &7.414&0.605 & 2.099 & 2.751\\
     \texttt{dist} &0.126 &0.002 & 1.435 & 0.018\\
     \texttt{comlang\_ethno} &-291.789 &27.015 & 4815.905 & 98.851\\
     \texttt{colony} &-594.58 &76.447 & 4138.528 & 122.826\\
     \texttt{contig} &-1751.404 &68.006 & 15729.506 & 187.48\\
     $\sigma$ & & & 1.854 & 0.032 \\
    \bottomrule
    \end{tabular}
\end{table}

\end{footnotesize}

\begin{figure}[t]
  \centering

  \begin{minipage}[t]{0.49\textwidth}
    \vspace{0pt}
    \centering

    \begin{minipage}[t][0.15\textheight][t]{\linewidth}
      \centering
      \includegraphics[
        width=0.95\linewidth,
        height=0.45\textheight,
        keepaspectratio
      ]{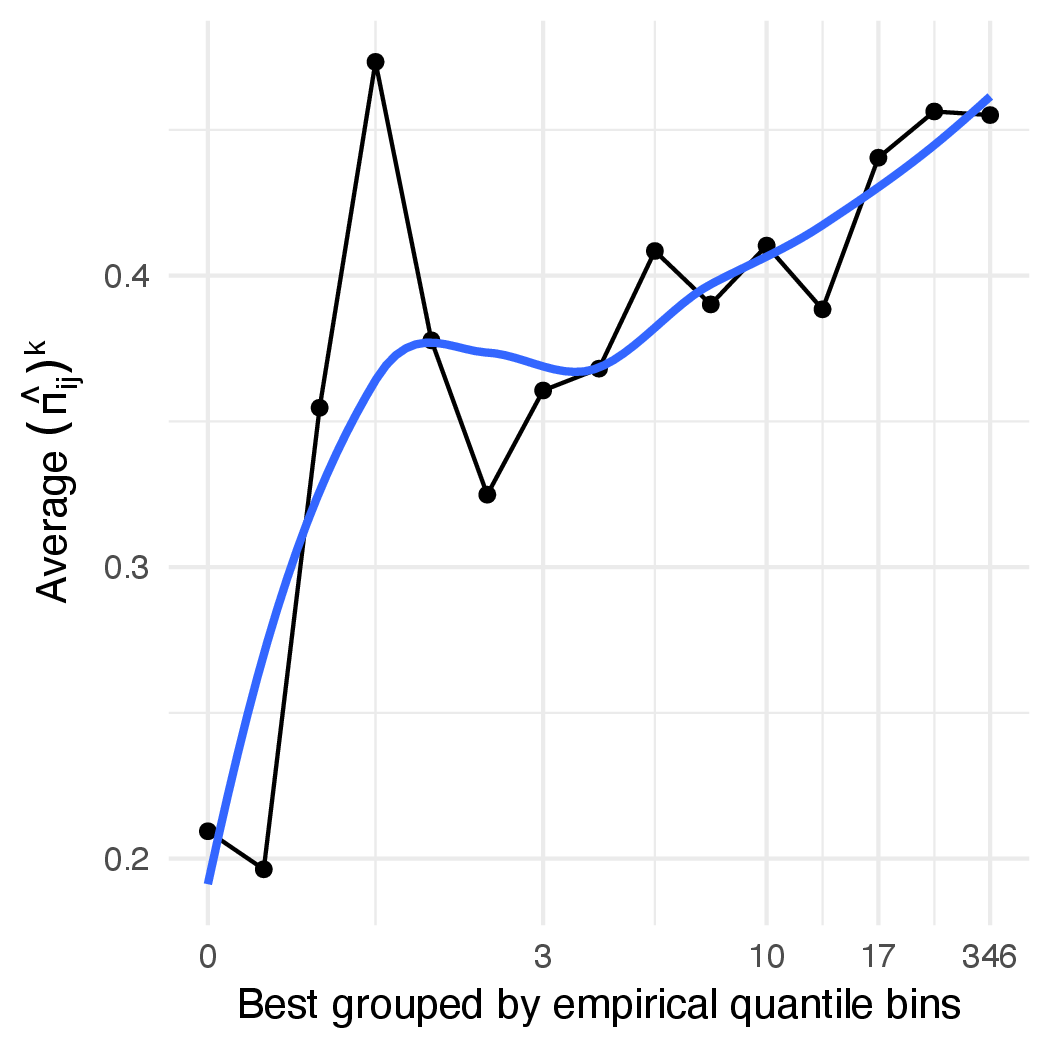}
    \end{minipage}

    \vspace{0.5em}

    {\small \begin{flushleft}
      (a)~Relationship of \texttt{best} versus estimated zero probability. Continuous curve is a locally weighted regression line (loess smoother).
    \end{flushleft}}
  \end{minipage}
  \hfill
  \begin{minipage}[t]{0.49\textwidth}
    \vspace{0pt}
    \centering

    \begin{minipage}[t][0.15\textheight][t]{\linewidth}
      \centering
      \includegraphics[
        width=0.95\linewidth,
        height=0.45\textheight,
        keepaspectratio
      ]{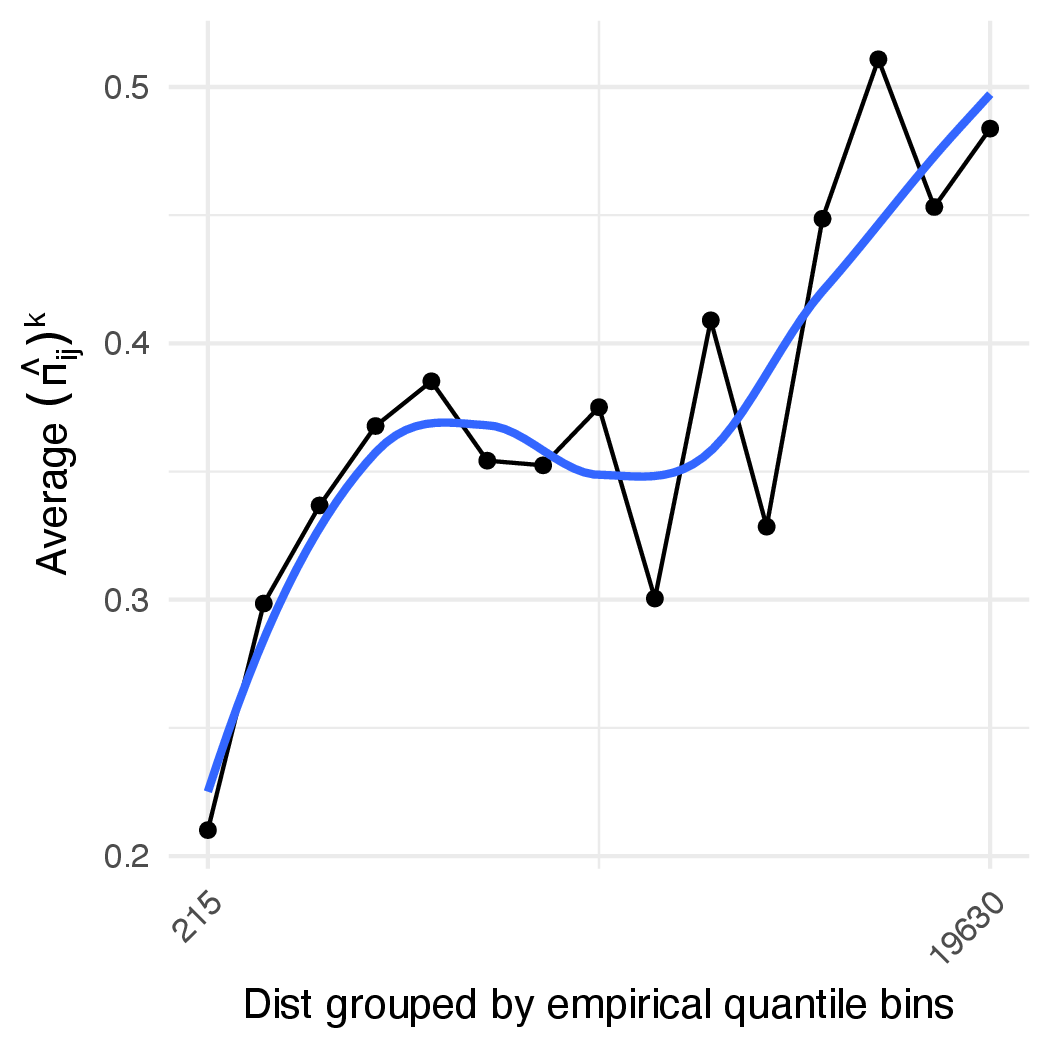}
    \end{minipage}

    \vspace{0.5em}

    {\small \begin{flushleft}
      (b)~Relationship of \texttt{dist} versus estimated zero probability. Continuous curve is a locally weighted regression line (loess smoother).
    \end{flushleft}}
  \end{minipage}
  \vspace{0.2em}
  \caption{Relationship of \texttt{best} (left) and \texttt{dist} (right) versus estimated zero probability.}
  \label{fig:wto-best-dist-pi}
\end{figure}

The positive coefficient of \texttt{dist} should not be interpreted as a standard gravity effect; rather, it reflects the selection mechanism induced by the ZAGA specification, whereby long-distance dyads tend to exhibit fewer (see \autoref{fig:wto-best-dist-pi}(b)) but larger (see \autoref{tab:wto-glamle-estimates-colony-contig}) flows. This feature underscores the importance of jointly modelling zero inflation and the magnitude of positive trade. In addition, the inclusion of latent variables allows the model to capture unobserved heterogeneity across dyads.  The significant effect of \texttt{comlang\_ethno} 
 indicates that a shared ethnological language  increases the conditional positive-trade mean. The estimated $\sigma$ confirms substantial residual dispersion among positive trade values. 

We compare the trade volumes  fitted by our ZAGA-GLAMLE and by the PPMLE approaches. 
For the PPMLE,  we consider the model 
  $\ln \mijk
  =
  \beta_{0}
  +
  \beta_{1} \mathrm{best}_{ij}^{(k)} 
  +
  \beta_{2}\ln \mathrm{dist}_{ij} +
  \bg^\top \boldsymbol{\omega}_{ij}
  +
  \nu_{i}
  +
  \nu_{j},$
with $\boldsymbol{\omega}$ containing the  covariates  $\texttt{comlang\_ethno}$, $\texttt{colony}$, and $\texttt{contig}$, whereas $ \nu_{i}, \nu_{j}$ are sender and receiver fixed effects.

We compute the fitted mean for the PPMLE approach and the fitted Gamma mean for the GLAMLE approach,  
as well as the fitted zero-inflation probability. 
In Figure \ref{fig:wto-glamle-vs-ppml-cottoncars}, we already gave a preview of the superior performance yielded by our method. Here, we complete the picture reporting that the PPMLE method entails a median MAE  which is about 2 times higher than the one of the ZAGA-GLAMLE.  

Finally, we report that our routine (most of it in C++ via TMB) for the ZAGA-GLAMLE takes on average about 0.52 seconds per parameter, with about 0.03 seconds for the estimate and the rest for the computation of the Hessian matrix in \texttt{R}.\footnote{This computation time could be reduced by coding the Hessian in C++.}. By comparison, the PPMLE takes 0.06 seconds for each parameter estimate and its standard error. 

\section{Conclusion}
We propose a novel Laplace-based inference framework for multiview network data and establish its asymptotic theory. Our methodology accommodates nonlinear latent variable models beyond the exponential family and allows for the simultaneous inclusion of local and global covariates. Extensive Monte Carlo experiments and the WTO trade application demonstrate the practical effectiveness of the approach.

The results presented here  open several promising research directions. A first important challenge concerns the selection of the latent dimension $q$, which remains unresolved. Preliminary investigations based on SCAD-penalized Laplace approximate likelihood have yielded encouraging results and are currently under further study. Another promising direction is the development of robust inference procedures, including both estimation and testing, based on bounded Laplace-approximated scores \citep{HR94,CR01}. Such ideas may also prove useful in the emerging literature on tensor factor models for dynamic networks \citep{chen2022factor,chang2023modelling,barigozzi2025factor}, where the use of Laplace approximation is unexplored.

From an economic perspective, the foundations of the proposed ZAGA specification for trade flows on networks also deserve further investigation. In particular, an open question is whether ZAGA-distributed equilibrium flows can be rationalized through an appropriate random utility framework and what is the related E-OT formulation. Such a framework would be analogous in spirit to \citet{GS22} and \citet{GH26}, while relying on alternative notions of entropy, such as Einstein--Bose, Fermi--Dirac, or Tsallis entropy. Establishing this connection would provide an economic micro-foundation for the model and further bridge the gap between modern network methods and economic equilibrium theory.

\section{Acknowledgments}
We thank Edvinas Drevinskas, Statistical Officer at WTO in Geneva, for providing  the data and for his explanation about variables (especially the need to study \texttt{best}) and data collection process at WTO. Davide La Vecchia thanks the Swiss National Science Foundation grant number CR00-5L-239816 for the financial support.

\bibliographystyle{abbrvnat}
\bibliography{reference}

@article{WGRD26,
	author = {Whiteley, Nick and Gray, Annie and Rubin-Delanchy, Patrick},
	journal = {Journal of the Royal Statistical Society Series B: Statistical Methodology},
	number = {2},
	pages = {353--385},
	publisher = {Oxford University Press UK},
	title = {Statistical exploration of the manifold hypothesis},
	volume = {88},
	year = {2026}}

@article{CR01,
	author = {Cantoni, Eva and Ronchetti, Elvezio},
	journal = {Journal of the American Statistical Association},
	number = {455},
	pages = {1022--1030},
	publisher = {Taylor \& Francis},
	title = {Robust inference for generalized linear models},
	volume = {96},
	year = {2001}}

@book{S10,
	author = {Small, Christopher G},
	publisher = {Chapman and Hall/CRC},
	title = {Expansions and asymptotics for statistics},
	year = {2010}}

@article{GS22,
	author = {Galichon, Alfred and Salani{\'e}, Bernard},
	journal = {The Review of Economic Studies},
	number = {5},
	pages = {2600--2629},
	publisher = {Oxford University Press},
	title = {Cupid's invisible hand: Social surplus and identification in matching models},
	volume = {89},
	year = {2022}}

@book{R_etal_19,
	author = {Rigby, Robert A and Stasinopoulos, Mikis D and Heller, Gillian Z and De Bastiani, Fernanda},
	publisher = {Chapman and Hall/CRC},
	title = {Distributions for modeling location, scale, and shape: Using GAMLSS in R},
	year = {2019}}

@article{rue2009approximate,
	author = {Rue, H{\aa}vard and Martino, Sara and Chopin, Nicolas},
	journal = {Journal of the Royal Statistical Society Series B: Statistical Methodology},
	number = {2},
	pages = {319--392},
	publisher = {Oxford University Press},
	title = {Approximate Bayesian inference for latent Gaussian models by using integrated nested Laplace approximations},
	volume = {71},
	year = {2009}}

@article{GH26,
	author = {Galichon, Alfred and Henry, Marc},
	journal = {arXiv preprint arXiv:2604.04227},
	title = {An econometrician's guide to optimal transport},
	year = {2026}}

@article{barigozzi2025factor,
	author = {Barigozzi, Matteo and Cavaliere, Giuseppe and Moramarco, Graziano},
	journal = {Journal of Business \& Economic Statistics},
	number = {4},
	pages = {1105--1118},
	publisher = {Taylor \& Francis},
	title = {Factor network autoregressions},
	volume = {43},
	year = {2025}}

@article{Gollini2016246,
	author = {Gollini, I. and Murphy, T.B.},
	journal = {Journal of Computational and Graphical Statistics},
	number = {1},
	pages = {246-265},
	title = {Joint Modeling of Multiple Network Views},
	volume = {25},
	year = {2016}}

@article{Salter-Townshend20171217,
	author = {Salter-Townshend, M. and McCormick, T.H.},
	journal = {Annals of Applied Statistics},
	number = {3},
	pages = {1217-1244},
	title = {Latent space models for multiview network data},
	volume = {11},
	year = {2017}}

@article{silva2006log,
	author = {Santos Silva, J.M.C. and Tenreyro, Silvana},
	journal = {The Review of Economics and statistics},
	pages = {641--658},
	publisher = {JSTOR},
	title = {The log of gravity},
	year = {2006}}

@book{Yotov2016,
	author = {Yotov, V. and Piermartini, R. and Monteiro, J.A. and Larch, M.},
	publisher = {WTO iLibrary},
	title = {An advanced guide to trade policy analysis: the structural gravity model},
	year = {2016}}

@article{AvW03,
	author = {Anderson, James E and Van Wincoop, Eric},
	journal = {American economic review},
	number = {1},
	pages = {170--192},
	publisher = {American Economic Association},
	title = {Gravity with gravitas: {A} solution to the border puzzle},
	volume = {93},
	year = {2003}}

@article{chang2023modelling,
	author = {Chang, Jinyuan and He, Jing and Yang, Lin and Yao, Qiwei},
	journal = {Journal of the Royal Statistical Society Series B: Statistical Methodology},
	pages = {127--148},
	title = {Modelling matrix time series via a tensor CP-decomposition},
	volume = {85},
	year = {2023}}

@article{chen2021nonlinear,
	author = {Chen, Mingli and Fern{\'a}ndez-Val, Iv{\'a}n and Weidner, Martin},
	journal = {Journal of Econometrics},
	number = {2},
	pages = {296--324},
	publisher = {Elsevier},
	title = {Nonlinear factor models for network and panel data},
	volume = {220},
	year = {2021}}

@article{chen2022factor,
	author = {Chen, Rong and Yang, Dan and Zhang, Cun-Hui},
	journal = {Journal of the American Statistical Association},
	pages = {94--116},
	title = {Factor models for high-dimensional tensor time series},
	volume = {117},
	year = {2022}}

@book{H81,
	author = {Huber, P.J.},
	publisher = {Wiley, New York},
	series = {(2nd edition, 2009, {H}uber, P.J. and {R}onchetti, E.)},
	title = {Robust Statistics},
	year = {1981}}

@article{HR94,
	author = {Heritier, S. and Ronchetti, E.},
	journal = {Journal of the American Statistical Association},
	number = {427},
	pages = {897--904},
	publisher = {Taylor \& Francis},
	title = {Robust Bounded-influence Tests in General Parametric Models},
	volume = {89},
	year = {1994}}

@article{HRVF04,
	author = {Huber, P. and Ronchetti, E. and Victoria-Feser, M.-P.},
	journal = {Journal of the Royal Statistical Society: Series B (Statistical Methodology)},
	number = {4},
	pages = {893--908},
	publisher = {Wiley Online Library},
	title = {Estimation of Generalized Linear Latent Variable Models},
	volume = {66},
	year = {2004}}

@article{Hui22,
	author = {Kidzinski, Lukasz and Hui, Francis KC and Warton, David I and Hastie, Trevor J},
	journal = {Journal of machine learning research},
	number = {291},
	pages = {1--29},
	title = {Generalized Matrix Factorization: efficient algorithms for fitting generalized linear latent variable models to large data arrays},
	volume = {23},
	year = {2022}}

@article{JLVR24,
	author = {Jiang, Chaonan and La Vecchia, Davide and Rastelli, Riccardo},
	journal = {Econometrics and Statistics},
	publisher = {Elsevier},
	ISSN={2452-3062},
	title = {{GLAMLE}: inference for multiview network data in the presence of latent variables, with an application to commodities trading},
	volume = {in press},
	year = {2024}}

@article{Kr15TMB,
	author = {Kristensen, K. and Nielsen, A. and Berg, C.W. and Skaug, H. and Bell, B.},
	journal = {arXiv preprint arXiv:1509.00660},
	title = {{TMB}: automatic differentiation and {L}aplace approximation},
	year = {2015}}

@article{N17,
	author = {Niku, J. and Warton, D.I. and Hui, F.K.C. and Taskinen, S.},
	journal = {Journal of Agricultural, Biological, and Environmental Statistics},
	number = {4},
	pages = {498--522},
	title = {Generalized Linear Latent Variable Models for Multivariate Count and Biomass Data in Ecology},
	volume = {22},
	year = {2017}}

@article{O16,
	author = {Ovaskainen, O. and Abrego, N. and Halme, P. and Dunson, D.},
	journal = {Methods in Ecology and Evolution},
	number = {5},
	pages = {549--555},
	publisher = {Wiley Online Library},
	title = {Using latent variable models to identify large networks of species-to-species associations at different spatial scales},
	volume = {7},
	year = {2016}}

@article{RFR16,
	author = {Rastelli, R. and N. Friel and A.E. Raftery},
	journal = {Network Science},
	month = {dec},
	number = {4},
	pages = {407-432},
	publisher = {Cambridge University Press ({CUP})},
	title = {Properties of {L}atent {V}ariable {N}etwork Models},
	volume = {4},
	year = 2016}

@article{RVL09,
	author = {Rizopoulos, D. and Verbeke, G. and Lesaffre, E.},
	journal = {Journal of the Royal Statistical Society: Series B (Statistical Methodology)},
	number = {3},
	pages = {637--654},
	publisher = {Wiley Online Library},
	title = {Fully Exponential {L}aplace Approximations for the Joint Modelling of Survival and Longitudinal Data},
	volume = {71},
	year = {2009}}

@article{WAR13,
	author = {Ward, Michael D and Ahlquist, John S and Rozenas, Arturas},
	journal = {Network Science},
	number = {1},
	pages = {95--118},
	publisher = {Cambridge University Press},
	title = {Gravity's rainbow: A dynamic latent space model for the world trade network},
	volume = {1},
	year = {2013}}

@article{Bianconcini2014,
	author = {Bianconcini, Silvia},
	journal = {Bernoulli},
	pages = {1507--1531},
	title = {Asymptotic properties of adaptive maximum likelihood estimators in latent variable models},
	volume = {3},
	year = {2014}}

@article{BL12,
	author = {Jushan Bai and Kunpeng Li},
	journal = {The Annals of Statistics},
	number = {1},
	pages = {436 -- 465},
	publisher = {Institute of Mathematical Statistics},
	title = {{Statistical analysis of factor models of high dimension}},
	volume = {40},
	year = {2012}
	}

@book{vdW98,
	author = {Van der Vaart, Aad W},
	publisher = {Cambridge University Pressniversity press},
	title = {Asymptotic Statistics},
	volume = {3},
	year = {2000}}

@article{V96,
	author = {Vonesh, Edward F},
	journal = {Biometrika},
	number = {2},
	pages = {447--452},
	publisher = {Oxford University Press},
	title = {A Note on the Use of {L}aplace's Approximation for Nonlinear Mixed-effects Models},
	volume = {83},
	year = {1996}}

@article{SMcC95,
	author = {Shun, Zhenming and McCullagh, Peter},
	journal = {Journal of the Royal Statistical Society: Series B (Methodological)},
	number = {4},
	pages = {749--760},
	publisher = {Wiley Online Library},
	title = {Laplace Approximation of High Dimensional Integrals},
	volume = {57},
	year = {1995}}

@book{BKM11,
	author = {Bartholomew, David J and Knott, Martin and Moustaki, Irini},
	publisher = {John Wiley \& Sons},
	title = {Latent Variable Models and Factor Analysis: A Unified Approach},
	volume = {904},
	year = {2011}}

@article{ogdenAsymptoticValidityNaive2017a,
	author = {Ogden, H. E.},
	doi = {10.1093/biomet/asx002},
	issn = {0006-3444},
	journal = {Biometrika},
	month = mar,
	number = {1},
	pages = {153--164},
	title = {On Asymptotic Validity of Naive Inference with an Approximate Likelihood},
	urldate = {2025-12-08},
	volume = {104},
	year = 2017}

@article{ogdenErrorLaplaceApproximations2021,
	author = {Ogden, Helen},
	copyright = {\copyright{} 2021 The Author. Stat published by John Wiley \& Sons Ltd.},
	doi = {10.1002/sta4.380},
	issn = {2049-1573},
	journal = {Stat},
	langid = {english},
	number = {1},
	pages = {e380},
	title = {On the Error in {{Laplace}} Approximations of High-Dimensional Integrals},
	urldate = {2025-12-03},
	volume = {10},
	year = 2021}

@article{dasCentralLimitTheorem2021,
  title = {Central {{Limit Theorem}} in High Dimensions: {{The}} Optimal Bound on Dimension Growth Rate},
  shorttitle = {Central {{Limit Theorem}} in High Dimensions},
  author = {Das, Debraj and Lahiri, Soumendra},
  year = 2021,
  month = jul,
  journal = {Transactions of the American Mathematical Society},
  issn = {0002-9947, 1088-6850},
  doi = {10.1090/tran/8459},
  urldate = {2026-02-27},
  langid = {english}
}

\clearpage
\newpage

\section*{Supplementary Material for ``Flexible latent variable models on graphs: Laplace approximated inference for multiview network data"}

This file contains supplementary material to the main paper. Section~\ref{AppModels} introduces the details specifications of ZIP and ZAGA models. Section \ref{AppA}  contains technical material (analytical derivations, assumptions, lemmas, corollaries, and detailed proofs); Section \ref{App: MCsettings} contains details for the Monte Carlo simulations. Finally, Section \ref{AppSec: RealData} provides details and further results for WTO data analysis. For the taxonomy of the models considered in this document, we refer to Section~\ref{Sec: frame} of the main paper.

\section{Models definition} \label{AppModels}

In this Section we derive the expression of the Laplace approximated log-likelihood 
\begin{eqnarray} 
\tilde{\ell}_K(\bt) 
&=&
\sk 
\l(
  -\frac{1}{2}\ln
  \l[
    \det
    \l\{
      \Gamma
      \l(
        \bt,\bxx, \bw, \zhk
      \r)      
    \r\}
  \r] 
+
 m Q 
  \l( 
    \bt, \bxx, \bw,\zhk, \boldsymbol{Y}^{(k)}
  \r)
+
\frac{q}{2} \ln 2 \pi 
\r) \label{Laplace-approx-SM}
\end{eqnarray} 
for the ZIP and the ZAGA models. $\bt$ is the entire set of parameters that includes the parameters of the underlying distribution function and those implied by the model specification of the (linear) predictor $\eta_{ij}$ (see Models~1-6 in the main text).

\subsection{ZIP model}  \label{AppZIP} 

The ZIP pmf is 

\begin{equation}
  P_{\bfth}(Y_{ij} = y \vert \bz, \bxx, \bw) =  \left\{
    \begin{array}{ll}
      \pij + (1-\pij) \exp\{ - \lambda_{ij}\} &  \ \text{if} \ y=0\\
      (1-\pij) \exp\{ - \lambda_{ij}\} \frac{\lambda_{ij}^{y}}{y!}  & \ \text{if} \  y>0,
    \end{array}
  \right.  \label{ZIPpmfApp}
\end{equation}
with $\ln \lambda_{ij} =  \eta_{ij}$. For 
$\pij=0$, 
the ZIP corresponds to a standard Poisson model, whilst values 
$\pij\in (0,1)$ 
allow for modelling the excess of zeros. 

For each $k$-th layer, the corresponding marginal density is 
\begin{align*}
    \dens & = \int \exp \Biggl\{  \snv  \Biggl[  \ln\left\{ \pij^{(k)}+\left(1-\pij^{(k)}\right) \exp \left\{-\lambda_{ij}^{(k)}\right\} \right\} \ind{=} \\
    & + \ln \left[ \left(1-\pij^{(k)}\right) \exp \left\{-\lambda_{ij}^{(k)}\right\} \frac{(\lambda_{ij}^{(k)})^{\yijk}}{\yijk!} \right]   \ind{>} \Biggr] + \ln \left(  (2 \pi)^{-q / 2} \exp \left(-\frac{\bztkt \bztk}{2}\right) \right) \Biggr\} d \bztk \\
    & = \int \exp \Bigl\{  \snv  \Biggr[ \ln\left\{ \pij^{(k)}+\left(1-\pij^{(k)}\right) \exp \left\{-\lambda_{ij}^{(k)}\right\} \right\} \ind{=} + \bigl(  \ln \left(1-\pij^{(k)}\right) -\exp (\eijk) + \yijk \ln (\eijk) \\
    & - \ln \left(  \yijk! \right)  \bigr) \ind{>} \Biggl] + \ln \left(  (2 \pi)^{-q / 2} \exp \left(-\frac{\bztkt \bztk}{2}\right) \right) \Bigr\} d \bztk,
\end{align*}
which is equal to $\int \exp \left\{m \qf \right\} d \bztk$
with 
\begin{equation*}
\begin{aligned}
m
\qf
&=  \snv \Bigl\{  \ind{=} \ln\left( \pij^{(k)} + \left(1-\pij^{(k)} \exp \left( - \lambda_{ij}^{(k)} \right) \right) \right) \\
& + \ind{>} \left( \ln \left(1-\pij^{(k)}\right) - \exp \left( \eijk\right) + \yijk \eijk -\ln \left( \yijk ! \right) \right) -\frac{\bztkt \bztk}{2}-\frac{q}{2} \ln (2 \pi) \Bigr\}.
\end{aligned}
\end{equation*}
Moreover,
$$
\begin{aligned}
\gf
= & \snv
\left(\exp \left(\hat{\eta}^{(k)}_{i j}\right) \ind{>} -\left(\frac{\left(1-\pij^{(k)}\right) \hat\xi^{(k)}_{ij} \exp \left(\hat{\eta}^{(k)}_{i j}\right)\left(\exp \left(\hat{\eta}^{(k)}_{i j}\right)-1\right)}{\pij^{(k)}+\left(1-\pij^{(k)}\right) \hat\xi^{(k)}_{i j}}\right.\right. \\
& \left.\left.-\frac{\left(1-\pij^{(k)}\right)^2 \left(\hat{\xi}^{(k)}_{ij}\right)^2 \exp \left(2 \hat{\eta}^{(k)}_{i j}\right)}{\left(\pij^{(k)}+\left(1-\pij^{(k)}\right) \hat{\xi}^{(k)}_{i j}\right) ^2}\right) \ind{=}\right) \aat
+
\boldsymbol{I}_q,
\end{aligned}
$$
where $\hat\xi^{(k)}_{ij}=\exp \left\{-\exp \left(\eijkh\right)\right\}$ and $\zhk$ is the maximum of $\qf$.

The Laplace approximated log-likelihood~Equation~\eqref{Laplace-approx-SM} is 

\begin{equation*}
\begin{aligned}
    \tilde{\ell}_K(\bt)&=\sk\Biggl(-\frac{1}{2} \ln \left[\operatorname{det}\left\{
      \gf
      \right\}\right]\\
    & + \snv \Biggl\{ \ind{=}  \ln\left( \pij^{(k)} + \left(1-\pij^{(k)} \right) \hat{\xi}_{ij}^{(k)} \right) \\
    & + \ind{>} \left( \ln \left(1-\pij^{(k)}\right) - \exp \left( \eijkh \right) + \yijk \eijkh -\ln \left( \yijk ! \right) \right) \Biggr\}-\frac{\left(\bztkh\right)^\T \bztkh}{2}\Biggr).
\end{aligned}
\end{equation*}

\subsection{ZAGA model} \label{AppZAGA}

Let $Y_{ij} \vert \bz, \bxx, \bw \sim \text{ ZAGA}$ with $\pij$, $\mij$, and scale $\sigma$ with pmf defined by Equation Equation~\eqref{ZIGpmf}, namely 
\begin{equation}
P_{\bt}\left(Y_{ij} = y_{ij} \mid \bz, \bxx, \bw\right) =  \left\{
    \begin{array}{ll}
        \pij &  \ \text{if} \ y_{ij}=0\\
             (1-\pij) \frac{1}{\Gamma \left( 1/\sigma^2 \right)(\sigma^2 \mij)^{1/\sigma^2} } y_{ij}^{1/\sigma^2 - 1} \exp \left(-y_{ij}/(\sigma^2 \mij)\right) & \ \text{if} \  y_{ij}>0,
    \end{array}
    \right.  \label{ZIGpmf}
\end{equation}

with $\ln \mij =  \eta_{ij}$. For  $\pij=0$, the ZAGA corresponds to a standard Gamma model, whilst values $\pij\in (0,1)$ allow for modelling the excess of zeros. 

For each $k$-th layer, the corresponding marginal density is 
\begin{equation}
\begin{aligned}
    \dens & = \int \exp \Biggl\{  \snv   \elk (\bt; \yijk, \bxx, \bw, \bzk) + \ln \left(  (2 \pi)^{-q / 2} \exp \left(-\frac{\bztkt \bztk}{2}\right) \right) \Biggr\} d \bztk    \\
    & = \int \exp \Biggl\{  \snv \left( \llp + \llg (\bzk)\right) - \frac{q}{2}\ln (2 \pi) - \frac{\bztkt \bztk}{2}\Biggr\} d \bztk    \\
    & = \exp \left(  \snv \llp \right) \int \exp \Biggl\{  \snv \left( \llg (\bzk)\right) - \frac{q}{2}\ln (2 \pi) - \frac{\bztkt \bztk}{2}\Biggr\} d \bztk,\\
    \label{marg-dens-ZAGA}
\end{aligned}
\end{equation}

where we have
\begin{equation*}
    \begin{aligned}
        \elk (\bt; \yijk, \bxx, \bw, \bzk) & = \underbrace{\left({\ind{=}} \ln \pij + {\ind{>}} \ln (1-\pij) \right)}_{:= \llp} \\
        & \quad + \underbrace{{\ind{>}} \left( (1/\sigma^2 -1) \ln \yijk  - \ln \Gamma (1/\sigma^2) - 1/\sigma^2 \left(2 \ln \sigma + \ln \mijk \right)   - \frac{\yijk}{\sigma^2 \mijk} \right)}_{:= \llg (\bzk)}.
    \end{aligned}
\end{equation*}

The integral in the last expression of~Equation~\eqref{marg-dens-ZAGA} is equal to $\int \exp \left\{m \qf \right\} d \bztk$ with 
\begin{equation}
\label{mQ-ZAGA}
    \begin{aligned}
        & m \qf[,\sigma] = \Bigl[\snv \llg (\bzk)  - \frac{\bztkt \bztk}{2}-\frac{q}{2} \ln (2 \pi) \Bigr]  \\
        & =  \Biggl\{\snv {\ind{>}} \Biggl[ (1/\sigma^2 -1) \ln \yijk  - \ln \Gamma (1/\sigma^2) - 1/\sigma^2 \left(2 \ln \sigma + \ln \mijk \right)   - \frac{\yijk}{\sigma^2 \mijk} \Biggr] \\
        & \quad - \frac{q}{2}\ln (2 \pi) - \frac{\bztkt \bztk}{2} \Biggr\}.
    \end{aligned}
\end{equation}

We see that the likelihood derived from~(\ref{marg-dens-ZAGA}) factorizes in two parts and can be optimized separately to estimate $\pi_{ij}$ and to estimate all the other parameters.  

For the zero-inflation part, the exact MLE can be obtained based on a dichotomized version of all the $y_{ij}^{(k)}$. 
More precisely, to estimate the $\pi_{ij}$, define  $ \ell_{\pi} := \sk \snv \llp$, $ \noij : = \sk \ind{=}$ and  $\sk \ind{>} = K - \noij$.  For each edge $i\to j$, we solve
    \begin{eqnarray*}
    \label{eq:pihatmle}
      \lefteqn{\sk \partial_{\pij} \elk (\bt; \yijk, \bxx, \bw, \bzk) = \sk \partial_{\pij} \left( \llp + \llg (\bzk)\right) = \sk \partial_{\pij} \llp}\\
        = & 
        \sk \partial_{\pij} \left({\ind{=}} \ln \pij + {\ind{>}} \ln (1-\pij) \right)  =
        \frac{\noij}{\pij} - \frac{K - \noij}{1- \pij} = 0,
    \end{eqnarray*}
    to obtain $ \hat{\pi}_{ij, MLE} = \frac{\noij}{K}$.
Alternatively, in the presence of a set of covariates with design matrix $\boldsymbol{X}_{\pi, ij}$, 
one can link $\pi_{ij}^{(k)}$ to a linear predictor $\eta_{\pi,ij}^{(k)} = \boldsymbol({X}_{\pi,ij}^{(k)})^\T \bb_{\pi}$ through a logit transformation $\pi_{ij}^{(k)} = \operatorname{logit}^{-1} (\eta_{\pi,ij}^{(k)})$. The covariates can be either edge- and layer-dependent or not, and $\bb_{\pi}$ are their associated coefficients.

For the Gamma part, whose likelihood is in integral form, the log-likelihood need to be approximated using Laplace and the estimation is done based on $K - \noij$ observations per layer. In addition to $m \qf[,\sigma]$ as in~Equation~\eqref{mQ-ZAGA}, we have that 
$$
\gf
= 
\snv {\ind{>}} \left( \frac{\yijk}{\sigma^2 \mijk} \right) \aat + \boldsymbol{I}_q,
$$
so that finally the Laplace approximated log-likelihood~Equation~\eqref{Laplace-approx-SM}  corresponding to the Gamma part is 
\begin{equation*}
\begin{aligned}
    \tilde{\ell}_{K}(\bt)&= \sk \Biggl(-\frac{1}{2} \ln \left[\operatorname{det}\left\{\gf\right\}\right] 
    + 
    \snv \Biggl\{ \llg \left( \bzhk\right) \Biggr\}-\frac{\hat{\bz}_{\mathbf{( 2 )}}^{(k)^{\top}} \hat{\bz}_{(2)}^{(k)}}{2}\Biggr) \\
    & = \sk \Biggl(-\frac{1}{2} \ln \left[\operatorname{det}\left\{\gf\right\}\right]\\
    & + \snv \Biggl\{ {\ind{>}} \left( (1/\sigma^2 -1) \ln \yijk  - \ln \Gamma (1/\sigma^2) - 1/\sigma^2 \left(2 \ln \sigma + \ln \mijk \right)   - \frac{\yijk}{\sigma^2 \mijk} \right)  \Biggr\} \\
    & \quad  -\frac{\hat{\bz}_{\mathbf{( 2 )}}^{(k)^{\top}} \hat{\bz}_{(2)}^{(k)}}{2}\Biggr).
\end{aligned}
\end{equation*}

\section{Analytical derivations and  proofs} \label{AppA}

\subsection{Lemmas, Corollaries, and Proofs}
\label{app:tech-det}

To prove \Autoref{thm:model1_double} about double asymptotics, 
we develop a sequence of intermediate lemmas establishing the ingredients needed for the final result. The proof strategy is closely related to that of Theorem $1$ in \citet{ogdenAsymptoticValidityNaive2017a}. Specifically, let $\epsk[\bt]=\ln \tdens-\ln \dens$ denote the difference between the Laplace-approximated and exact log-likelihoods. We first represent the Laplace approximation error as a series expansion, analogously to Equation~$(4)$ in \citet{SMcC95}. We then determine the order of this series by studying the order of its individual components in Lemmas~\ref{lem:hatQr-order}-\ref{lem:Q2inv-order}. Combining these intermediate results yields Lemma~\ref{lem:diff-ll}, which establishes the order of $\epsk[\bt]$. We then study the derivative of $\epsk[\bt]$ with respect to a fixed parameter. Proceeding in the same way, we derive the orders of the relevant intermediate terms in Lemmas~\ref{lem:mode-deriv-vector}-Corollary~\ref{lem:order-dQP-vector}, and combine them in Lemma~\ref{lem:order-depsilon-vector} to obtain the order of this derivative. Finally, Lemma~\ref{lem:delta-inf} provides the final ingredient needed to complete the proof of \autoref{thm:model1_double} and establish the desired asymptotic result. 

Prior to start, it is useful to introduce the following Definitions. 

\begin{definition}[Sup-norm] 
    \label{def:sup-norm}
    For a tensor $\boldsymbol{A}_r$ with dimension $q^r$, define the componentwise sup-norm (max entry norm)

    $$
    \|\boldsymbol{A}_r\|_{\max}
    :=
    \max \left\{ 
    \left|
    \bigl(\boldsymbol{A}_r\bigr)_{a_1\cdots a_r}
    \right| \vert
    1\le a_1,\ldots,a_r\le q
    \right\} 
    =
    \max_{1\le a_1,\dots,a_r\le q}
    \left|
    (\boldsymbol{A}_r)_{a_1 \ldots a_r}
    \right|.
    $$
\end{definition}

\begin{definition}[Spectral norm]
    \label{def:eig-norm}
    For a matrix $\boldsymbol{A}$, the spectral norm $\snorm{\cdot}$ is defined as

    $$
    \snorm{\boldsymbol{A}}
    :=
    \sigma_1(\boldsymbol{A}),
    $$

    the largest singular value of $\boldsymbol{A}$.
\end{definition}

\begin{definition}[Norms and balls used throughout]
For any vector $\bxx$, $\mnorm{\bxx}=\max_i |x_i|$.
For any matrix $\boldsymbol{A}$, $\mnorm{\boldsymbol{A}}=\max_i \sum_j |A_{ij}|$ (the largest absolute row sum).
For any tensor $\boldsymbol{T}$, $\enorm{\boldsymbol{T}}$ is as in Definition~\ref{def:sup-norm}. For any centre $\boldsymbol{a} \in \mathbb{R}^n$ and radius $t>0$, define the ball
$$
\Ball{\boldsymbol{a}}{t}:=\{\boldsymbol{u} \in \mathbb{R}^n:\mnorm{\boldsymbol{u}-\boldsymbol{a}}<t\}.
$$

The norm $\mnorm{\cdot}$ will be written as $\norm{\cdot}$ for simplicity. 
\end{definition}

In this Section, we write $Q(\bt, \bztkh)$ for $\qf[,\zhatk_{(2)}]$ and $\elk (\bt, \bzk)$ for $\elk (\bt; \yijk, \bxx, \bw, \bzk)$ for brevity. We let $\bt^{l}$ denote the fact that we fix $\theta_l = \theta$ in $\bt$, where $l=1,\ldots, \dim(\boldsymbol{\ThKM})$.

We introduce an assumption to control the differentiability term-by-term of the series, which in turns controls the differentiability of the Laplace error expansion in the likelihood score. To this end, we proceed similarly to \cite{SMcC95} and \cite{ogdenAsymptoticValidityNaive2017a,ogdenErrorLaplaceApproximations2021} and make use of their formal expansion of the Laplace approximation error. More precisely, for positive integers $v$ and $s$, we define the set of S-bipartitions $\mathcal{M}_{v,s}$ to be all $(P, T)$ such that $P = \left(p_1 | \ldots | p_v\right)$ and $T = \left(t_1 | \ldots | t_s\right)$ are both partitions of $\left\{ 1, \ldots , 2s\right\}$, such that each block of $P$ contains at least three elements and each block of $T$ contains exactly two elements. For each $(P,T) \in \mathcal{M}_{v,s}$, define a corresponding graph $\mathcal{G} (P,T)$ with vertices $1, \ldots , 2s$, and an edge between each pair of vertices contained in the same block of either $P$ or $T$. If $\mathcal{G} (P,T)$ is a connected graph, we say that $(P,T)$ is a connected bipartition, and write $(P,T) \in \mathcal{M}_{v,s}^C$. We define the level of $(P,T) \in \mathcal{M}_{v,s}^C$ to be $d=s-v$, and write $\mathcal{M}_d^C$ for all connected  level-$d$ S-bipartitions. Write 
$$
  \boldsymbol{\boldsymbol{\hat{Q}}}_{r}^{(k)} (\tl) = - m \, D^r_{\bz} Q\l(\bt^{l}, \bxx,  \bw,  \mathbf{\hat{z}}^{(k)}, \boldsymbol y^{(k)}\r),
  $$ 
  an order $r$ tensor of $r$-th derivatives and for an index tuple $j_t \in T$ write $\boldsymbol{\hat{Q}}^{j_t} = \left(\boldsymbol{\hat{Q}}_2^{-1} (\tl)\right)_{j_t}$ for the chosen entries of the inverse of the Hessian, according to the chosen bipartition. 
  For notational brevity, we henceforth write $\boldsymbol{\boldsymbol{\hat{Q}}}_{r}(\tl)$ in place of $\boldsymbol{\boldsymbol{\hat{Q}}}_{r}^{(k)}(\tl)$, while emphasizing that this quantity retains its dependence on the layer $k$.
  For each $\tl$, define the 
  expansion
\begin{equation}
  \label{eq:eps-error}
  \epsk[\bt^l]  
  =
  \sum_{s=1}^\infty 
  \sum_{P, T} 
  \frac{(-1)^v }{(2s)!}
  \sum_{j \in \left[ 1:q\right]^{2s}}
  \boldsymbol{\hat{Q}}_{j_{p_1}}(\tl) \ldots
  \boldsymbol{\hat{Q}}_{j_{p_v}}(\tl)
  \boldsymbol{\hat{Q}}^{j_{t_1}}  \ldots
  \boldsymbol{\hat{Q}}^{j_{t_s}} =  \sum_{s=1}^\infty 
  \sum_{P, T} 
  \frac{(-1)^v }{(2s)!}  Q_{P,T}(\tl),
\end{equation}

where $\left[1:q\right]^{2s} = \left\{ \left(j_1, \ldots, j_{2s}\right): j_l \in \left\{1, \ldots, q\right\}\right\}$ and $j_p$ is the sub-vector of $j=\left(j_1, \ldots, j_{2s}\right)$ corresponding to the indices in $p$, where $\sum_{P, T}$ considers all connected partitions of level $d$, and $Q_{P,T}$ depends on $k$ since it is related to $\hat Q_(j_p)$ and $\hat Q^{j_t}$ which all depend on $k$. 
Following \citep{ogdenAsymptoticValidityNaive2017a,ogdenErrorLaplaceApproximations2021}, we interpret Equation~Equation~\eqref{eq:eps-error} as a formal expansion. Before proceeding, we illustrate the meaning of its terms with an example of the computation of $Q_{P,T}$ in the right hand-side expression in Equation~Equation~\eqref{eq:eps-error} that also provides intuition for the procedure.\\

\textit{Example: Take $q=3$, $s=2$, then $2s=4$ and  
$
[1:3]^4
=
\left\{(j_1,j_2,j_3,j_4):j_r\in\{1,2,3\},\ r=1,\ldots,4\right\},
$
which contains $3^4=81$ index tuples. Since every block of $P$ must contain at least three elements, and we have $4$ indices, necessarily the only possible option consists of a single block with $4$ elements, hence
$v=1$ and $P=(1234)$. Since each block of $T$ must contain exactly two elements, the possible pair partitions are $T_1 = (12|34), T_2 = (13|24), T_3 = (14|23)$ and $\mathcal{M}_{1,2}= \left\{ (P,T_1), (P,T_2), (P,T_3)\right\}$. 
All three bipartitions are connected because the single block of $P$
contains all four indices. This will be of importance for Equation~Equation~\eqref{eq:eps-error-per-level-d} below, which regroups all connected bipartitions of level $d$. 
For example, for $(P,T_1)$,
\begin{align*}
Q_{P,T_1}(\tl)
&=
\sum_{j\in[1:3]^4}
\hat Q_{j_{p_1}}(\tl)
\hat Q^{j_{t_1}}
\hat Q^{j_{t_2}}\\
&=
\sum_{j_1=1}^{3}
\sum_{j_2=1}^{3}
\sum_{j_3=1}^{3}
\sum_{j_4=1}^{3}
\hat Q_{j_1j_2j_3j_4}(\tl)
\hat Q^{j_1j_2}
\hat Q^{j_3j_4} \\
&\quad=
\underbrace{
\hat Q_{1111}(\tl)
\hat Q^{11}
\hat Q^{11}
}_{(j_1,j_2,j_3,j_4)=(1,1,1,1)}
\\
&\qquad\quad+
\underbrace{
\hat Q_{1112}(\tl)
\hat Q^{11}
\hat Q^{12}
}_{(j_1,j_2,j_3,j_4)=(1,1,1,2)}
+\cdots
\\
&\qquad\quad+
\underbrace{
\hat Q_{3333}(\tl)
\hat Q^{33}
\hat Q^{33}
}_{(j_1,j_2,j_3,j_4)=(3,3,3,3)}.
\end{align*}
Similarly,
\begin{align*}
Q_{P,T_2}(\tl)
&=
\sum_{j_1,j_2,j_3,j_4=1}^{3}
\hat Q_{j_1j_2j_3j_4}(\tl)
\hat Q^{j_1j_3}
\hat Q^{j_2j_4},\\
Q_{P,T_3}(\tl)
&=
\sum_{j_1,j_2,j_3,j_4=1}^{3}
\hat Q_{j_1j_2j_3j_4}(\tl)
\hat Q^{j_1j_4}
\hat Q^{j_2j_3}.
\end{align*}
The $s=2$ contribution in the right hand-side expression in Equation~Equation~\eqref{eq:eps-error} with $v=1$ is
$$
-\frac{1}{4!}
\left\{
Q_{P,T_1}(\tl)
+
Q_{P,T_2}(\tl)
+
Q_{P,T_3}(\tl)
\right\}.
$$
Thus, the contribution considered above is of level $d=s-v=2-1=1$. By comparison, when $s=3$, the underlying set contains $2s=6$ elements, and $P$ may consist either of two blocks of size three, yielding $v=2$ and hence $d=1$, or of a single block of size six, yielding $v=1$ and hence $d=2$. 
} \\

We now define the following asymptotic expansion to $\epsk[\tl]$:
\begin{equation}
  \label{eq:as-eps-error-per-level-d}
  \begin{aligned}
    \epsk[\bt^l] & \sim \sum_{d=1}^{\infty} \epsilon_d^{(k)} (\tl), \quad m \to \infty,
  \end{aligned}
\end{equation}
with 
\begin{equation}
  \label{eq:eps-error-per-level-d}
  \begin{aligned}
    \epsilon_d^{(k)} (\bt^l) & 
    =
  \sum_{P, T \in \mathcal{M}_d^C} 
  \underbrace{\frac{(-1)^{v (P)}}{(2s(P,T))!}
  Q_{P,T}(\tl)}_{:=\epsilon_{P,T}(\tl)}.
  \end{aligned}
\end{equation}
For example, the contribution for $d=1$ will come from $s \in \{2,3\}$.
We interpret Equation~\eqref{eq:as-eps-error-per-level-d} as a formal asymptotic expansion ordered by the level $d$, without requiring convergence of the infinite series. Our analysis follows the termwise approach used in Section~1 of the supplementary material to \citet{ogdenAsymptoticValidityNaive2017a}. To relate the formally differentiated expansion to the derivative of the actual approximation error, we impose the following remainder condition.

For each fixed truncation level $D \ge 1$, write
\begin{equation}
  \label{eq:eps-D-rem}
  \epsk[\bt^l]
  =
  \sum_{d=1}^{D}\epsilon_d^{(k)}(\tl)
  +
  R_{D+1}^{(k)}(\tl).
\end{equation}
We assume that the terms and remainder are continuously differentiable with respect to $\tl$. We first establish the orders of the tensor contractions contributing to each level $d$ and of their derivatives with respect to $\tl$. These bounds show that successive levels decrease by a factor $m^{-1}$ and identify the leading orders of the formal expansions for $\epsk[\bt^l]$ and $\partial_{\tl}\epsk[\bt^l]$. 

\setcounter{theorem}{0} % make sure Lemma numbering starts from 1 

\begin{lemma}[Order of $\boldsymbol{\boldsymbol{\hat{Q}}}_{r} (\tl)$]\label{lem:hatQr-order}
Under Assumption~\textbf{A4}~(\ref{ass:smooth-z}), for each fixed integer $r_{max}\ge4$ and 
$2 \le r \le r_{max}$,
$$
\|\boldsymbol{\boldsymbol{\hat{Q}}}_{r} (\tl)\|_{\max}= \bigo{m},
$$
uniformly in $k$ and $l$.
\end{lemma}

\begin{proof}
By definition, $Q(\bt, \bzk)=m^{-1}\{\snv \elk(\bt, \bzk)- \tfrac12 \bztkt \bztk -\tfrac q2\ln(2\pi)\}$. For $r\ge 3$ the Gaussian term $- \tfrac12 \bztkt \bztk -\tfrac q2\ln(2\pi)$ contributes no derivatives, hence
$$
\boldsymbol{\boldsymbol{\hat{Q}}}_{r} (\tl)
=
-m\,D^r_{\bz} \, Q(\bt^l, \bztkh)
=
-\snv D^r_{\bz} \, \elk (\bt^l, \bzk).
$$

By Assumption~\textbf{A4}~(\ref{ass:smooth-z}),
$$
\|\boldsymbol{\boldsymbol{\hat{Q}}}_{r} (\tl)\|_{\max}
\le
\snv\|D^r_{\bz} \, \, \elk (\bt^l, \bzk) \|_{\max}
\le 
\snv \sup_{\bz \in \Ball{\zhatk(\tl)}{\epsilon}} \|D^r_{\bz} \, \, \elk (\bt^l, \bzk) \|_{\max}
\le m C_r.  
$$
For $r=2$ the same argument applies, with an additional $+\boldsymbol{I}_q$ term in $D^r_{\bz} \, Q(\bt^l, \bztkh)$ (from the derivative of $\tfrac12 \bztkt \bztk$), which is $\bigo{1}$.
\end{proof}

\begin{lemma}[Order of $\boldsymbol{\hat{Q}}_2(\tl)^{-1}$ and $\left(\boldsymbol{\hat{Q}}_2^{-1} (\tl)\right)_{j_k}$]\label{lem:Q2inv-order}
Under Assumption~\textbf{A4}~(\ref{ass:curvature-mode}) and for a fixed $q$,

$$
\snorm{\boldsymbol{\hat{Q}}_2(\tl)^{-1}}
= \bigo{m^{-1}}
\quad\text{and hence}\quad
\mnorm{\boldsymbol{\hat{Q}}_2(\tl)^{-1}}
= \bigo{m^{-1}},
\quad\text{and}\quad
\left|  \left(\boldsymbol{\hat{Q}}_2^{-1} (\tl)\right)_{j_k} \right| = \bigo{m^{-1}},
$$

uniformly in $k$ and $l$. 
\end{lemma}

\begin{proof}
Assumption~\textbf{A4}~(\ref{ass:curvature-mode}) implies
$\lambda_{\min}(\boldsymbol{\hat{Q}}_2(\tl))\ge c m$ and $\lambda_{\max}(\boldsymbol{\hat{Q}}_2(\tl))\le C m$. Since $\boldsymbol{\hat{Q}}_2(\tl)$ is symmetric positive definite with fixed dimension $q$, eigenvalues and singular values are identical and
$$
\enorm{\boldsymbol{\hat{Q}}_2(\tl)^{-1}}
\le
\mnorm{\boldsymbol{\hat{Q}}_2(\tl)^{-1}} 
\le 
\sqrt{q} \snorm{\boldsymbol{\hat{Q}}_2(\tl)^{-1}} 
= 
\sqrt{q} \lambda_{\max} (\boldsymbol{\hat{Q}}_2(\tl)^{-1}) 
=  
\frac{\sqrt{q}}{\lambda_{\min}(\boldsymbol{\hat{Q}}_2(\tl))}
\le 
\frac{\sqrt{q}}{cm},
$$

which proves the result.
Recall $\mnorm{.}$ is the largest absolute row sum. Furthermore, note that for any pair-block $j_k$, $\left|\left(\boldsymbol{\hat{Q}}_2^{-1} (\tl)\right)_{j_k} \right| \le \mnorm{\boldsymbol{\hat{Q}}_2^{-1} (\tl)} \le \frac{\sqrt{q}}{cm}$, hence 
$$
\left|  \left(\boldsymbol{\hat{Q}}_2^{-1} (\tl)\right)_{j_k} \right| = \bigo{m^{-1}}.
$$
\end{proof}

\begin{lemma}[Difference of the likelihoods]
    \label{lem:diff-ll}
    Fix $\tl \in \bt^l $. 
    Under Assumption~\textbf{A4} and Conditions 1-2 of \citet{ogdenErrorLaplaceApproximations2021}, with effective information about each integrated coordinate proportional to $m$ and fixed $q$, uniformly in $\theta$, it holds that

    \begin{equation}
        \label{eq:diff-ll}
        \epsk[\bt^l] = \ln \tilde{f}_{\bt^l}(\boldsymbol{y}^{(k)})-\ln {f}_{\bt^l}(\boldsymbol{y}^{(k)}) = \bigo{m^{-1}}.
    \end{equation}
\end{lemma}

\begin{proof}
    Analogously to Supplementary file of \cite{ogdenAsymptoticValidityNaive2017a}, and using Eq.$(4)$ in \citet{SMcC95},
the error in the Laplace approximation to the log-likelihood $\ell^{(k)}$ can be written as a series as in Equation~\eqref{eq:eps-error}.

Using Lemma~\ref{lem:hatQr-order},

\begin{equation}
    \label{eq:order-of-qr}
    \|\boldsymbol{\hat{Q}}_{r} (\tl)\|_{\max} 
    = 
    \bigo{m} \mbox{ for } r\ge 2.
\end{equation}

Note that $\boldsymbol{D}_{\bz}^r Q(\bt^l, \bztkh)= \boldsymbol{0}$ is identically zero when $r=1$ and corresponds to the value $Q(\bt^l, \bztkh)$, already present in the leading Laplace term, when $r=0$. Hence, we focus on $r \ge 2$. Using Lemma~\ref{lem:Q2inv-order}, uniformly in $k$ and $l$, 
$\enorm{\boldsymbol{\hat{Q}}_2(\tl)^{-1}}=\bigo{m^{-1}}$.
Hence, for any admissible partitions $(P,T)$ with $v$ blocks and the corresponding $s$ in Equation~\eqref{eq:eps-error}, with dimension of the integral ($q$) fixed, the fully contracted scalar result satisfies
$$
Q_{P,T}(\tl) = \bigo{m^{v-s}} = \bigo{m^{-d}}.
$$
Since there are $2s$ indices in total and each of $v$ blocks' size is at least $3$, 
we have that $|j_{p_i}| \ge 3, i = 1, \ldots, v$, $s \ge 2$, hence $2s = \sum_{i=1}^{v} |j_{p_i}| \ge 3v \Leftrightarrow v \le 2s/3$ and 
    $$
    v \le \left\lfloor 2s/3 \right\rfloor \Leftrightarrow v - s \le \left\lfloor 2s/3 \right\rfloor - s = \left\lfloor -s/3 \right\rfloor.
    $$
    This is a non-increasing piecewise constant (step) function, and for $s\ge 2$, the maximum is attained at $s\in \left\{2,3\right\}$ and is equal to $-1$. Putting it differently, among admissible partitions, the largest possible $-d=v-s$ is $-1$. Consequently, each term in Equation~\eqref{eq:eps-error} is at most of order $\bigo{m^{-1}}$, and therefore $Q_{P,T}(\tl) = \bigo{m^{-1}}$ and $\epsk[\bt^l] = \bigo{m^{-1}}$. This aligns with the results obtained in \citet{Bianconcini2014}, Appendix A, where the expansion in Equation~$(4)$ of \citet{SMcC95} is applied for AGH-based estimators. 
\end{proof}

\begin{lemma}[Mode equation and derivative of $\bztkh(\tl)$]
    \label{lem:mode-deriv-vector}
    For each $k$, we have 
    \begin{equation}
    \label{eq:dzhat-vector-general}
    \frac{d}{d \tl}\bztkh(\tl)
    =
    -\Bigl\{\boldsymbol{D}_{\bz}^2 Q(\bt^{l}, \bztkh)\Bigr\}^{-1}
    \partial_{\tl}\boldsymbol{D}_{\bz} Q(\bt^{l}, \bztkh).
    \end{equation}

    Moreover, 

    \begin{equation}
    \label{eq:dzhat-vector-gamma}
    \frac{d}{d \tl}\bztkh(\tl)
    =
    \Gamma_k(\tl)^{-1}
    \snv \partial_{\tl}
    \Bigl\{\boldsymbol{D}_{\bz}\elk (\bt^{l}, \zhk )\Bigr\},
    \end{equation}
    where
    \begin{equation}
    \label{eq:Gamma-k-def}
    \Gamma_k(\tl)
    :=
    -\snv \boldsymbol{D}_{\bz}^2\elk (\bt^{l}, \zhk)  
    +\boldsymbol{I}_q
    =
    -m\,\boldsymbol{D}_{\bz}^2 Q(\bt^{l}, \bztkh).
    \end{equation}
\end{lemma}

\begin{proof}
By the first order condition we have that $\boldsymbol{D}_{\bz}Q(\bt^l, \bztkh) = \mathbf 0$. Differentiating this identity with respect to $\tl$ and applying the multivariate chain rule, we obtain:

$$
\frac{d}{d \tl}    \boldsymbol{D}_{\bz}Q(\bt^l, \bztkh) 
=
\Bigl\{ \boldsymbol{D}_{\bz}^2 Q(\bt^l, \bztkh)\Bigr\} \frac{d}{d \tl}\bztkh(\tl)
+
\partial_{\tl} \boldsymbol{D}_{\bz} Q(\bt^l, \bztkh) = \mathbf 0.
$$

Solving the expression above for $d\bztkh(\tl)/d\tl$ gives Equation~\eqref{eq:dzhat-vector-general}. Using the definitions of $Q$ and $\Gamma$ functions yields

$$
\boldsymbol{D}_{\bz}  Q(\bt^l, \bztkh)
=
m^{-1}\left\{\snv   \boldsymbol{D}_{\bz}  \elk(  \bt^l, \bzk)-\bztkh\right\},
$$

hence

$$
\boldsymbol{D}_{\bz}^2 Q(\bt^l, \bztkh)
=
m^{-1}\left\{\snv \boldsymbol{D}_{\bz}^2\elk( \bt^l,\zhk) 
-\boldsymbol{I}_q\right\}
=
-m^{-1}\Gamma_k(\tl),
$$

and
$$
\partial_{\tl} \boldsymbol{D}_{\bz} Q(\bt^l, \bztkh) 
=
m^{-1}\snv \partial_{\tl}\Bigl\{\boldsymbol{D}_{\bz}\elk( \bt^l, \zhk )\Bigr\}.
$$

Substituting the last expression into Equation~\eqref{eq:dzhat-vector-general} gives Equation~\eqref{eq:dzhat-vector-gamma}.
\end{proof}

%--------------------------------------------------------
% Orders
%--------------------------------------------------------

\setcounter{theorem}{0} % make sure Corollary numbering starts from 1 

\begin{corollary}[Order of $d\bztkh(\tl)/d\tl$]
    \label{cor:order-dzhat-vector}
    Let $m_{\tl}:=|\mathcal I_{\tl}|$ with

    $$
    \mathcal I_{\tl}:=\{(i,j):i\neq j,\ \elk(\cdot)\  \text{depends on } \tl \}.
    $$

    Under Assumption~\textbf{A4}~(\ref{ass:mixed-derivs})-(\ref{ass:curvature-mode}), we have

    \begin{equation}
        \label{eq:order-dzhat-vector}
        \frac{d}{d \tl}\bztkh(\tl)=\bigo{\frac{m_{\tl}}{m}},
    \end{equation}

    uniformly in $k$ and $l$. In particular, $m_{\tl}=1$ for dyad-specific $\alpha_{0,ij}, \alpha_{(2),ij}, \varphi_{ij}$ (or dyad-specific $\beta_{ij}^{(l_x)}$) and $m_{\tl}=m$ for global $\bg^{(l_w)}$, $l_x = 1, \ldots, L_x$, $l_w = 1, \ldots, L_w$.
\end{corollary}

\begin{proof}
  First, note 
  $$
  \snv
  \partial_{\tl}
  \boldsymbol{D}_{\bz}  \elk(  \bt^l, \bzk)
  =
  \sum_{(i,j)\in\mathcal I_{\tl}}
  \partial_{\tl}
  \boldsymbol{D}_{\bz}  \elk(  \bt^l, \bzk)
  $$
  since the likelihood terms not depending on $\theta$ have no contribution to the sum.
    From Lemma~\ref{lem:mode-deriv-vector}, 

    $$
    \norm{\frac{d}{d \tl}\bztkh(\tl)}
    \le
    \norm{\Gamma_k(\tl)^{-1}}
    \norm{\sum_{(i,j)\in\mathcal I_{\tl}}
    \partial_{\tl}\Bigl\{\boldsymbol{D}_{\bz}\elk \left(\bt^l, \bzk \right)\Bigr\}}.
    $$

    Under Assumption~\textbf{A4}~(\ref{ass:curvature-mode}), and using Lemma~\ref{lem:Q2inv-order}, $\norm{\boldsymbol{\hat{Q}}_2(\tl)^{-1}} = \norm{\Gamma_k(\tl)^{-1}} =\bigo{m^{-1}}$. 
    Using Assumption~\textbf{A4}~(\ref{ass:mixed-derivs}) bounds the sum by $\bigo{m_{\tl}}$.
   Multiplying all the terms yields Equation~\eqref{eq:order-dzhat-vector}. 
\end{proof}

\setcounter{theorem}{4} % make sure Lemma numbering continues from 5

\begin{lemma}[Orders of $\partial_{\tl} \boldsymbol{\hat{Q}}_r({\tl})$, and $\boldsymbol{\hat{Q}}_r'({\tl})$]
    \label{lem:orders-hatQr-vector}

    Under Assumption~\textbf{A4}, uniformly in $k$ and $l$, it holds

    \begin{align}
        \label{eq:hatQr-order-vector}
        \boldsymbol{\boldsymbol{\hat{Q}}}_{r} (\tl)&=\bigo{m},\\
        \label{eq:partial-hatQr-order-vector}
        \partial_{\tl}\boldsymbol{\boldsymbol{\hat{Q}}}_{r}\bigl(\tl\bigr)&=\bigo{m_{\tl}},\\
        \label{eq:hatQr-prime-order-vector}
        \boldsymbol{\boldsymbol{\hat{Q}}}_{r}' (\tl)=\frac{d}{d \tl}\boldsymbol{\boldsymbol{\hat{Q}}}_{r} (\tl)&=\bigo{m_{\tl}}.
    \end{align}
\end{lemma}

\begin{proof}

By Lemma~\ref{lem:hatQr-order}, $\boldsymbol{\boldsymbol{\hat{Q}}}_{r} (\tl)=\bigo{m}$.
The partial derivative $\partial_{\tl}$ treats $\bzth$ as fixed, so only dyads in $\mathcal I_{\tl}$ contribute:

$$
\partial_{\tl} \boldsymbol{\boldsymbol{\hat{Q}}}_{r} (\tl)
=
- \sum_{(i,j)\in\mathcal I_{\tl}}
\partial_{\tl}\,\boldsymbol{D}_{\bz}^r \elk ( \bt^l,  \bzk),
$$

hence, using the max-entry norm and the triangle inequality,

\begin{align*}
    \enorm{\partial_{\tl}\boldsymbol{\boldsymbol{\hat{Q}}}_{r} (\tl)}
    &\le
    \sum_{(i,j)\in\mathcal I_{\tl}}
    \enorm{\partial_{\tl}\,\boldsymbol{D}_{\bz}^r \elk( \bt^l, \bzk)} 
    \le
    \sum_{(i,j)\in\mathcal I_{\tl}}
    \sup_{\bz \in \Ball{\zhatk(\bt)}{\epsilon}}
    \enorm{\partial_{\tl}\,\boldsymbol{D}_{z}^r \elk( \bt^l, \bzk)} \\
    & \le
    \sum_{(i,j)\in\mathcal I_{\tl}} C_{r,1}
    =
    m_{\tl}\,C_{r,1},
\end{align*}

and therefore

$$
\enorm{\partial_{\tl}\boldsymbol{\boldsymbol{\hat{Q}}}_{r} (\tl)}
=
\bigo{m_{\tl}},
$$
which proves Equation~\eqref{eq:partial-hatQr-order-vector}.

To prove Equation~\eqref{eq:hatQr-prime-order-vector}, we apply the chain rule in tensor form for the total derivative. 

\begin{align}
\label{hatQr-prime}
    \boldsymbol{\boldsymbol{\hat{Q}}}_{r}' (\tl)
    & = \partial_{\tl}\boldsymbol{\boldsymbol{\hat{Q}}}_{r} (\tl)
    + 
    \sum_{a=1}^q 
 \boldsymbol{\hat{Q}}_{r,a}
    \frac{d}{d \tl}{\hat{z}}_{(2),a}^{(k)}(\tl),
\end{align}

where $\boldsymbol{z}_{(2),a}$ identifies the $a$ column of the matrix $\boldsymbol{z}_{(2)}$, and $\boldsymbol{\hat{Q}}_{r,a} = - m \partial_{\bz_{(2),a}} \boldsymbol{D}_{\bz}^r Q\left(\bt^l, \bztkh\right)$. 
The first term in the right-hand side of~Equation~\eqref{hatQr-prime} is $\bigo{m_{\tl}}$ as above. The second term has order

$$
\boldsymbol{\boldsymbol{\hat{Q}}}_{r+1}(\tl) 
\norm{d\hat\bz_2^{(k)}(\tl)/d\tl} 
= 
\bigo{m}\cdot\bigo{m_{\tl}/m}=\bigo{m_{\tl}}
$$

by Equation~\eqref{eq:hatQr-order-vector} and Corollary~\ref{cor:order-dzhat-vector}.
Therefore, $\boldsymbol{\boldsymbol{\hat{Q}}}_{r}' (\tl)=\bigo{m_{\tl}}$. 
Note that for a chosen $r_{\max}$, under \textbf{A4}, the bounds for $\boldsymbol{\hat{Q}}_r$ and its partial derivative hold for $2\le r\le r_{\max}$. The bound for its total derivative holds for $2\le r\le r_{\max}-1$.
\end{proof}

\setcounter{theorem}{1} % make sure Corollary numbering continues from 2

\begin{corollary}[Order of $dQ_{P,T}(\tl)/d\tl$]
\label{lem:order-dQP-vector}
Under the conditions of Lemma~\ref{lem:orders-hatQr-vector} and Assumption~\textbf{A4}~(\ref{ass:curvature-mode}),
\begin{equation}
\label{eq:order-dQP-vector}
\frac{d}{d \tl}Q_{P,T}(\tl)=\bigo{m_{\tl}\,m^{\,v-s-1}} = \bigo{m_{\tl}\,m^{\,-d-1}}.
\end{equation}
Consequently, since $v-s-1\le -2$ for all $s\ge 2$ (because 
$2s=\sum_{i=1}^v |j_{p_i}| \ge 3v$),
we have

$$
\frac{d}{d \tl}Q_{P,T}(\tl)=
\begin{cases}
\bigo{m^{-2}}, & m_{\tl}=1\ \text{(dyad-specific $\baa, \bphi$ or dyad-specific $\bb$)},\\[1mm]
\bigo{m^{-1}}, & m_{\tl}=m\ \text{(global $\bg$)}.
\end{cases}
$$
\end{corollary}

\begin{proof}
Let $\left[\boldsymbol{\hat{Q}}_2^{\prime}\right]_{cd}$ denote the $(c,d)$ entry of $\boldsymbol{\hat{Q}}_2^{\prime}$ matrix. Differentiate
$Q_{P,T}(\tl)
    = 
    \sum_{j \in \left[ 1:q\right]^{2s}} 
    \boldsymbol{\hat{Q}}_{j_{p_1}}(\tl) \ldots
  \boldsymbol{\hat{Q}}_{j_{p_v}}(\tl)
  \boldsymbol{\hat{Q}}^{j_{t_1}}  \ldots
  \boldsymbol{\hat{Q}}^{j_{t_s}} =
  \prod_{a=1}^{v} \boldsymbol{\hat{Q}}_{j_{p_a}}(\tl)
  \prod_{b=1}^{s} \boldsymbol{\hat{Q}}^{j_{t_b}}$:
\begin{align*}
  \frac{d}{d \tl}Q_{P,T}(\tl)
  & =
  \sum_{j \in \left[ 1:q\right]^{2s}}
  \Bigg\{
  \sum_{i=1}^{v} \left[ \boldsymbol{\hat{Q}}_{j_{p_i}}'(\tl)\prod_{j \in \left\{1, \ldots, v\right\} \backslash i}\boldsymbol{\hat{Q}}_{j_{p_j}}(\tl)
  \boldsymbol{\hat{Q}}^{j_{t_1}}  \ldots
    \boldsymbol{\hat{Q}}^{j_{t_s}}
    \right] \\
    & \qquad -
    \sum_{b=1}^{s} 
    \left[
      \left(
        \sum_{c=1}^{q} \sum_{d=1}^{q}
        \boldsymbol{\hat{Q}}^{j_{t_b} c}
        \left[\boldsymbol{\hat{Q}}_2^{\prime} (\tl)\right]_{cd}
        \boldsymbol{\hat{Q}}^{d j_{t_b}}
      \right)
      \prod_{a \in \left\{1, \ldots, s\right\} \backslash b}
      \boldsymbol{\hat{Q}}^{j_{t_a}} 
      \prod_{e=1}^{v} \boldsymbol{\hat{Q}}_{j_{p_e}}(\tl)
    \right]
    \Bigg\}.
\end{align*}

From Lemmas~\ref{lem:hatQr-order} and \ref{lem:orders-hatQr-vector}, it follows that $\boldsymbol{\hat{Q}}_{j_p}'(\tl)=\bigo{m_{\tl}}$, $\boldsymbol{\hat{Q}}_{j_p}(\tl)=\bigo{m}$, $\boldsymbol{\hat{Q}}^{j_t}=\bigo{m^{-1}}$, and $\boldsymbol{\hat{Q}}_2'(\tl)=\bigo{m_{\tl}}$. Each summand is $\bigo{m_{\tl}}\bigo{m^{v-1}}\bigo{m^{-s}}=\bigo{m_{\tl} m^{v-s-1}}$, and the term subtracted from the first has the same order. This yields Equation~\eqref{eq:order-dQP-vector}. 
\end{proof}
The above Lemmas control the differential of each term in Equation~Equation~\eqref{eq:eps-error}-Equation~\eqref{eq:eps-D-rem}. To conclude on the order of $\partial_{\tl} \epsilon_k (\boldsymbol{\theta}^l)$, we proceed as follows. We fix $D=1$, and first establish the order of the leading error term. We then establish the order of the remainder term and conclude on the final result.

\setcounter{theorem}{5} % make sure Lemma numbering continues from 6

\begin{lemma}[Order of $\partial_{\tl} \epsilon_k (\boldsymbol{\theta}^l)$]
    \label{lem:order-depsilon-vector}
    For $\epsk[\bt^l]$ as in Equation~\eqref{eq:eps-D-rem}, and under the conditions of Lemma~\ref{lem:order-dQP-vector},
    \begin{equation}
        \label{eq:order-depsilon-vector}
        \partial_{\tl}\epsk[\bt^l] =\bigo{\frac{m_{\tl}}{m^2}},
    \end{equation}

    uniformly in $k$ and $l$. In particular,

    $$
    \partial_{\tl}\epsk[\bt^l] 
    =
    \begin{cases}
    \bigo{m^{-2}}, & m_{\tl}=1\ \text{($\tl$ is an entry in a dyad-specific $\baa, \bphi$ or dyad-specific $\bb$)},\\[1mm]
    \bigo{m^{-1}}, & m_{\tl}=m\ \text{($\tl$ is an entry in a global $\bg$)}.
    \end{cases}
    $$
\end{lemma}

\begin{proof}
    Using Equation~Equation~\eqref{eq:eps-D-rem} with $D=1$, 
    $$
    \epsk[\tl]
    =
    \epsilon_1^{(k)}(\tl)
    +
    R_{2}^{(k)}(\tl).
    $$
    Since all terms are differentiable,
$$
    \frac{d}{d \tl} \epsk[\tl]
    =
    \frac{d}{d \tl} \epsilon_1^{(k)}(\tl)
    +
    \frac{d}{d \tl} R_{2}^{(k)}(\tl).
    $$
    By Corollary~\ref{lem:order-dQP-vector}, with $r_{\max}\ge 5$, for each admissible $(s,P,T)$,
    $\frac{d}{d \tl}Q_{P,T}(\tl)=\bigo{m_{\tl} m^{v-s-1}}$ and $v-s-1\le -2$ for all $s\ge 2$.
    Thus, $\frac{d}{d \tl}Q_{P,T}(\tl)=\bigo{m_{\tl} m^{-2}}$ for each $(P,T)\in\mathcal{M}_1^C$. Now,
    $$
    \frac{d}{d \tl} \epsilon_1^{(k)}(\tl) 
    = 
    \frac{d}{d \tl}
    \sum_{P, T \in \mathcal{M}_1^C} 
    \frac{(-1)^{v (P)}}{(2s(P,T))!}
    Q_{P,T}(\tl)
    = 
    \sum_{P, T \in \mathcal{M}_1^C} 
    \frac{(-1)^{v (P)}}{(2s(P,T))!}
    \frac{d}{d \tl}
    Q_{P,T}(\tl).
    $$
    Hence, $\partial_{\tl}\epsilon_1^{(k)}(\tl) = \bigo{m_{\tl} m^{-2}}$. 
    By assuming the expansion is valid under parameter differentiation, and following upon \citet{ogdenAsymptoticValidityNaive2017a}, the derivative of the remainder term is asymptotically smaller than the last retained term of the formally differentiated expansion, hence $\partial_{\tl}R_{2}^{(k)}(\tl) = \bigo{m_{\tl} m^{-3}}$, and 
    $$
    \partial_{\tl}  \epsk[\tl]
    =
    \partial_{\tl} \epsilon_1^{(k)}(\tl)
    +
    \bigo{m_{\tl} m^{-3}} 
    = \bigo{m_{\tl} m^{-2}},
    $$
    proving Equation~\eqref{eq:order-depsilon-vector}.
\end{proof}

The result in Lemma~\ref{lem:order-depsilon-vector} is still for a univariate $\theta$. In the next Lemma we prove the result that we will need to control the error term for a vector of model parameters. To this end, let us
consider $\bt_E$, containing all edge-specific parameters, and its subvector $\bvt_E$, corresponding to a subset of dyads. Label the corresponding true population parameter $\bvt_{E,0}$. Consider furthermore the global block $\bt_G$ and its corresponding true population parameter $\bt_{G,0}$. Let $\bt_0$ be either $\bvt_{E,0}$ or $\bt_{G,0}$. We consider this partition to ensure the objects studied henceforth are of finite dimension. However, these results can be applied to all subsets of parameters. 

\begin{lemma}
\label{lem:delta-inf}
Under Corollary~\ref{cor:order-dzhat-vector} and Lemma
\ref{lem:order-depsilon-vector}, uniformly in $k$ and $\bvt \in\Ball{\bt_0}{t}$,

\begin{equation}
    \label{eq:order-delta}
    \delk
    =
    \norm{ \partial_{\bvt} \epsk[\bvt]}
    =
    \begin{cases} 
        \bigo{m^{-2}}, & \text{ for } \bvt=\bvt_E,\\[1mm]
        \bigo{m^{-1}}, &  \text{ for } \bvt=\bt_G.
    \end{cases}
\end{equation}

Furthermore, define
    
    $$
    \delta_{K,m}^{\infty} (\Ball{\bt_0}{t}) = \SupK[1 \le k \le K][\bvt \in \Ball{\bt_0}{t}] 
    \norm{ \partial_{\bvt} 
    \epsk[\bvt]}.
    $$
Then, for all $m\ge m_0$, there exist constants $C_E,C_G<\infty$ and $m_0<\infty$, such that 
\begin{equation}
    \label{eq:order-delta-inf}
    \delta_{K,m}^{\infty}(\Ball{\bt_0}{t})
    \le
    \begin{cases}
    C_E\,m^{-2}, & \text{ for } \bvt=\bvt_E,\\[1mm]
    C_G\,m^{-1}, & \text{ for } \bvt=\bt_G.
    \end{cases}
\end{equation}

\end{lemma}

\begin{proof}
  Recall $\norm{\bxx} = \mnorm{\bxx}=\max_i |x_i|$ is the maximum absolute value of a vector. Using Lemma~\ref{lem:order-depsilon-vector},  uniformly in $k$ and $l$. It follows that
  $$
  \norm{ \partial_{\bvt} \epsk[\bvt]}
  =
  \bigo{\frac{\max_{\bt^l\in\bvt}m_{\bt^l}}{m^2}}.
  $$
  Hence, the order of $\delk$ is
    $$
    \delk
    =
    \norm{ \partial_{\bvt} \epsk[\bvt]}
    =
    \begin{cases} 
        \bigo{m^{-2}}, & \text{ for } \bvt=\bvt_E,\\[1mm]
        \bigo{m^{-1}}, &  \text{ for } \bvt=\bt_G.
    \end{cases}
    $$

    uniformly in $k$ and uniformly for $\bvt$ in a neighbourhood of $\bt_0$, proving Equation~\eqref{eq:order-delta}. 

    Equation~Equation~\eqref{eq:order-delta-inf} follows by applying Assumption~\textbf{A4}~(\ref{ass:unique-optimizer})-(\ref{ass:curvature-mode}) and by choosing $t$ such that $\Ball{\bt_0}{t} \subset \Theta_0$. By Assumption~\textbf{A4}~(\ref{ass:smooth-z})-(\ref{ass:mixed-derivs}), for each derivative order $r \le r_{\max}$, the bounds on $\boldsymbol{D}_{\bz}^r \ell_{ij}^{(k)}$ and $\partial_{\tl} \, \boldsymbol{D}_{\bz}^r \ell_{ij}^{(k)}$ hold uniformly over $1 \le k \le K$ and $\bvt \in \ThKM$, hence in particular uniformly on $\bvt \in \Ball{\bt_0}{t}$.

    By Assumption~\textbf{A4}~(\ref{ass:curvature-mode}), the matrix $m^{-1} \boldsymbol{\hat{Q}}_2 (\tl)$ has eigenvalues uniformly bounded away from $0$ and infinity on $\Ball{\bt_0}{t}$, so by Lemma~\ref{lem:Q2inv-order} $\left\| \boldsymbol{\hat{Q}}_2 (\tl)^{-1}\right\| \le \sqrt{q}/(cm)$  uniformly in $k$ and $\bvt \in \Ball{\bt_0}{t}$. Therefore, every step in the pointwise bound can be written with constants depending only on $\left\{C_r, C_{r,1}, c, C, r_{\max}\right\}$. Consequently, there exist $m_0$ and finite constants $C_E, C_G$ such that, for all $m \ge m_0, 1 \le k \le K, \bvt \in \Ball{\bt_0}{t}$,

    $$
    \delta_{K,m}^{\infty} (\mathcal{B}(\bt_0,t))
    =
    \SupK[1 \le k \le K][\bvt \in \Ball{\bt_0}{t}] 
    \norm{\partial_{\bvt} \epsk[\bvt]} 
    \le
    \begin{cases}
        C_E\,m^{-2} \text{ (dyad-specific $\baa, \bphi$ or dyad-specific $\bb$)},\\[1mm]
        C_G\,m^{-1} \text{ (global $\bg$)}.
    \end{cases}
    $$
\end{proof}

\subsection{Proof of Theorem~\ref{thm:model1_double}} 
\label{proof-thm1}

\begin{proof}
    The arguments follow similar lines to those in \citet{Bianconcini2014} (Appendix B-C), \citet{V96}, and \citet{RVL09}, but need to be adapted to our multiview network setting, with the key additional step being control of the score approximation error, as in the approximate-likelihood theory of \citet{ogdenAsymptoticValidityNaive2017a}. The proof proceeds in three main steps, establishing consistency,  rate of convergence, and asymptotic normality. Define the effective information rates $n_{E,0}=K,$ and $n_{G,0}=Km$, with $m = \Theta(K^{\varrho})$, 
    and write $n_{0}$ for either edge-specific $(n_0= n_{E,0} )$ or global $(n_0=n_{G,0})$, depending on the chosen parameter. 

    \begin{enumerate}[(i)]
    \item \label{res:score-bias} \emph{Approximate score and bias order.} 
    First, 
    
    $$
    \tsco[\bvt]
    =
    \sco[\bvt]
    + 
    \partial_{\bvt} \epsk[\bvt],
    $$

    where using Lemma~\ref{lem:order-depsilon-vector} $\partial_{\bvt} \, \epsk[\bvt]$ is $\bigo{m^{-1}}$ for global parameters and $\bigo{m^{-2}}$ for edge-specific parameters uniformly in $\bvt$ and $\sco[\bvt]$ is the exact score. Summing across layers gives the total approximate score:

    $$
    \frac{1}{n_0} \sk \tsco[\bvt]
    =
    \frac{1}{n_0}\sk \sco[\bvt]
    + 
    \frac{1}{n_0}\sk \partial_{\bvt} \epsk[\bvt].
    $$

    By using Lemmas~\ref{lem:order-depsilon-vector} and \ref{lem:delta-inf},

    $$
    \delk[\bvt] = \left\| \partial_{\bvt}\epsk[\bvt]\right\| 
    = 
    \begin{cases}
    \bigo{m^{-2}}, & \bvt=\bvt_E,\\
    \bigo{m^{-1}}, & \bvt=\bt_G.
    \end{cases}
    $$

    This decomposition allows for a granular treatment of edge-dependent and edge-independent parameters. Indeed,
    $$
    \left\|
    \frac{1}{n_0}\sk \partial_{\bvt} \epsk[\bvt]
    \right\|
    =
    \left\|
    \frac{1}{n_0}\sk \partial_{\bvt} \epsk[\bvt] 
    \right\|
    \le
    \frac{K}{n_0}
    \sup_{\bvt \in \mathcal{B}(\vto, t)}\left\|\partial_{\bvt} \epsk[\bvt] \right\|.
    $$

    For $\bvt_E$, $K/n_{0}=1$, giving $\bigo{m^{-2}}$.
    For $\bt_G$, $K/n_{0}=1/m$, giving $(1/m)\,\bigo{m^{-1}}=\bigo{m^{-2}}$. Hence, in both cases

    $$
    \frac{1}{n_0}\sk \tsco[\vto]
    =
    \frac{1}{n_0}\sk \sco[\vto]
    + \bigo{m^{-2}}.
    $$

    As $m\to\infty$, the Laplace bias $\bigo{m^{-2}}$ vanishes; as $n_0 \to\infty$, sampling variability averages out.

    \item \label{res:consistency-rate}  \emph{Consistency and rate.}
A mean-value expansion of the block first-order condition $\tilde{\boldsymbol{\mathcal{S}}}_K(\tvtkm)= \boldsymbol{0}$ around $\bt_0$ yields
\[
\sk \tsco[\bt_0]
+
\left\{\sk \thes[\bar{\bt}]\right\}
(\tvtkm-\bt_0) = \boldsymbol{0},
\]
for some $\bar{\bt}$ on the line segment between $\tvtkm$ and $\bt_0$.
Rearranging,
\begin{equation}
\label{eq:basic_block_expansion}
\tvtkm-\bt_0
=
-\left[
\frac{1}{n_0}\sk \thes[\bar{\bt}]
\right]^{-1}
\left[
\frac{1}{n_0}\sk \tsco[\bt_0]
\right].
\end{equation}

    Now analyse the two components. We start from the score term. 

    From \ref{res:score-bias}, we have

    $$
    \frac{1}{n_0}
    \sk \tsco[\bt_0]
    =
    \frac{1}{n_0}
    \sk \sco[\bt_0]
    + \bigo{m^{-2}}.
    $$

Under the assumed CLT scaling, 
$(1/n_0)\sk \sco[\bt_0]$ is $\bigop{n_0^{-1/2}}$.  

    Since the second term is deterministic $\bigo{m^{-2}}$, the sum is:
    $$
    \frac{1}{n_0}
    \sk \tsco[\bt_0] = 
    \bigop{n_0^{-1/2}} 
    + \bigo{m^{-2 }} = 
    \bigop{\max\{n_0^{-1/2}, m^{-2}\}}.
    $$
    Next we look at the {Hessian term}. Under regularity conditions, the Hessian converges by the weak law of large numbers (WLLN), so: 
    $$
    \frac{1}{n_0}
    \sk \thes[\bt_0] \xrightarrow{p} -B_{\bvt}(\bt_0),
    $$
    which is non-singular. Thus, the inverse of 
    $n_0^{-1}\sk \thes[\bt_0]$ 
    is still $\bigop{1}$.

    Combining these results:
    \[
    \norm{\tvtkm-\bt_0}
    =
    \bigop{1} \cdot 
    \bigop{\max\{n_0^{-1/2}, m^{-2}\}}
    =
    \bigop{\max\{n_0^{-1/2}, m^{-2}\}}
    \]
    proving~(C\ref{thm-c1-cons-rate}). This establishes consistency as $n_0^{-1/2} \to 0$ and $m^{-2} \to 0$, namely under the double asymptotics where both the network size ($n_V$) and the number of network views ($K$) diverge. 

    \begin{remark}
      \label{rem:as-clt}
        It is important to notice the difference in the application of CLT for global and edge-specific parameters. By assumptions in \autoref{sec:model1_double} (see main paper), global parameters have a fixed dimension, hence the CLT applies as in the standard case. However, as the dimension of the edge-specific parameters is a function of number of edges $m$, one must rather consider high-dimensional applications of CLT. In their work, \citet{dasCentralLimitTheorem2021} developed optimal bound on the parameter dimension growth rate for the validity of CLT in high dimensions. In particular, and in application to GLAMLE, the critical rate of growth of parameter $\pKMe$ should be $\ln \pKMe = o(\sqrt{n_0})$. For edge-specific parameters, we have $n_0 = n_{E,0} = K$ with at worst $\pKMe = m(q+L_x +2) - \left(L_w + \frac{q(q+1)}{2} \right)=m C_{q,L} - C_q$, where $C_{q,L} = q+L_x +2$, $C_q = L_w + \frac{q(q+1)}{2}$ are constants since $L_x, L_w, q$ are fixed. Since $m = \Theta(K^\varrho)$, we have that 
\begin{align*}
    cK^\varrho \le m \le CK^\varrho & \Leftrightarrow c C_{q,L} K^\varrho \le m C_{q,L} \le C C_{q,L} K^\varrho
    \Leftrightarrow
    c C_{q,L} K^\varrho - C_q \le \pKMe \le C C_{q,L} K^\varrho - C_q\\
    & \Leftrightarrow
    \ln (c C_{q,L} K^\varrho - C_q)
    \le 
    \ln \pKMe 
    \le 
    \ln (C C_{q,L} K^\varrho - C_q) \\
    & \Leftrightarrow 
    \frac{\ln (c C_{q,L} K^\varrho - C_q)}{\sqrt{K}} 
    \le 
    \frac{\ln \pKMe }{\sqrt{K}} 
    \le 
    \frac{\ln (C C_{q,L} K^\varrho - C_q)}{\sqrt{K}}.
\end{align*}

For $m$ large enough, $m C_{q,L} - C_q >0$, hence $\ln \pKMe = \ln (m C_{q,L} - C_q)$ is well defined. 
Now consider 

\begin{align*}
    \frac{\ln \pKMe }{\sqrt{K}} 
    & \le 
    \frac{\ln (C C_{q,L} K^\varrho - C_q)}{\sqrt{K}} = \frac{\ln \left(C C_{q,L} K^\varrho \left(1- \frac{C_q}{C C_{q,L} K^\varrho} \right)\right)}{\sqrt{K}} \\
    & = \left( 
    \frac{\ln \left(C_{q,L}C\right)}{\sqrt{K}} 
    +
    \varrho \frac{\ln \left(K\right)}{\sqrt{K}}
    +
    \frac{\ln \left(1- \frac{C_q}{C C_{q,L} K^\varrho} \right)}{\sqrt{K}} \right) \to 0 \text{ as } K \to \infty,
\end{align*}
hence $\ln \pKMe = o(\sqrt{K})$ and as it satisfies the conditions mentioned in \citet{dasCentralLimitTheorem2021}, the CLT applies.
    \end{remark}
    \item \label{res:asymptotic-normality} \emph{Asymptotic normality.}
    To characterize the limiting distribution, we need to specify the relative growth of $m$ and $K$. 
    Assuming
    $m = \Theta(K^{\varrho})$ for some $\varrho > 0$,

    $$
    c K^{\varrho} \le m \le C K^{\varrho}
    \ \Longleftrightarrow\
    \frac{1}{C K^{\varrho}} \le \frac{1}{m} \le \frac{1}{c K^{\varrho}}
    \ \Longleftrightarrow\
    \frac{1}{(C K^{\varrho})^2} \le \frac{1}{m^2} \le \frac{1}{(c K^{\varrho})^2}.
    $$

    Consequently, we have two scenarios:

    \begin{enumerate}[(i)]
        \item global parameters: $n_0 = Km$ and $\bigop{n_0^{1/2}m^{-2}} = \bigop{K^{1/2-3/2\varrho}}$, 
        \item edge-specific parameters: $n_0 = K$ and $\bigop{n_0^{1/2}m^{-2}} = \bigop{K^{1/2-2\varrho}}$.
    \end{enumerate}

    Because $\varrho>0$, $1/2-2\varrho<1/2-3/2\varrho<0$, hence $\bigop{K^{1/2-3/2\varrho}}=o_p(1)$ and $\bigop{K^{1/2-2\varrho}}=o_p(1)$. Precisely, we need $\varrho >1/4$ for edge-specific parameters and $\varrho >1/3$ for global parameters.

    For~(C\ref{thm-c2-as-norm}), in the case of edge-specific parameters, we multiply Equation~\eqref{eq:basic_block_expansion} by $\sqrt{K}$ and use the same decomposition:
$$
\sqrt{K}(\tvtkm-\bt_0)
=
-\left[
\frac{1}{n_0}\sk \thes[\bar{\bt}]
\right]^{-1}
\left[
\frac{\sqrt{K}}{n_0}\sk \sco[\bt_0]
+
\frac{\sqrt{K}}{n_0}\sk \partial_{\bvt} \epsk[\bvt] 
\right].
$$
The first term in brackets is asymptotically normal by the assumed CLT.
For the second term, the bound above gives

$$
\norm{\frac{\sqrt{K}}{n_0}\sk \partial_{\bvt} \epsk[\bvt]}
\le
\frac{K^{3/2}}{n_0} \SupK[1 \le k \le K][\bvt \in \Ball{\bt_0}{t}]
\left\| \partial_{\bvt} \epsk[\bvt] \right\| 
= 
\frac{K^{3/2}}{n_0} \delta_{K,m}^{\infty} (\Ball{\bt_0}{t}).
$$

Using Lemma~\ref{lem:delta-inf}, for $\bvt_E$, this is $\sqrt{K}\,\bigo{m^{-2}} = \bigop{K^{1/2-2\varrho}}$, which is $o_p(1)$.
Finally, the Hessian matrix converges in probability to $-\boldsymbol  B_{\bvt_E}(\bt_0)$, so by Slutsky theorem,
$\sqrt{K}(\tvtkm-\bt_0)$ converges in distribution to a Gaussian distribution with covariance
$\boldsymbol  B_{\bvt_E}(\bt_0)^{-1}\boldsymbol  A_{\bvt_E}(\bt_0) \boldsymbol  B_{\bvt_E}(\bt_0)^{-1}$.

For global parameters, we proceed similarly, but we use  $n_0 = Km$, which yields the $\sqrt{Km}$ rate. Specifically, for $\bt_G$,  the calculation on the order of the differentiated Laplace approximation error becomes (still using  Lemma~\ref{lem:delta-inf}):
$$\norm{\frac{\sqrt{Km}}{n_0}\sk \partial_{\bvt} \epsk[\bvt]} = K^{-1/2}m^{-1/2} O(Km^{-1})=  O(K^{1/2}m^{-3/2}), $$ 
which is $o_p(1)$, if $m = \Theta(K)^{\varrho}$ with $\varrho > 1/3$. Therefore the Laplace score error is asymptotically negligible at the $\sqrt{Km}$ scale. The asymptotic variance is obtained via CLT, as in the case of edge-specific parameters, using $\boldsymbol B_{\bvt_G}(\bt_0)$ and $\boldsymbol A_{\bvt_G}(\bt_0)$ and rescaling by $\sqrt{Km}$.
\end{enumerate}

\end{proof}

\subsection{Different number of edges in each layer}
\label{SecNbEdges}

In the statement and in the proof of Theorem~\ref{thm:model1_double}, for simplicity, we assume a common number of edges $m$
across all $K$ network views. In practice, different layers may exhibit different numbers
of dyads $m_k$, $k=1,\dots,K$. We discuss in detail the implications of this inference aspect.

As far as (C1) is concerned, this has no impact. Indeed, (C1) is an order-of-magnitude
upper bound, and the proof bounds the Laplace bias via
\begin{equation}
\left\| \frac{1}{n_0} \sk \partial_{\bvt} \epsk[\bvt] \right\|
\;\le\; \frac{K}{n_0} \sup_{1\le k \le K} \left\| \partial_{\bvt} \epsk[\bvt] \right\|.
\label{Eq.B2}
\end{equation}
Replacing $m_k$ by the
uniform bound $m=\min_k(m_k)$ enlarges both bias orders -- $\sup_k O(m_k^{-2})=O(m^{-2})$
for edge-specific and $\sup_k O(m_k^{-1})=O(m^{-1})$ for global parameters -- since both
$m_k^{-2}$ and $m_k^{-1}$ are decreasing in $m_k$, and the supremum over $k$ is therefore
attained at the sparsest layer. In either case the right-hand side of Equation~\eqref{Eq.B2} is
only enlarged. Therefore, the inequality (and hence consistency) still holds, provided
every layer's edge count diverges, $\min_k(m_k)\to\infty$.

As far as (C2) is concerned, the normalizing sequence for global parameters cannot simply
be taken as $Km$. The correct effective sample size remains $n_0=\sum_{k=1}^K m_k$, the
total number of dyad--layer observations informative about $\bt_G$. To account for this,
let us (re-)define
\[
\boldsymbol{A}_{G}(\bt_0) = \plim_{K\to\infty} \frac{1}{n_0} \sum_{k=1}^K
\boldsymbol{\tilde{S}}_{G}^{(k)}(\bt_0)\, \boldsymbol{\tilde{S}}_{G}^{(k)}(\bt_0)^\top,
\qquad
\boldsymbol{B}_{G}(\bt_0) = -\plim_{K\to\infty} \frac{1}{n_0} \sum_{k=1}^K
\partial_{\bt_G} \boldsymbol{\tilde{S}}_{G}^{(k)}(\bt_0),
\]
which reduce to the matrices in Theorem~\ref{thm:model1_double}(C2) when $m_k=m$ for all $k$. Adapting
Step~(iii) of the proof of Theorem~\ref{thm:model1_double}, a mean-value expansion of $\tilde
S_K(\hat\bt_G)=\bf{0}$ around $\bt_{G,0}$ gives
\[
\sqrt{n_0}\,(\hat\bt_G-\bt_{G,0}) = -\left[\frac{1}{n_0}\sum_{k=1}^K \partial_{\bt_G}
\boldsymbol{\tilde{S}}_{G}^{(k)}(\bar\bt_G)\right]^{-1} \left\{
\underbrace{\frac{1}{\sqrt{n_0}}\sum_{k=1}^K \boldsymbol{{S}}_{G}^{(k)}(\bt_{G,0})}_{\text{stochastic term}}
\;+\;
\underbrace{\frac{\sqrt{n_0}}{n_0}\sum_{k=1}^K \partial_{\bt_G}\epsk[\bt_{G,0}]}_{\text{bias term}}
\right\}.
\]

By Lemma~\ref{lem:delta-inf}, for the bias term it holds that $\|\partial_{\bt_G}\epsk[\bt_{G,0}]\|=O(m_k^{-1})\le
O(m^{-1})$ for every $k$, with $m=\min_k(m_k)$. Summing over $K$ layers,
\[
\left\|\frac{\sqrt{n_0}}{n_0}\sum_{k=1}^K \partial_{\bt_G}\epsk[\bt_{G,0}]\right\|
= O\!\left(\frac{K}{m\sqrt{n_0}}\right),
\]
which is $o(1)$, and hence negligible, provided
\begin{equation}
\frac{K}{m\sqrt{n_0}} = \frac{K}{m\left(\sum_k m_k\right)^{1/2}} \;\longrightarrow\; 0.
\label{EqNew}
\end{equation}
This condition can be related directly to the growth rate of $m=\min_k(m_k)$ alone: since
$m_k\ge m$ for every $k$, $n_0\ge Km$ regardless of the pattern of heterogeneity, so
\[
\frac{K}{m\sqrt{n_0}} \;\le\; \frac{K}{m\sqrt{Km}} = \frac{K^{1/2}}{m^{3/2}}
= \Theta\!\left(K^{\,1/2-\frac32\varrho}\right) \quad \text{when } m=\Theta(K^\varrho),
\]
which vanishes if and only if $\varrho>1/3$. Thus, $\varrho>1/3$ applied to the growth rate of the
sparsest layer is sufficient for Equation~\eqref{EqNew}, irrespective of how the remaining $m_k$
grow: the homogeneous case $m_k\equiv m$ is the least favourable configuration ($n_0=Km$
exactly), and any additional edges elsewhere only help. The threshold in Theorem~\ref{thm:model1_double} is
therefore unchanged, once $\varrho$ is understood to refer to $\min_k(m_k)$ rather than to
a common $m$.

Under a Lindeberg-type condition ruling out a single dominant layer -- satisfied, e.g., whenever $\max_k m_k/n_0\to 0$ -- the CLT for
independent (across $k$) triangular arrays gives the following behaviour of the stochastic term:
\[
\frac{1}{\sqrt{n_0}}\sum_{k=1}^K \boldsymbol{S}_G^{(k)}(\bt_{G,0}) \xrightarrow{d} N\!\left(\boldsymbol 0,\,
\boldsymbol{A}_G(\bt_0)\right).
\]

 Under Equation~\eqref{EqNew}, the bias term vanishes in probability,
the Hessian term converges to $-\boldsymbol{B}_G(\bt_0)$, and Slutsky's theorem yields
\begin{equation}
\sqrt{n_0}\,(\hat\bt_G-\bt_{G,0}) \xrightarrow{d} N\!\left(\boldsymbol 0,\,
\boldsymbol{B}_{G}(\bt_0)^{-1}\, \boldsymbol{A}_{G}(\bt_0)\, [\boldsymbol{B}_{G}(\bt_0)^{-1}]^\top \right). \label{Eq.GSM}
\end{equation}
This formula (re)emphasizes the important point: $\sqrt{n_0}$, and not $\sqrt{Km}$, is the
correct normalizing sequence. Indeed, if $\min_k(m_k)$ stays asymptotically negligible
relative to $n_0$, so that $Km=o(n_0)$ and Equation~\eqref{EqNew} fails, normalizing by $\sqrt{Km}$
instead of the correct $\sqrt{n_0}$ produces a degenerate limit: 
\[
\sqrt{Km}\,(\hat\bt_G-\bt_{G,0}) = \sqrt{\frac{Km}{n_0}}  \sqrt{n_0}\,
(\hat\bt_G-\bt_{G,0}),
\]
the deterministic factor $\sqrt{Km/n_0}\to 0$ while the second factor converges to the
Gaussian law as in Equation~\eqref{Eq.GSM}, so the product converges to $\boldsymbol 0$ in probability. 

Finally, we notice that, if the $m_k$ are of the same asymptotic order -- i.e.\ there exist
constants $0<c\le C<\infty$, independent of $K$, with $c\,m\le m_k\le C\,m$ for all $k$ --
then $n_0=\Theta(Km)$, condition Equation~\eqref{EqNew} reduces exactly to $\varrho>1/3$ with
$m=\Theta(K^\varrho)$, and (C2) still holds as stated in Theorem~\ref{thm:model1_double}, under the same growth
condition used in the homogeneous case, without the need to track $n_0$ and $Km$
separately.

\subsection{Identifiability conditions}
\label{SecIdentif}

In this section, we derive the conditions that ensure all the model's parameters are correctly identified. First, we take care of rotational indeterminacy related to the factor loadings and latent variables. Then, we discuss the identifiability of parameters related to layer-(in)dependent covariates, see \Autoref{Sec: uniq} in the main paper, specifically Assumptions~\textbf{A2}(i)-\textbf{A2}(iii). Moreover, we explain how to implement the constraints, how to transform (some) covariates, and how to obtain standard errors for the estimated parameters.

\subsubsection{Theoretical aspects} \label{Sec.ConstTeo}

For the sake of notation convenience, let us re-write the linear predictor as: 
\begin{equation}
\eijk
=
\alpha_{0,ij}
+
\baa_{(2),ij}^\T \bztk
+
\bb_*^\T \bxx_*^{(k)}
+
\bg^\T  \bw_{ij},
\label{Eq.AppLP}
\end{equation}
where $\bztk\in\mathbb R^q$ denotes the centred latent vector for layer $k$, $\bw_{ij} \in \mathbb R^{L_w}$ is a vector of layer-independent covariates, i.e. such that $\bw_{ij}$ is common for all layers and $\bg\in \mathbb R^{L_w}$ is common across dyads, and $\baa_{(2),ij} = \left[ \baat \right]_{\cdot, ij} \in\mathbb R^q$ is the vector of non-intercept loadings for dyad $(i,j)$. In Equation~\eqref{Eq.AppLP}, the layer-dependent covariates $\bxx_{*}^{(k)}$ and their respective coefficients $\bb_*$ can be either edge-dependent (e.g. we can set $\bxx_{ij}^{(k)}$ and/or $ \bb_{ij}$), common across all edges (e.g., we can set $\bxx^{(k)}$ and/or $\bb$), or a combination of both. With this regard, a first key remark is that the layer-dependent covariates are required to be non-constant across layers. This condition must be suitably adapted considering the different specifications of the linear predictor, namely the different Models 1-6 introduced in Section~\ref{Sec: frame} of the main paper.
More precisely, we need that:
\begin{enumerate}
    \item \label{case-x} In 
    Model~2, 
    for every covariate $l_x = 1, \ldots, L_x$, there does not exist $t^{l_x} \in \mathbb{R}$ such that
    $\bxx^{l_x} = t^{l_x}\boldsymbol{1}_K$;
    
    \item \label{case-xij} In 
    Models~3-6, 
    for every edge $(i,j ) \in E$ and every covariate $l_x = 1, \ldots, L_x$, there does not exist $t_{ij}^{l_x} \in \mathbb{R}$ such that $
    \bxx_{ij}^{l_x} = t_{ij}^{l_x}\boldsymbol{1}_K$.
\end{enumerate}

The rationale behind these restrictions is simple:  if there exist such $t^{l_x}$ and/or  $t_{ij}^{l_x}$, then the related covariates are in fact layer-independent and, they must be included in the specification of $\bw_{ij}$.

With these conditions in mind, we can proceed with the identification aspect. To this end, for each dyad $(i,j)$, define 
$$
\beij
=
\begin{pmatrix}
\eta_{ij}^{(1)}\\
\vdots\\
\eta_{ij}^{(K)}
\end{pmatrix} \in\mathbb R^K,
\qquad
\boldsymbol{v}_{ij}
=
\begin{pmatrix}
\baa_{(2),ij}^\T \bzt^{(1)}\\
\vdots\\
\baa_{(2),ij}^\T \bzt^{(K)}
\end{pmatrix}
\in\mathbb R^K.
$$

Collect the intercepts into $\baa_0=(\alpha_{0,1},\dots,\alpha_{0,m})^\T\in\mathbb R^m$, and define the layer-independent design matrix
$$
\bW =
\begin{pmatrix}
\bw_1^\T\\
\vdots\\
\bw_m^\T
\end{pmatrix}
\in\mathbb R^{m\times L_w}, \mbox{ with }  \bw_{ij}^\T = ( w_{ij1}, \ldots, w_{ij L_w}), \, ij = 1, \ldots, m.
$$

For each dyad define layer-dependent data matrix
$$
\boldsymbol{X}_{*}
=
\begin{pmatrix}
\bxx_{*}^{(1)\T}\\
\vdots\\
\bxx_{*}^{(K)\T}
\end{pmatrix}
\in\mathbb R^{K\times L_x}, \mbox{ with }  \bxx_{*}^{(k)\T} = ( x_{* 1}^{(k)}, \ldots, x_{* L_x}^{(k)}), \, k=1,\ldots,K, 
$$
where, similarly to notation in Equation~\eqref{Eq.AppLP}, $\boldsymbol{X}_{*}$ can be either edge-specific ($\boldsymbol{X}_{ij}$) or not ($\boldsymbol{X}$). 

We refer to  Assumptions~\textbf{A2}(i)-\textbf{A2}(iii) in the main paper: here  we consider in detail their role for identification purposes. \\

\textit{Latent variables.} Let
$\boldsymbol{O}_q \in\mathbb R^{q\times q}$ 
 be orthogonal, so 
$\boldsymbol{O}_q^\T \boldsymbol{O}_q =\boldsymbol{I}_q$.
Define
$\widetilde{\baa}_{(2),ij}=\boldsymbol{O}_q \baa_{(2),ij},$ 
$\mathbf{\tilde{Z}}_{(2)}=\boldsymbol{O}_q \bZt.$
Then $
\widetilde{\baa}_{(2),ij}^\T \mathbf{\tilde{Z}}_{(2)}
=
\baa_{(2),ij}^\T \boldsymbol{O}_q ^\T \boldsymbol{O}_q \bZt
=
\baa_{(2),ij}^\T \bZt.$
Hence, the latent bilinear term is invariant under orthogonal transformation. Since
 $\bZt \sim N_q(\mathbf 0,\boldsymbol{I}_q)$, the transformed latent vector $\mathbf{\tilde{Z}}_{(2)}$ has the same
distribution as the original one. Therefore, without further restrictions, the model is not uniquely
identified.  Assumption \textbf{A2} (i) takes care of this aspect.  Let $\boldsymbol{B}$ denote the leading $q\times q$ block of $\hat{\baa}_{(2)}$, i.e. the $q$ rows and  first $q$ columns, and suppose that both
$\boldsymbol{B}$ and $\boldsymbol{O}_q\boldsymbol{B}$ are upper triangular with strictly positive diagonal entries.
Since $\boldsymbol{B}$ is upper triangular with nonzero diagonal, it is invertible, and
$\boldsymbol{O}_q =(\boldsymbol{O}_q\boldsymbol{B})\boldsymbol{B}^{-1}.$
Hence, $\boldsymbol{O}_q $ is itself upper triangular.
Since $\boldsymbol{O}_q$ is also orthogonal, it must be diagonal with diagonal entries in $\{-1,1\}$.
Because $\boldsymbol{B} $ and $\boldsymbol{O}_q\boldsymbol{B}$ both have strictly positive diagonal entries, each diagonal sign must
equal $1$. Therefore, $\boldsymbol{O}_q =\boldsymbol{I}_q$.
This proves that the upper-triangular block with positive diagonal removes both the rotation and the residual sign indeterminacy.  

\textit{Covariates}. 
Let us start from Assumption~\textbf{A2}(iii), due to its relation with the latent variables just discussed. To develop further, let us consider 
\begin{equation}
  \label{eq:s-wta0}
  \bs=\baa_0 + \bW\bg, 
\end{equation}
for $\bs\in\mathbb R^m,$
and let  
$
\mathbf{P}_w := \bW (\bW^{\T} \bW )^{-1} \bW^{\T},$ and $
\mathbf{M}_w := \boldsymbol{I}_{m} - \mathbf{P}_w
$
be the orthogonal projectors onto $\operatorname{span}(\bW)$ and its orthogonal complement $\operatorname{span}(\bW)^\perp$, respectively. Then for any $\bs \in\mathbb{R}^m$, consider
$
\bs  = \mathbf{M}_w \bs  + \mathbf{P}_w \bs, $ and $
\mathbf{M}_w\bs \in \operatorname{span}(\bW)^\perp, $
$ \mathbf{P}_w\bs \in \operatorname{span}(\bW).$ Under the normalization $\bW^{\T}\baa_{0}=\boldsymbol{0}$, the unique decomposition
$\bs =\baa_{0}+ \bW \bg$
is given by
$$
\baa_{0} = \mathbf{M}_w \bs ,
\qquad 
\bW \bg = \mathbf{P}_w \bs ,
\qquad
\bg = (\bW^{\T} \bW )^{-1} \bW ^{\T}\bs.
$$
The QR factorization yields
$$
\bW  = \mathbf{Q} \mathbf{R} = \begin{bmatrix}
\mathbf{Q_1} \quad \mathbf{Q_2}
\end{bmatrix}
\begin{bmatrix}
\mathbf{R_1}\\[2pt]
\boldsymbol{0}
\end{bmatrix}
= \mathbf{Q_1} \mathbf{R_1},
$$

where $\mathbf{Q}\in\mathbb{R}^{m\times m}$ is orthogonal, $\mathbf{R}\in\mathbb{R}^{m\times L_{w}}$ is upper triangular,
$\mathbf{Q_1}\in\mathbb{R}^{m\times L_{w}}$ contains the first $L_{w}$ columns of $\mathbf{Q}$, and the remaining
$m-L_{w}$ columns $\mathbf{Q_2} \in \mathbb{R}^{m\times (m-L_{w})}$ span $\ker(\bW^{\T})$. Set
$
\baa_0 := \mathbf{Q_2} \bze,$ and $ \bze\in\mathbb{R}^{m-L_{w}}.$ Then $\bW^{\T}\baa_0=\bW^{\T}Q_2\bze=\boldsymbol{0}$, so the constraint holds exactly. Since $\mathbf{Q_2}$ has orthonormal columns, the orthogonal projector onto $\ker(\bW^{\T})$ is $\mathbf{M}_w=\mathbf{Q_2} \mathbf{Q_2^{\T}}$ and
$\baa_0=\mathbf{Q_2} \mathbf{Q_2^{\T}}\bs $ with $\bs =\baa_0+\bW \bg$. This reduces the free dimension of $\baa_0$ by $L_{w}$. 
We emphasize that this constraint does not modify the identification conditions for the latent part, and it does not alter the linear independence conditions required for the layer-dependent covariates. 

Now we can claim that the constraint $\bW^\T\baa_0=\mathbf{0}$ uniquely separates $\baa_0$ and $\bg$. To this end, consider that $\mathbf{P}_w$ and $\mathbf{M}_w$ are complementary orthogonal projectors, $\mathbf{P}_w+\mathbf{M}_w=\boldsymbol{I}_{m}$ and $\mathbf{P}_w \mathbf{M}_w=\mathbf{M}_w \mathbf{P}_w=\mathbf{0}$. Hence, $\bs =(\mathbf{P}_w+\mathbf{M}_w)\bs =\mathbf{P}_w\bs +\mathbf{M}_w\bs $ with $\mathbf{P}_w\bs \in\operatorname{span}(\bW)$, $\mathbf{M}_w\bs \in\operatorname{span}(\bW)^\perp$. If $\bs =\baa_{0}+\bW \bg$ and $\bW^{\T}\baa_{0}=\mathbf{0}$, then necessarily $\bW\bg=\mathbf{P}_w\bs $ and $\baa_{0}=\mathbf{M}_w\bs $, since 
\begin{align*}
  \mathbf{P}_w &= \bW (\bW^{\T} \bW )^{-1} \bW^{\T} = \mathbf{Q_1} \mathbf{R_1} \left( (\mathbf{Q_1} \mathbf{R_1})^\T \mathbf{Q_1} \mathbf{R_1}\right)^{-1} (\mathbf{Q_1} \mathbf{R_1})^\T \\
  & = \mathbf{Q_1} \mathbf{R_1} \left( \mathbf{R_1}^\T \mathbf{R_1}\right)^{-1} \mathbf{R_1}^\T \mathbf{Q_1}^\T = \mathbf{Q_1} \mathbf{Q_1}^\T, \\
  \mathbf{M}_w &= \mathbf{I} - \mathbf{P}_w = (\mathbf{Q}\mathbf{Q}^\T) - \mathbf{P}_w =\mathbf{ Q_2} \mathbf{Q_2}^\T,
\end{align*}
and 
$$
\baa_0 = \mathbf{M}_w \bs = \mathbf{Q_2} \mathbf{Q_2}^{\T} \bs = \mathbf{Q_2} \boldsymbol{\zeta} \qquad \Rightarrow \qquad \boldsymbol{\zeta} =\mathbf{ Q_2}^{\T} \bs.
$$
Moreover,
$$
\bW \bg = \mathbf{P}_w \bs = \bW (\bW^{\T} \bW )^{-1} \bW^{\T} \bs \qquad \Rightarrow \qquad \bg = (\bW^{\T} \bW)^{-1}\bW^{\T} \bs = \mathbf{R_1}^{-1} \mathbf{Q_1}^\T \bs.
$$
By the property of orthogonal projectors, this decomposition is unique. This proves that the condition in \textbf{A2} (iii), namely $\bW^\T\baa_0=\mathbf{0}$, resolves the confounding between the unrestricted dyad intercepts and the layer-independent covariate effect, provided that $\bW$ has full column rank. The layer-independent term $\bW_{ij} \bg$ does not involve $\baa_{(2),ij}$ or $\bzt$. Therefore, it does not interact with the orthogonal rotational indeterminacy of the latent bilinear
term $\baa_{(2),ij}^\T \bzt$. The same upper-triangular plus positive-diagonal constraint on the leading $q\times q$ block of $\hat{\baa}_{(2)}$ continues to remove the latent rotation.

Finally, we consider \textbf{A2} (ii), which refers to  layer-dependent covariates. Let us consider the specification in Equation~\eqref{Eq.AppLP}: for fixed dyad $(i,j)$, stack the $K$ predictors and write
$$
\beij = \alpha_{0,ij}\boldsymbol{1}_K + \boldsymbol{v}_{ij} + \boldsymbol{X}_{*}\bb_{*} + \left(\mathbf{W}_{ij} \bg\right) \boldsymbol{1}_K.
$$
After the latent $\boldsymbol{v}_{ij}$ and the layer-independent $\mathbf{W}_{ij} \bg$ parts have been identified (so \textbf{A2} (iii) holds), we have: 
$$
\eijk- v_{ij}^{(k)} - \alpha_{0,ij} - \mathbf{W}_{ij} \bg
=
\mathbf{x}_{*}^{(k) \T}\bb_{*}.
$$
This defines a purely parametric linear predictor with design matrix $\boldsymbol{X}_*$ and no intercept. Hence, it suffices that $\boldsymbol{X}_*$ is of full rank, which is imposed by Assumption \textbf{A2} (ii).
In the case where there are no layer-independent covariates and the restriction in \textbf{A2} (iii) is not imposed, the remaining part of the linear predictor is
\begin{equation}
\eijk- v_{ij}^{(k)}
=
\alpha_{0,ij}
+
\mathbf{x}_*^{(k)}\bb_{*} \label{Eq:LMSM}
\end{equation}
and  similar arguments apply. Indeed, for 
Models~2, 3, 4, Equation~\eqref{Eq:LMSM} is similar to a linear predictor associated with a generalized linear model with design matrix $[\boldsymbol{1}_K \ \ \boldsymbol{X}_*]$. Since $\rank(\boldsymbol{X}_*) = L_x$ by Assumption \textbf{A2} (ii) and the covariates are layer-dependent, i.e. there is no $t_{ij}^{l_x} \in \mathbb{R}$ such that $\bxx_{ij}^{l_x} = t_{ij}^{l_x} \boldsymbol{1}_K$ for all $ij \in E$, $l_x = 1, \ldots, L_x$ by construction, and therefore $(\alpha_{0,ij},\bb_{*})$ is identifiable. 

Next, consider the design matrix for $\alpha_{0,ij}
+
\mathbf{x}_*^{(k)}\bb_{*}$ in Equation~\eqref{Eq:LMSM}:
\begin{align*}
  \boldsymbol{D}_* &= [\boldsymbol{1}_K \ \ \boldsymbol{X}_*] \in \mathbb{R}^{K \times L_x}, \quad 
  \boldsymbol{D}
  =
  \begin{pmatrix}
  \boldsymbol{D}_1 & 0   & \cdots & 0 \\
  0   & \boldsymbol{D}_2 & \ddots & \vdots \\
  \vdots & \ddots & \ddots & 0 \\
  0   & \cdots & 0 & \boldsymbol{D}_{m}
  \end{pmatrix}. \\
\end{align*}

We now consider different scenarios for combinations of covariates specifications and their associated parameters. First, consider the case where the covariates are layer-specific with edge-specific coefficients (Model~2).  In this case, we have $\boldsymbol{X}_* = \boldsymbol{X}$, $\boldsymbol{D}_* = \boldsymbol{D}$ and the model
decomposes into $m$ separate regressions all sharing the same design matrix $\boldsymbol{D}$, hence there must be enough variation across layers for the model to be identifiable. Note also that if $m \gg K$, the estimates may become very noisy due to a high number of parameters compared to the number of layers. To improve conditioning across all $m$ regressions, covariates can be centred across layers.

Next, consider the case where the covariates are edge- and layer-specific and their associated coefficients are edge-specific (Models~4 and 5). We thus have $\boldsymbol{D}_{ij}=[\boldsymbol{1}_K \  \boldsymbol{X}_{ij}]$. The block-diagonal design matrix $\boldsymbol{D}$, consisting of blocks $\boldsymbol{D}_{ij}$, is no longer common across all edges, and the identifiability has to hold edge by edge. The columns of $\boldsymbol X$ must be linearly independent across all layers for all edges. i.e. each block $\boldsymbol{D}_{ij}$ must be of full rank. Hence, one must pay particular attention to this fact when constructing the model. It is also important to notice the problem decomposes into $m$ linearly independent systems and the information for $\bb_{ij}$ comes only from $K$, not $mK$ observations.  The stacked design is block-diagonal and there is no pooling of information. To improve conditioning across all $m$ regressions, covariates can be centred across layers. 

Finally, consider the case where the covariates are edge- and layer-specific, but their coefficients are global (Models~3 and 6).  
Define

\begin{align*}
  \boldsymbol{X}_{ij}
  &=
  \begin{pmatrix}
  \bxx_{ij}^{(1)\T}\\
  \vdots\\
  \bxx_{ij}^{(K)\T}
  \end{pmatrix}
  \in\mathbb R^{K\times L_x}, \quad
  \boldsymbol{D}
  =
  \left[
  \boldsymbol{I}_m \otimes \boldsymbol{1}_K
  \ \ 
  \begin{pmatrix}
  \boldsymbol{X}_{1}\\
  \vdots\\
  \boldsymbol{X}_{m}
  \end{pmatrix}
  \right].
  \end{align*}

Since $\rank \left(\boldsymbol{I}_m \otimes \boldsymbol{1}_K \right) = m$ and $\rank \boldsymbol{X} = L_x$ by Assumption \textbf{A2} (ii), it follows that  $\rank \boldsymbol{D} = m + L_x$ and the model is identifiable. The stacked design matrix is no longer block-diagonal and hence there is pooling of information, which thus presents as the simplest case for identifiability.

\subsubsection{Implementation aspects}
\label{Sec-implem}

In this section, we discuss the key steps needed to implement our method. We first consider the problem of enforcing the constraints. Naive implementation of the theoretical constraints discussed in Section \ref{Sec.ConstTeo} results in quadratic memory complexity. Our new implementation of constraints allows reducing the memory cost from quadratic to linear in $m$. Then, we explain how to transform the covariates and estimate the related model parameters. This transformation entails the need for back-transforming the resulting estimates, hence, we explain how to obtain the standard errors of the original model parameters.  Additionally, we tackle a key implementation aspect: the selection of the starting values that help find the solution to the approximate likelihood equations. \\

 \textit{Enforcing the constraints.} While setting up the constraint as in Assumptions \textbf{A2} (i)-(iii) is exact and produces an orthogonal matrix $\boldsymbol{Q}_2$, in practice, specifically when the number of edges $m$ is large, the yielded matrix $\boldsymbol{Q}_2 \in \mathbb{R}^{m\times (m-L_w)}$ is too large for numeric{al operations to be performed efficiently. Now recall that $\bW^\T \baa_0 = \mathbf{0}$ represents a system of $L_w$ linear equations with $m$ variables, where the admissible intercept vector lies in the null space of the covariates data matrix. To gain in numerical efficiency, one can rather use a different basis of the same null space. For example, one can choose $\mathcal{C} \subset \left\{1,\ldots, m\right\}$ rows of size $L_w$ of $\bW$ such that $\mathbf{W}_{\mathcal{C}} \in \mathbb{R}^{L_w \times L_w}$ is nonsingular. Furthermore, let $\mathcal{F}=\left\{1,\ldots, m\right\} \backslash \mathcal{C}$, hence $\mathbf{W}_{\mathcal{F}} \in \mathbb{R}^{(m-L_w) \times L_w}$. Partition accordingly
  $$\baa_0 = \begin{pmatrix}
    \baa_{0, \mathcal{C}} \\
    \baa_{0, \mathcal{F}}
  \end{pmatrix},
  $$
  then the constraint can be written as
  $$
  \mathbf{W}_{\mathcal{C}}^\T \baa_{0, \mathcal{C}}
  +
  \mathbf{W}_{\mathcal{F}}^\T \baa_{0, \mathcal{F}} = \boldsymbol 0.
  $$
  Since $\mathbf{W}_{\mathcal{C}}$ is nonsingular, we can solve for unique set of constrained intercepts:
  $$
  \baa_{0, \mathcal{C}} = -\left(\mathbf{W}_{\mathcal{C}}^\T\right)^{-1} \mathbf{W}_{\mathcal{F}}^\T \baa_{0, \mathcal{F}}.
  $$
  Hence, $\baa_{0, \mathcal{F}} \in \mathbb{R}^{m-L_w}$ are the free parameters to be estimated, and we can reconstruct the total intercept vector as
  $$\baa_0 = \begin{pmatrix}
    \baa_{0, \mathcal{C}} \\
    \baa_{0, \mathcal{F}}
  \end{pmatrix}
  =
  \begin{pmatrix}
    -\left(\mathbf{W}_{\mathcal{C}}^\T\right)^{-1} \mathbf{W}_{\mathcal{F}}^\T  \\
    \mathbf{I}_{m-L_w}
  \end{pmatrix} \baa_{0, \mathcal{F}}
  := \mathbf{G} \baa_{0, \mathcal{F}},
  $$
  and one can show that under this partition, the constraint is satisfied exactly since 
  $$
  \bW^\T \baa_0 = \bW^\T \mathbf{G} \baa_{0, \mathcal{F}} = \left( - \mathbf{W}_{\mathcal{C}}^\T \left(\left(\mathbf{W}_{\mathcal{C}}^\T\right)^{-1} \mathbf{W}_{\mathcal{F}}^\T \right) + \mathbf{W}_{\mathcal{F}}^\T \right) 
  \baa_{0, \mathcal{F}}
  = \mathbf{0}.
  $$
  This method allows to obtain a linear in $m$ memory cost---not a quadratic cost in $m$, which is the one obtained by a naive implementation
  of the constraints, because  $\boldsymbol{Q}_2$ is of dimension $m\times (m-L_w)$, whereas $\mathbf{W}_{\mathcal{F}}$ is of dimension $m-L_w$.\\
  
\textit{Transforming the covariates in Models ~5 and 6}. In our experience, to improve numerical stability in optimization, the covariates should be centred across layers and edges. For the interpretation of the model coefficients, it would be suitable to back-transform the centred and scaled covariates: this entails the need for back-transforming the related coefficients as well. Due to the flexibility of our model setting, this operation requires specific tasks, which change accordingly to the different model specification. Here we discuss Model~5 and 6 which have all the key aspects that one needs to tackle. All the other models can be derived as special cases of these two models. For example, Models~3 and 4 can be recovered by setting $\bg = 0$.

To begin with,  let us start from Model~5, whose linear predictor contains edge specific covariates and coefficients, and additional global parameters. For each $(i,j)$,  define averages and standard errors across layers as follows
  \begin{align*}
    % mu_x
    \mu_{x, ij}^{(l)} &= \frac{1}{K} \sk x_{ij}^{(k),(l)}, 
    \quad 
    && \boldsymbol{\mu}_{x,ij} = (\mu_{x, ij}^{(1)},\dots,\mu_{x, ij}^{(p_x)})^\T \in\mathbb R^{p_x}, \\
    % S_x
    s_{x, ij}^{(l)} &= \left(  \frac{1}{K-1} \sk \left( x_{ij}^{(k),(l)} - \mu_{x, ij}^{(l)}\right)^2 \right)^{1/2}, 
    \quad
    && \mathbf{S}_{x, ij} = \begin{pmatrix}
    s_{x, ij}^{(1)} & 0   & \cdots & 0 \\
    0   & s_{x, ij}^{(2)} & \ddots & \vdots \\
    \vdots & \ddots & \ddots & 0 \\
    0   & \cdots & 0 & s_{x, ij}^{(p_x)}
    \end{pmatrix} \in\mathbb R^{p_x \times p_x}, \\
    % mu_w
    \mu_{w}^{(l)} &= \frac{1}{m} \snv  \bw_{ij}^{(l)}, 
    \quad 
    && \boldsymbol{\mu}_{w} = (\mu_{w}^{(1)},\dots,\mu_{w}^{(p_w)})^\T\in\mathbb R^{p_w}, \\
    % S_w
    s_{w}^{(l)} &= \left( \frac{1}{m-1} \snv \left( \bw_{ij}^{(l)} - \mu_{w}^{(l)}\right)^2 \right)^{1/2}, 
    \quad
    && \mathbf{S}_{w} = \begin{pmatrix}
    s_{w}^{(1)} & 0   & \cdots & 0 \\
    0   & s_{w}^{(2)} & \ddots & \vdots \\
    \vdots & \ddots & \ddots & 0 \\
    0   & \cdots & 0 & s_{w}^{(p_w)}
    \end{pmatrix} \in\mathbb R^{p_w \times p_w}, \\
  \end{align*}
  where $l$ denote a specific covariate. Recall that covariates can be of any type (e.g.\ continuous, discrete, categorical, binary, ordinal) and that usually only the continuous covariates are suitable for centring and scaling. Let $J_x \subseteq \left\{1, \ldots, L_x\right\}, p_x = |J_x|$ denote the number of transformed  layer-dependent covariates, $J_w \subseteq \left\{1, \ldots, L_w\right\}, p_w = |J_w|$ the number of transformed layer-independent covariates. Furthermore, let subscript $t, n$ denote the transformed covariates and the non-transformed covariates, respectively. For each $\bxx_{ij,t}$ with columns $1, \ldots, p_x$, $\bw_t$ with columns $1, \ldots, p_w$, we define their centred and scaled counterpart
  \begin{align*}
    \tilde{\bxx}_{ij,t} &:= (\bxx_{ij,t} - \boldsymbol{1}_K \boldsymbol{\mu}_{x,ij}^\T) \boldsymbol{S}_{x, ij}^{-1} 
    && \Leftrightarrow 
    \quad
    \bxx_{ij,t} = \boldsymbol{1}_K \boldsymbol{\mu}_{x,ij}^\T +  \tilde{\bxx}_{ij,t} \boldsymbol{S}_{x, ij}, \\
    %w
    \widetilde{\bw}_{t} &:= (\bw_t - \boldsymbol{1}_m \boldsymbol{\mu}_{w}^\T) \boldsymbol{S}_{w}^{-1}, 
    && \Leftrightarrow 
    \quad
    \bw_t = \boldsymbol{1}_m \boldsymbol{\mu}_{w}^\T + 
    \widetilde{\bw}_{t} \boldsymbol{S}_{w}.
  \end{align*}
We also define their corresponding parameters 
  \begin{align*}
    \tilde{\bb}_{ij}  &= \boldsymbol{S}_{x, ij} \bb_{ij} 
    && \Leftrightarrow 
    \bb_{ij} = \boldsymbol{S}_{x, ij}^{-1} \tilde{\bb}_{ij}, \\ 
    \tilde{\bg} 
    &= \boldsymbol{S}_{w} \bg 
    && \Leftrightarrow  \bg = \boldsymbol{S}_{w}^{-1} \tilde{\bg}.
  \end{align*}
  Substituting into the linear predictor under the constraint $\bW^\T \baa_0= \mathbf{0}$ gives 
  \begin{align*}
    \beij & =
    \alpha_{0,ij}\boldsymbol{1}_K 
    + 
    \bzt^\T \baat[ij] 
    +
    \bxx_{ij} \bb_{ij}  
    +
    \boldsymbol{1}_K  \left(\bw_{ij} \bg \right) \\ 
    & = 
    (\mathbf{G}_{ij} \ac) \boldsymbol{1}_K 
    + 
    \bzt^\T \baat[ij]
    + 
    \left(\boldsymbol{1}_K \boldsymbol{\mu}_{x,ij}^\T +  \tilde{\bxx}_{ij} \boldsymbol{S}_{x, ij}\right) \boldsymbol{S}_{x,ij}^{-1} \tilde{\bb}_{ij} 
    + 
    \boldsymbol{1}_K  \left(\left(\boldsymbol{\mu}_{w}^\T + \tilde{\bw}_{ij} \boldsymbol{S}_{w}\right) \boldsymbol{S}_{w}^{-1} \tilde{\bg}\right) \\
    & = 
    (\mathbf{G}_{ij} \ac)\boldsymbol{1}_K 
    + 
    \bzt^\T \baat[ij]
    + 
    \left(\boldsymbol{1}_K \boldsymbol{\mu}_{x,ij}^\T \boldsymbol{S}_{x,ij}^{-1} +  \tilde{\bxx}_{ij} \right)  \tilde{\bb}_{ij} 
    + 
    \boldsymbol{1}_K  \left(\left(\boldsymbol{\mu}_{w}^\T \boldsymbol{S}_{w}^{-1} + \tilde{\bw}_{ij} \right)  \tilde{\bg}\right).
  \end{align*}

  In addition we partition the loadings matrix into constrained ($\baat[,\mathcal{C}]$) and free ($\baat[,\mathcal{F}]$) parameters as in the main text and consider block-diagonal matrix and stacked vectors of parameters (where $\mathbf{\tilde I}$ has dimension 
$(q+1)\left(m-q/2\right)-L_w$),
  \begin{align}
    \mathbf{T} &= \begin{pmatrix}
    \mathbf{\tilde I}  & 0   & \cdots & & 0 \\
    0   & \boldsymbol{S}_{x,1}^{-1} & \ddots & \ddots & \vdots \\
    \vdots & \ddots & \ddots & \ddots & \vdots \\
    0   & \cdots & 0 & \mathbf{S}_{x,m}^{-1} & 0 \\
    0   & \cdots & \cdots & 0 & \mathbf{S}_{w}^{-1}
    \end{pmatrix}, 
    \quad 
    \hat{\bt} = \begin{pmatrix}
      \ach \\
      \baath[,\mathcal{F}] \\
      \hat{\bb}_{1} \\
      \vdots \\
      \hat{\bb}_{m} \\
      \hat{\bg}
    \end{pmatrix}, 
    \quad 
    \hat{\tilde{\boldsymbol{\theta}}} = \begin{pmatrix}
      \ach \\
      \baath[,\mathcal{F}] \\
      \hat{\tilde{\bb}}_{1} \\
      \vdots \\
      \hat{\tilde{\bb}}_{m} \\
      \hat{\tilde{\bg}}
    \end{pmatrix}, 
    \label{Eq. TransfD5}
  \end{align}
we obtain  $\hat{\bt} = \mathbf{T} \hat{\tilde{\boldsymbol{\theta}}} $. \\

Now, let us consider  Model~6, which has covariates that vary across edges and layers, but with common coefficients.  In this case we define averages and standard errors across layers and edges
  \begin{align*}
    % mu_x
    \mu_{x}^{(l)} &= \frac{1}{Km} \sk \snv x_{ij}^{(k),(l)}, 
    \quad 
    && \boldsymbol{\mu}_{x} = (\mu_{x}^{(1)},\dots,\mu_{x}^{(p_x)})^\T \in\mathbb R^{p_x}, \\
    % S_x
    s_{x}^{(l)} &= \left(  \frac{1}{Km-1} \sk \snv \left(x_{ij}^{(k),(l)}  - \mu_{x}^{(l)}\right)^2 \right)^{1/2}, 
    \quad
    && \mathbf{S}_{x} = \begin{pmatrix}
    s_{x}^{(1)} & 0   & \cdots & 0 \\
    0   & s_{x}^{(2)} & \ddots & \vdots \\
    \vdots & \ddots & \ddots & 0 \\
    0   & \cdots & 0 & s_{x}^{(p_x)}
    \end{pmatrix} \in\mathbb R^{p_x \times p_x}, \\
  \end{align*}

where $l$ denote a specific covariate, and $J_x \subseteq \left\{1, \ldots, L_x\right\}, p_x = |J_x|$ denotes the number of transformed layer-dependent covariates, $J_w \subseteq \left\{1, \ldots, L_w\right\}, p_w = |J_w|$ the number of transformed layer-independent covariates. We also define  $\mu_{w}^{(l)}, s_w^{(l)}, \bmu[w], \mathbf{S}_w$ as for Model~5 above. 
Set, for each $\bxx_{ij,t}$ with columns $1, \ldots, p_x$, $\bw_t$ with columns $1, \ldots, p_w$,
  \begin{eqnarray*}
    \tilde{\bxx}_{ij,t} := (\bxx_{ij,t} - \boldsymbol{1}_K \boldsymbol{\mu}_{x}^\T) \boldsymbol{S}_{x}^{-1} 
    & \Leftrightarrow &
    \quad
    \bxx_{ij,t} = \boldsymbol{1}_K \boldsymbol{\mu}_{x}^\T +  \tilde{\bxx}_{ij,t} \boldsymbol{S}_{x}, \\
    %w
    \tilde{\bw}_{t} := (\bw_t - \boldsymbol{1}_m \boldsymbol{\mu}_{w}^\T) \boldsymbol{S}_{w}^{-1}, 
    & \Leftrightarrow &
    \quad
    \bw_t = \boldsymbol{1}_m \boldsymbol{\mu}_{w}^\T + \tilde{\bw}_t \boldsymbol{S}_{w}, \\
    \bw_{*}  := \bw_t \boldsymbol{S}_w^{-1}    = 
    \left( \boldsymbol{1}_m \boldsymbol{\mu}_{w}^\T + \tilde{\bw}_t \boldsymbol{S}_{w} \right) \boldsymbol{S}_w^{-1} 
    & =  & \boldsymbol{1}_m \left(\boldsymbol{\mu}_{w}^\T \boldsymbol{S}_w^{-1} \right) + \tilde{\bw}_t .
  \end{eqnarray*}

  Since $\boldsymbol{S}_w^{-1}$ is invertible, the column and null spaces of $\bw, \bw_*$ are identical, hence $\bw^\T \baa_0= \boldsymbol 0 \Leftrightarrow \bw_*^\T \baa_0= \boldsymbol 0$. Define also the relationship between the original and transformed parameter
  \begin{eqnarray*}
    \tilde{\bb}  = \boldsymbol{S}_{x} \bb & \Leftrightarrow  &\bb = \boldsymbol{S}_{x}^{-1} \tilde{\bb}, \\ 
    \tilde{\bg}  = \boldsymbol{S}_{w} \bg & \Leftrightarrow & \bg = \boldsymbol{S}_{w}^{-1} \tilde{\bg}.
  \end{eqnarray*}

  Substituting into the linear predictor under the constraint $\bW^\T \baa_0=0$ and using $\baa_0 = \mathbf{G}  \ac$ gives
  \begin{align*}
    \beij & =
    \alpha_{0,ij}\boldsymbol{1}_K 
    + 
    \bzt^\T \baat[ij] 
    +
    \bxx_{ij} \bb 
    +
    \boldsymbol{1}_K  \left(\bw_{ij} \bg \right) \\ 
    & = 
    (\mathbf{G}_{ij}  \ac)\boldsymbol{1}_K 
    + 
    \bzt^\T \baat[ij]
    + 
    \left(\boldsymbol{1}_K \boldsymbol{\mu}_{x}^\T +  \tilde{\bxx}_{ij} \boldsymbol{S}_{x}\right) \boldsymbol{S}_{x}^{-1} \tilde{\bb} 
    + 
    \boldsymbol{1}_K  \left(\left(\boldsymbol{\mu}_{w}^\T + \tilde{\bw}_{ij} \boldsymbol{S}_{w}\right) \boldsymbol{S}_{w}^{-1} \tilde{\bg}\right) \\
    & = 
    (\mathbf{G}_{ij} \ac)\boldsymbol{1}_K 
    + 
    \bzt^\T \baat[ij]
    + 
    \left(\boldsymbol{1}_K \boldsymbol{\mu}_{x}^\T \boldsymbol{S}_{x}^{-1} +  \tilde{\bxx}_{ij} \right)  \tilde{\bb} 
    + 
    \boldsymbol{1}_K  \left(\left(\boldsymbol{\mu}_{w}^\T \boldsymbol{S}_{w}^{-1} + \tilde{\bw}_{ij} \right)  \tilde{\bg}\right) = \\
    & = 
    (\mathbf{G}_{ij} \ac)\boldsymbol{1}_K 
    + 
    \bzt^\T \baat[ij]
    + 
    \boldsymbol{1}_K\left( \boldsymbol{\mu}_{x}^\T \boldsymbol{S}_{x}^{-1} \tilde{\bb} \right)
    +  
    \tilde{\bxx}_{ij} \tilde{\bb} 
    + 
    \boldsymbol{1}_K  \left(\bw_{*,ij}  \tilde{\bg}\right).
  \end{align*}

  Finally, consider block-diagonal matrix and stacked vectors of estimated free parameters
  \begin{align}
    \mathbf{T} &= \begin{pmatrix}
    \mathbf{\bar I} & 0  & 0&0  & 0 \\
    0 & {1} & 0&0  & 0 \\
    0   &  0 & \mathbf{S}_{x}^{-1} & 0 & 0\\
    0 & 0 & 0 & \mathbf{G} &0 \\
    0 & 0 &0  & 0  & \mathbf{S}_{w}^{-1}
    \end{pmatrix} , 
    \quad 
    \hat\bt = \begin{pmatrix}
      \baath[,\mathcal{F}] \\
      \ln \hat{\sigma} \\
      \hat{\bb} \\
      \hat{\baa}_0 \\
      \hat{\bg}
    \end{pmatrix}, 
    \quad 
    \hat{\tilde{\boldsymbol{\theta}}} = \begin{pmatrix}
      \baath[,\mathcal{F}] \\
      \ln \hat{\sigma} \\
      \hat{\tilde{\bb}} \\
      \ach \\
      \hat{\tilde{\bg}}
    \end{pmatrix},
    \label{Eq. TransfD6} 
  \end{align}
where $\mathbf{\bar I}$ has dimension $mq - \left(q(q+1)/2\right)$. So, $\hat\bt = \mathbf{T} \hat{\tilde{\boldsymbol{\theta}}} $. \\

\textit{Standard errors.} 

Standard errors for the original model parameters are obtained by
back-transforming the covariance matrix of the estimator $\hat{\tilde{\boldsymbol\theta}}$,
which is directly available from the Laplace-approximated likelihood on the centred and scaled
covariates. Two cases must be distinguished, depending on whether
the (back-) transformation is linear or nonlinear in the estimated quantities.

For the intercepts and for the vector of all model parameters, the transformation is a
{linear} map with a fixed (non-random) matrix, namely $\hat{\boldsymbol\alpha}_0 =
\mathbf{G}\hat{\boldsymbol\alpha}_{0,\mathcal{F}}$ and $\hat{\boldsymbol\theta} =
\mathbf{T}\hat{\tilde{\boldsymbol\theta}}$, with $\mathbf{G}$ and $\mathbf{T}$ as defined in
Equation~\eqref{Eq. TransfD5} and Equation~\eqref{Eq. TransfD6} for Models~5 and~6, respectively. In this case the
variance transformation is {exact}. We therefore obtain:
\[
\var\!\left(\hat{\boldsymbol\alpha}_{0}\right) = \mathbf{G}\, \var\!\left(\hat{\boldsymbol\alpha}_{0,\mathcal{F}}\right) \mathbf{G}^{\T},
\qquad
\var\!\left(\hat{\boldsymbol\theta}\right) = \mathbf{T}\, \var\!\left(\hat{\tilde{\boldsymbol\theta}}\right) \mathbf{T}^{\T}.
\]

The situation differs for the dispersion parameter $\sigma$ in the ZAGA specification, which is
estimated on the log scale as $\ln\hat\sigma$; the map $\sigma = \exp(\ln\sigma)$ is nonlinear, so
the corresponding variance can only be obtained via a first-order approximation (delta method),
\[
\var\!\left(\hat{\sigma}\right) \approx \exp\!\left(2\ln\hat\sigma\right) \var\!\left(\ln\hat\sigma\right)
\quad \Leftrightarrow \quad
\operatorname{SE}\!\left(\hat{\sigma}\right) \approx \hat\sigma\, \operatorname{SE}\!\left(\ln\hat\sigma\right),
\]
which is asymptotically valid.

\textit{Starting values.} Another central aspect in  real-data analysis is related to the starting values computation.  
To tackle this aspect,  let us consider  the most general case:  Model~6,   with covariates centred and scaled, under ZAGA specification, which is
the pdf applied in the numerical exercises (both synthetic and real-data) of the main paper. The other model specifications can be deduced straightforwardly. 
The idea is to fit an edge-wise independence model with valid distributional assumptions, followed by a second round of estimations with conditional log-likelihood, much faster to fit than Laplace approximation. 

Let 
$$
Y_{ij}^{k}\mid \bztk \sim \mathrm{ZAGA},
$$
with pmf defined by Equation~Equation~\eqref{ZIGpmf} and linear predictor $\mu_{ij}^{(k)}=\exp(\eta_{ij}^{(k)})$ with
$$
\eijk
  =
  \baa_{ij}^{\T}\bzk
  +
  \bb_{}^{\T}\bxx_{ij}^{(k)}
  +
  \bg^{\T}\bw_{ij}
  =
  \alpha_{0,ij}
  +
  \baat[ij]^{\T} \bztk
  +
  \bb_{}^{\T}\bxx_{ij}^{(k)}
  +
  \bg^{\T}\bw_{ij}
$$
under the constraint  $\bW^{\T} \baa_0= \boldsymbol{0}$, and $\sigma$ estimated on the log scale. 

Partition the covariates into transformed ($t$) and non-transformed ($n$) blocks as above:
\begin{align*}
  \bxx_{ij}^{(k)}=
  \begin{pmatrix}
    \bxx_{n,ij}^{(k)}\\
    \bxx_{t,ij}^{(k)}
  \end{pmatrix},
  \qquad
  \bw_{ij}=
  \begin{pmatrix}
    \bw_{ij,n}\\
    \bw_{ij,t}
  \end{pmatrix}.
\end{align*}
Centre and scale covariates, then define the scaled coefficient vectors
$$
\tilde \bb=
\begin{pmatrix}
  \bb_n\\
  \tilde\bb_t
\end{pmatrix},
\qquad
\tilde \bg=
\begin{pmatrix}
  \bg_n\\
  \tilde\bg_t
\end{pmatrix},
\qquad
\tilde\bb_t = \mathbf{S}_x\bb_t,
\qquad
\tilde\bg_t = \mathbf{S}_w\bg_t.
$$
Define the transformed layer-independent data matrix
$$
\mathbf{W}_*=
\bigl[\mathbf{W}_n, \ \mathbf{W}_t \boldsymbol{S}_w^{-1} \bigr].
$$
Since $\mathbf{W}_* = \mathbf{W} \mathbf{A}_w^{-1}$ with $\mathbf{A}_w=\mathrm{diag}(\mathbf{I},\boldsymbol{S}_w)$, we have
$
\operatorname{col}(\mathbf{W}_*)=\operatorname{col}(\mathbf{W}),$ and $\ker(\mathbf{W}_*^{\T})=\ker(\mathbf{W}^{\T}),$
hence the constraint $\mathbf{W}^{\T} \baa_0 =\mathbf 0$  is equivalent to $\mathbf{W}_*^{\T}\baa_0=\mathbf{0}$. For each edge $(i,j)$, fit (e.g.\ with the \texttt{gamlss} package in \texttt{R})
a ZAGA model  across all  $k=1,\ldots,K$ and with $s_{ij}$ the $ij$-th entry of $\boldsymbol{s}$ as in Equation~\eqref{eq:s-wta0}, $ij=1,\ldots,m$.  We have $Y_{ij}^{(k)}\sim \mathrm{ZAGA}$ and 
\begin{equation}
  \label{eq:lin-pred-start}
  \begin{aligned}
    \ln \mu_{ij}^{(k)}
    &=
    s_{ij} + (\bxx_{n,ij}^{(k)})^\T \bb_{n,ij} + (\bxx_{t,ij}^{(k)})^\T \bb_{t,ij}
    =
    s_{ij} + (\bxx_{n,ij}^{(k)})^\T \bb_{n,ij} 
    + 
    \bmu^\T S_x^{-1} \tilde \bb_{t,ij}
    + 
    (\tilde{\bxx}_{t,ij}^{(k)})^\T \tilde \bb_{t,ij} \\
    & =
    \tilde s_{ij}
    +
    (\bxx_{n,ij}^{(k)})^\T \bb_{n,ij}
    +
    (\tilde{\bxx}_{t,ij}^{(k)})^\T \tilde \bb_{t,ij},
  \end{aligned}
\end{equation}
 
  where
  \begin{equation}
    \label{eq:sij}
    \tilde s_{ij} = s_{ij} +\bmu^\T S_x^{-1} \tilde \bb_{t,ij} \Leftrightarrow s_{ij} = \tilde s_{ij} - \bmu^\T S_x^{-1} \tilde \bb_{t,ij},
  \end{equation}
and denote the corresponding estimators by
$
\hat{\tilde s}_{ij},$ 
$\hat{\tilde \bb}_{ij}= (\hat \bb_{n,ij}^T, \hat{\tilde \bb}_{t,ij}^T)^T$, 
and $\hat\sigma_{ij}.$
Equipped with these quantities,  we provide valid starting values for each parameter (denoted by the superscript $(0)$).  

\begin{enumerate}[(i)]
  \item Start for $\bg$. 
  Compute $\hat{s}_{ij}$ from $\hat{\tilde s}_{ij}$, as in Equation~\eqref{eq:sij}. Combine $\hat{\bs} = (\hat{s}_{1}, \ldots, \hat{s}_{m})$. Compute the starting values for $\bg$ using OLS - see \Autoref{SecIdentif} for more details:
  $$
  \hat{\tilde{\bg}}^{(0)} = (\boldsymbol{W}_{*}^{\T}\boldsymbol{W}_{*})^{-1}\boldsymbol{W}_{*}^{\T}\hat \bs.
  $$
  \item Start for $\ac[F]$.   Compute $\hat{\baa}_0 = \hat{\bs} - \boldsymbol{W}_{*}^{\T} \hat{\tilde{\bg}}^{(0)}$. 
  Finally, set $\ach[F]^{(0)} = \hat{\baa}_{0, \mathcal{F}}$.
  \item Starting values for $\bzt$ and $\baat$.
  Let $\mathbf{R}=(r_{ij}^{(k)})\in\mathbb R^{K\times m}$ be the matrix of Dunn--Smyth residuals from the edgewise  ZAGA fits, with rows indexed by $k$ and columns by $ij$. Define the column-centred residual matrix
  $$
  \mathbf{R}_c =   \mathbf{R}-\mathbf 1_K \bar r^{\T},
  \qquad
  \bar r_{ij} = \frac{1}{K}\sum_{k=1}^{K} r_{ij}^{(k)}.
  $$
  Let $  \mathbf{U}_q$, $  \mathbf{D}_q$, and $  \mathbf{V}_q$ denote the matrices formed by the first $q$ left singular vectors, singular values, and right singular vectors of $\mathbf{R}_c$, respectively. Then
  $$
    \mathbf{R}_c \approx   \mathbf{U}_q   \mathbf{D}_q   \mathbf{V}_q^{\T},
  $$
  where the approximation is the best rank-$q$ approximation to $  \mathbf{R}_c$ in Frobenius norm.
  Then set
  $$
  \bzth^{(0)} = \sqrt{K}\,  \mathbf{U}_q,
  \qquad
  \hat{\baa}_{(2)} = K^{-1/2}   \mathbf{V}_q   \mathbf{D}_q.
  $$

  Finally, to satisfy the additional rotational identifiability constraints in Assumption \textbf{A2}, post-multiply $\hat{\baa}_{(2)}$ by the same orthogonal matrix chosen to satisfy those constraints. Finally, set $\baath^{(0)} = \hat{\baa}_{(2),\mathcal{F}}$.
  
  \item Start for $\bb$. As per Equation~\eqref{eq:lin-pred-start}, the coefficients for the non-transformed variables are unchanged and can be used directly. Then, set for each covariate $l$ amongst the transformed layer-dependent covariates 
  $$
  \hat{\tilde{\beta}}_{t}^{(l)} = \frac{1}{m} \snv \hat{\tilde{\beta}}_{t, ij}^{(l)} , \quad \mbox{ and then }\hat{\tilde{\bb}}^{(0)} = \begin{pmatrix}
    \hat{\bb}_n \\
    \hat{\tilde{\bb}}_{t}
  \end{pmatrix}.
  $$
  \item Start for $\ln \sigma$.
Aggregate the edgewise dispersion estimates on the log scale:
$$
\ln \hat{\sigma}^{(0)} = \ln \left( \frac{1}{m} \snv \hat{\sigma}_{ij}\right). 
$$
\end{enumerate}

Therefore, a coherent starting vector for optimization in the scaled parameterization is
$$
\hat{\tilde\bt}^{(0)}
=
\Bigl(
\ach[F]^{(0)},
\baath^{(0)},
\hat{\tilde \bb}^{(0)},
\hat{\tilde \bg}^{(0)},
\ln\hat{\sigma}^{(0)}
\Bigr)^{\T}. 
$$

As a second step, the vector of parameters $\hat{\tilde\bt}^{(0)}$ can be used to perform a conditional likelihood optimization using TMB to generate a better set of starting values, $\hat{\tilde\bt}^{(1)}$. In particular, starting from $\bzth^{(0)}$ and $\hat{\tilde\bt}^{(0)}$ from the previous step, obtain new values
\begin{align*}
  \left(\hat{\tilde\bt}^{(1)},  \bzth^{(1)}\right)
  &=
  \argmin_{\bt, \bz}
  \sk \ln \l(  \l\{ \prod_{i\ne j}^{n_V} \mathsf{f}_{\bt} 
  \left( \Yijk \vert \eta^{(k)}_{ij}, 
  \bzk
  \right) \r\} 
  h (\bzk)
  \r).
\end{align*} 

Both the centring, scaling, and the correct starting values ensure the gradient and Hessian matrices are numerically stable. \\ 

\section{Monte Carlo simulations: detailed setting } \label{App: MCsettings}

We conducted Monte Carlo simulations to assess the finite-sample performance of the proposed estimators for both edge-specific and global parameters, under several distributional assumptions and model specifications.

We first considered two Poisson specifications. In Models~3 and~4, respectively, the linear predictor $\eijk = \ln \lambda_{ij}^{(k)}$ is given by
$$
  \eijk = \baa_{ij}^{\T}\bzk + \bb^{\T}\bxx_{ij},
  \qquad
  \eijk = \baa_{ij}^{\T}\bzk + \bb_{ij}^{\T}\bxx_{ij}.
$$
We set the number of nodes to $n_V=10$, giving $m=90$ directed edges. We used $q=1$ latent variable, $K=500$ layers and one layer-dependent covariate, so that $L_x=1$ and $L_w=0$. The $q \times K$ latent variable matrix $\bz_{(2)}$ was generated from a standard normal distribution. The $m \times (q+1)$ loading matrix was generated from a uniform distribution on $(0,1)$, subject to the identifiability constraint in Assumption~\textbf{A4}. Recall that the constrained block of the loading matrix was taken to be upper triangular with diagonal entries fixed at one.  For the intercept, we added a $\ln 5$ shift in order to keep the true Poisson means away from zero and to avoid simulations being dominated by zero counts. The regression coefficients and the layer-dependent covariate values were generated independently from a uniform distribution on $(-1,1)$. All parameters were generated only once and reused for each Monte Carlo replication. 

We next considered a zero-adjusted gamma model, denoted by $\text{ZAGA}(\mijk,\sigma,\pij)$, with the parameterization of \citet{R_etal_19}, as defined in Equation~Equation~\eqref{ZIGpmf}. For this case, we ran a Monte Carlo experiment with 100 replications under 
Model~5, 
for which the linear predictor $\eijk=  \ln \mijk$ defined as
\begin{equation}
\label{ZAGA-Model5}
  \eijk
  =
  \baa_{ij}^{\T}\bzk
  +
  \bb_{ij}^{\T}\bxx_{ij}^{(k)}
  +
  \bg^{\T}\bw_{ij},
\end{equation}
and parameters $\pij$ and $\sigma$ specified separately. 

The factor loadings were generated from $\mathrm{Unif}(0,1)$, while the free entries of the intercept were generated from a standard normal distribution. The edge-specific regression coefficients were generated from $\mathrm{Unif}(-0.5,0.5)$, and the layer-dependent and layer-independent covariates, as well as the global coefficient $\bg$, were generated from $\mathrm{Unif}(-1,1)$. We considered one layer-dependent covariate and one layer-independent covariate, $L_x=L_w=1$, together with one latent variable, $q=1$. The scale parameter was fixed at $\sigma=0.3$. The zero-adjustment probabilities $\pij$ were generated from a Gaussian distribution with mean $0.5$ and standard deviation $0.12$, truncated to the interval $[0.01,0.8]$. All parameters were generated only once and reused for each of the Monte Carlo replication. 

For each Monte Carlo replication, we generated an $m \times K$ response matrix $\by$ from the corresponding data-generating model. Estimation was based on the Laplace-approximated log-likelihood. The estimated parameters were compared with their true values, excluding the parameters fixed by the identifiability constraints. The simulation results are reported in the main text in \Autoref{Tfig:bij_xij} for 
Poisson Models~3 and~4, 
and in \Autoref{Tfig:bij-xij-wij-o} for 
ZAGA Model~5. \\

\textit{Computational aspects.} a) The identifiability constraint of Assumption~\textbf{A4} ensures uniqueness of the parameterization up to the usual rotational indeterminacy of latent variable models \citep{HRVF04}; see also the discussion of identifiability in the GLAMLE framework of \citet{JLVR24}. Without such a constraint, the numerical optimizer may converge to an arbitrary representative of an equivalence class of observationally equivalent solutions. This non-uniqueness can also be illustrated empirically by applying a post-hoc Procrustes rotation to the estimated loading matrix, which leaves the fitted model essentially unchanged but alters the interpretation of the loadings.
b) The Laplace approximated-loglikelihood is implemented in \texttt{R} using the \texttt{RTMB} package and maximized using the PORT routines implemented in \texttt{nlminb}. This computational strategy follows the general GLAMLE approach of \citet{JLVR24}, and its use of a Laplace-approximated likelihood is also motivated by the approximate-likelihood theory of \citet{ogdenAsymptoticValidityNaive2017a} and the high-dimensional Laplace error analysis of \citet{ogdenErrorLaplaceApproximations2021}.

\section{Complements to WTO data analysis} \label{AppSec: RealData}

\subsection{Trades data and gravity covariates}
\label{tradesdata}

We consider trade data for 2022, as made available from the WTO Data Statistical Office.
We restrict attention to a subset of $n_V=45$ countries, yielding $m=1980$ directed edges. The variables \texttt{reporter\_country} and \texttt{partner\_country} correspond to importing and exporting countries, respectively. Each edge (exporter-importer) corresponds to an ordered country pair, and for each edge we observe the traded value of goods, recorded by the variable \texttt{value}, across $K=72$ product layers (variable \texttt{product}). The dataset also contains one layer-dependent covariate, \texttt{best}, which reports bilateral best applied simple average tariff data for 2022. Further details on the data source and variable definitions are available from the WTO Data Portal.

We augment these data with layer-independent dyadic covariates obtained from the \texttt{cepiigeodist} package in \texttt{R}. In particular, we retain \texttt{dist}, the simple distance in kilometres between the most populated cities of the two countries, \texttt{comlang\_ethno}, a binary indicator equal to one when at least $9\%$ of the populations of the two countries speak the same language, \texttt{contig}, coded as 1 when the two countries are next to each other and 0 otherwise and \texttt{colony}, coded as 1 when the exporter country was ever a colony of the importer country.

\subsection{Estimation and diagnostics}
\label{tradesdata-estim}

In the fitted model, the traded value is assumed to follow a conditional
$\mathrm{ZAGA}(\pij, \mu_{ij}^{(k)}, \sigma)$ distribution. 
For each directed edge $(i,j)$ and product layer $k$, we fit a zero-adjusted gamma model
$$
Y_{ij}^{(k)} \mid \bzk \sim \mathrm{ZAGA},
\qquad
\log \mu_{ij}^{(k)} = \eta_{ij}^k .
$$
Here, the ZAGA is characterized by $\pij$, which denotes the edge-specific probability of an excess zero and by $\mu_{ij}^{(k)}$, which is the conditional mean of the positive traded value. The effects of covariates in $\eta_{ij}^{(k)}$ has to be interpreted as multiplicative effects on the conditional mean of positive trade, rather than directly on the probability of observing a nonzero
trade flow. We considered $q=1$ and $q=2$ latent variables. As both models yielded very similar results, we retained the model with one latent variable ($q=1$) due to it being more parsimonious.

Among the $m=1980$ directed country pairs initially considered, $47$ pairs do not trade any goods in any of the $K=72$ product layers. Since these all-zero edges provide no information on the positive component of the ZAGA model, we treat them as structural zeros for the purposes of estimation and exclude them from the fitted model. For these pairs, the zero probability is fixed at $\pij\equiv 1$, and no positive-trade parameters are estimated. Furthermore, we do not estimate the parameters on the $517$ countries pairs who trade less than $30$ goods. The fitted model is therefore estimated on the remaining $1416$ directed country pairs.

For the fitting of the zero component, we first convert the observed trade value into a binary
indicator $Y_{0,ij}^{(k)} = \mathbbm{1}\{Y_{ij}^{(k)} = 0\}$,
where $i$ denotes the exporter (partner) country, $j$ denotes the importer (reporter) country, and $k$ denotes the product layer. We then model the conditional probability that the trade value is zero by a Bernoulli generalized linear model with logit link  (fitted by maximum likelihood).
The numeric covariates used, namely, $\texttt{best}$ and $\texttt{dist}$, are centred and scaled. We also included factors $\texttt{contig}, \texttt{colony}, \texttt{comlang\_ethno}$ as well as exporter-, importer- and product-fixed effects.
In words, this part of the model estimates how observable edge-specific and product-level
covariates affect the probability that a given partner-reporter-product
trade flow is zero. The best applied tariff variable, distance, common language, colonial
history, and contiguity capture observed trade-resistance and affinity effects.
The partner and reporter fixed effects absorb country-specific propensities to appear in zero trade flows, while the product fixed effects absorb baseline
differences in sparsity across commodities. Hence, the fitted values
$\hat{\pi}_{ij}^{(k)}$ are estimated probabilities of observing a zero
trade value for each country pair and product layer. \Autoref{tab:wto-glamle-pi-estimates} reports the estimates and their associated standard errors for the zero part of the ZAGA model. As expected, the probability of observing no trade increases with increasing distance and/or tariff, see also \autoref{fig:wto-best-dist-pi}.  Furthermore, the presence of colonial ties, contiguity and/or common language decrease the probability of no trade happening. 

\begin{table}[tb]
    \centering
    \caption{GLAMLE estimates for the WTO trade data for the zero part of the ZAGA model. Numerical covariates are centred and scaled, their estimates are transformed back and presented on the original scale. All values reported are on the $10^3$ scale. The table only presents the estimates of the global parameters.}
    \label{tab:wto-glamle-pi-estimates}
    \begin{tabular}{lrrr}
     \toprule 
 Parameter & Estimate & Standard error  \\
 \midrule 
     Intercept &-2553.29&106.703 \\
     \texttt{best} &7.414&0.605 \\
     \texttt{dist} &0.126 &0.002 \\
     \texttt{comlang\_ethno} &-291.789 &27.015\\
     \texttt{colony} &-594.58 &76.447\\
     \texttt{contig} &-1751.404 &68.006\\
    \bottomrule
    \end{tabular}
\end{table}

To fit the positive part of the model, the numeric covariates  \texttt{best} and \texttt{dist}  were included (centred and scaled), in addition to factors \texttt{contig}, \texttt{colony} and \texttt{comlang\_ethno}. Due to the graph setup, exporter-important countries are included in the model as edges, and products act as observed graph layers. Table~\ref{tab:wto-glamle-estimates-colony-contig} of the main document reports the estimates of the global
parameters of the fitted ZAGA-GLAMLE model on the positive. All the estimated parameters are reported on their natural scales. Standard errors are computed from the sandwich covariance matrix associated with the Laplace-approximated log-likelihood. 

Interpretation of the regression coefficients requires some care. The coefficients in the log-mean component act on the positive Gamma mean, $\hat{\mu}_{ij}^{(k)}$, whereas the fitted mean of the response conditional on the latent variables  is
$$
 \hat{E}(\Yijk \mid \bztkh) = (1-\hat{\pi}_{ij}^{(k)})\hat{\mu}_{ij}^{(k)} .
$$
Thus, a positive coefficient in the Gamma component need not imply a large fitted trade value on the observed scale. In particular, a dyad may have a large fitted positive-trade mean but still have a small conditional on the latent variables fitted mean when the estimated probability of structural or excess zero trade is high.

For the positive Gamma component, \Autoref{tab:wto-glamle-estimates-colony-contig} provides little evidence of an association between \texttt{best} and the positive-trade mean after adjustment for the remaining model components. The estimated coefficient is positive, but imprecisely estimated, and its confidence interval contains zero. We therefore do not attach substantive meaning to its sign. Refitting the model without \texttt{best} yields very similar coefficients for other covariates. Inspection of the data suggests that the largest tariffs are concentrated in a few product classes, mainly alcohol, cereals, dairy products and tobacco.
For tobacco and alcohol, positive trade flows tend to remain large even when tariffs are high. 
For cereals, the largest values of \texttt{best} are primarily associated with imports by the Republic of Korea across several partners, consistent with tariff protection of domestic agriculture. In this case the estimated mean zero probability is moderate, for example $\bar{\hat{\pi}}_{i,\mathrm{Korea}}^{\mathrm{Cereals}}=0.37$, but, conditional on positive trade, the traded values are large. A similar pattern is observed for dairy products, where Switzerland, Canada, Turkey and Norway impose high tariffs on some partners, while positive flows, when present, remain non-negligible. 
Overall, these patterns suggest that tariffs may play a secondary role in explaining the magnitude of positive trade flows once trade occurs; historical, institutional and long-term trading relationships may be more relevant for the positive-trade intensity.

The coefficient of \texttt{dist} is positive and precisely estimated across the fitted models. This should not be interpreted as a marginal gravity-type effect on observed trade values. In the zero-adjusted Gamma specification, the coefficient acts on the conditional mean of the positive Gamma component, not on the unconditional mean of trade. A plausible explanation is selection into positive trade. Larger distance may reduce the probability that trade occurs, through the zero-adjustment component, while the positive flows that remain observed at large distances may be larger because small long-distance shipments are economically unattractive or less likely to be recorded. Thus distance can affect the two components of the model in opposite directions: it may increase the zero probability while increasing the conditional positive mean. Since \texttt{dist} is measured in kilometres, the coefficient per kilometre is numerically small. However, over 1000 kilometres, the fitted multiplicative effect on the positive Gamma mean is $\exp(1.435) \approx 4.2$. The raw scatterplots are therefore not a direct representation of this partial model-based effect, since they are dominated by zeros and by a small number of very large flows. The estimate is better interpreted as follows: among the part of the dyad-specific positive-trade baseline explained by the global layer-independent covariates, distance contributes positively to the conditional positive-trade mean.

The effect of \texttt{comlang\_ethno} is large: $\exp(4.816) \approx 123$. This suggests that common ethnological language acts as a strong edge-level baseline separator in the positive-trade component. In other words, after adjustment for the latent structure and other covariates, dyads sharing an ethnological language have a substantially larger conditional positive-trade mean.

As to \texttt{colony}, the coefficient is $\exp (4.139) \approx 63$, indicating (ex-)colonies trade slightly more. Finally, the coefficient for \texttt{contig} is $\exp (15.730) \approx 6780151$, which suggests sharing a common border results in a much higher positive trade value. 

Edge-specific quantities, including
$\hat\pi_{ij}$ and the latent loading estimates, are summarized in
Table~\ref{tab:wto-edge-summary-colony-contig}. 
The edge-specific estimates reveal substantial heterogeneity across country pairs.
Rather than reporting all edge-level parameters, we
summarize their empirical distributions and focus on model-derived quantities such
as fitted zero probabilities, and fitted positive trade intensities. These summaries provide a more interpretable description
of the fitted multiview network structure.

Finally, the fitted value of $\sigma$ indicates substantial dispersion among positive trade values. This confirms that, even after accounting for covariates, dyad effects and latent layer structure, the distribution of positive trade flows remains highly heterogeneous.

\begin{table}[!ht]
    \centering
    \caption{Summary of edge-specific estimates (rounded up to three decimals), $q=1$.}
    \label{tab:wto-edge-summary-colony-contig}
    \begin{tabular}{lrrrrr}
     \toprule 
 Parameter & Mean & SD & Median & IQR & Range \\
 \midrule 
            $\hat{\pi}_{ij}$ & 0.374 &0.355 & 0.245 & 0.679 & (0,1) \\
            $\hat{\alpha}_{(0),ij}$ & 3.202 &7.572&3.424&10.625 & (-18.266,24.485) \\
            $\hat{\alpha}_{(2),ij}$ & 1.588 &1.359&1.648&1.627 & (-2.912,6.417) \\
    \bottomrule
    \end{tabular}
\end{table}

\begin{table}[!ht]
    \centering
    \caption{Summary of edge-specific estimates (rounded up to three decimals), $q=2$.}
    \label{tab:wto-edge-summary-q2}
    \begin{tabular}{lrrrrr}
     \toprule 
 Parameter & Mean & SD & Median & IQR & Range \\
 \midrule 
            $\hat{\pi}_{ij}$ & 0.374 &0.355 & 0.245 & 0.679 & (0,1) \\
            $\hat{\alpha}_{(0),ij}$ & 3.165 &7.451&3.375&10.367 & (-17.848,24.263) \\
            $\hat{\alpha}_{(2),ij}$ & 0.984 &1.425&1.12&1.636 & (-4.23,5.599) \\
    \bottomrule
    \end{tabular}
\end{table}

\subsection{Comparing ZAGA-GGLLVM and PPMLE} \label{Sec: Comparison}
To assess the performance of GGLLVM, and similarly to Monte-Carlo exercises described in \Autoref{Sec: GLPP_comp}, we compared the zero, positive and total parts of fitted ZAGA-GLAMLE model to a fitted PPMLE. For each observation corresponding to dyad $(i,j)$ and product layer $k$,
let $Y_{ij}^{(k)}$ denote the observed trade volume. After removing observations with missing fitted values, we fit

\begin{equation}
  \label{wto:ppml}
  \Yijk \sim \mathrm{Poisson}\left(\mijk\right),
  \qquad
  \log \mijk
  =
  \beta_{0}
  +
  \beta_{best} \mathrm{best}_{ij}^{(k)}
  +
  \beta_{dist}\log \mathrm{dist}_{ij} +
  \bg^\T \bw_{ij}
  +
  \mathrm{importer}_{i}
  +
  \mathrm{exporter}_{j},
\end{equation}

with $\texttt{dist}$ included separately and $\bw$ including the rest of the layer-independent covariates as in the fitted ZAGA-GLAMLE model, i.e. $\texttt{comlang\_ethno}, \texttt{colony}, \texttt{contig}$.

We then compute the fitted mean for the PPMLE model $\hat{\mu}^{\mathrm{PPMLE}}_{ij,k}$ and the fitted Gamma mean $\hat{\mu}^{\mathrm{GLAMLE}}_{ij,k}$, as well as the fitted zero-inflation probability $
\hat{\pi}^{\mathrm{GLAMLE}}_{ij}$. The fitted  mean conditional on the latent variable under the ZAGA-GLAMLE model is therefore
$$
\hat{\mathbb{E}}_{\mathrm{GLAMLE}}( \Yijk \mid \bztkh)
=
\left(1-\hat{\pi}^{\mathrm{GLAMLE}}_{ij}\right)
\hat{\mu}^{\mathrm{GLAMLE}}_{ij,k}.
$$

Let
$$
\mathcal{I}_{k}
=
\left\{(i,j): 
\hat{\mathbb{E}}_{\mathrm{GLAMLE}}(\Yijk \mid \bztkh),
\widehat{\mu}^{\mathrm{PPMLE}}_{ij,k},
\widehat{\pi}^{\mathrm{GLAMLE}}_{ij},
\widehat{\mu}^{\mathrm{GLAMLE}}_{ij,k}
\text{ are available}
\right\}
$$
denote the set of dyads retained for product $k$, and let $n_k = |\mathcal{I}_{k}|$. Furthermore, we have 
$$
\hat{\mathbb{E}}_{\mathrm{PPMLE}}( \Yijk) = \widehat{\mu}^{\mathrm{PPMLE}}_{ij,k}.
$$
For each product $k$, and each model $\mathcal{M} \in \left\{\mathrm{ZAGA-GLAMLE}, \mathrm{PPMLE}\right\}$ we compute the mean absolute error
as
$$
\mathrm{MAE}^{\mathcal{M}}_{k}
=
\frac{1}{n_k}
\sum_{(i,j)\in\mathcal{I}_{k}}
\left|
Y_{ij}^{(k)}
-
\hat{\mathbb{E}}_{\mathcal{M}}(\cdot)
\right|.
$$

To assess calibration of the zero-probability component, we compute the Brier score for the binary event $\{Y_{ij}^{(k)}=0\}$ as
$$
\mathrm{Brier}^{\mathcal{M}}_{0,k}
=
\frac{1}{n_k}
\sum_{(i,j)\in\mathcal{I}_{k}}
\left(
\mathbf{1}\{Y_{ij}^{(k)}=0\}
-
\widehat{\pi}^{\,\mathcal{M}}_{0,ij,k}
\right)^2,
$$

where $\mathbf{1}\{Y_{ij}^{(k)}=0\}$ is the observed zero indicator and $\widehat{\pi}^{\,\mathcal{M}}_{0,ij,k}$ is either the fitted zero probability (ZAGA-GLAMLE) or a working Poisson zero probability (PPMLE).

The MAE and RMSE compare the fitted means conditional on the latent variables, whereas the Brier
scores compare the fitted probabilities of observing zeros. Hence, the former
assess predictive accuracy for trade values, while the latter assess calibration of the zero event. This distinction is important because ZAGA-GLAMLE contains an explicit zero component, whereas PPMLE only induces a working zero probability through the Poisson identity. 
\Autoref{fig:wto-glamle-vs-ppml-cottoncars}(a) shows the layer-wise Zero Brier scores for both models, while \Autoref{fig:wto-glamle-vs-ppml-mae} shows MAE. 

\begin{figure}
  \centering
    \includegraphics[width=0.5\textwidth]{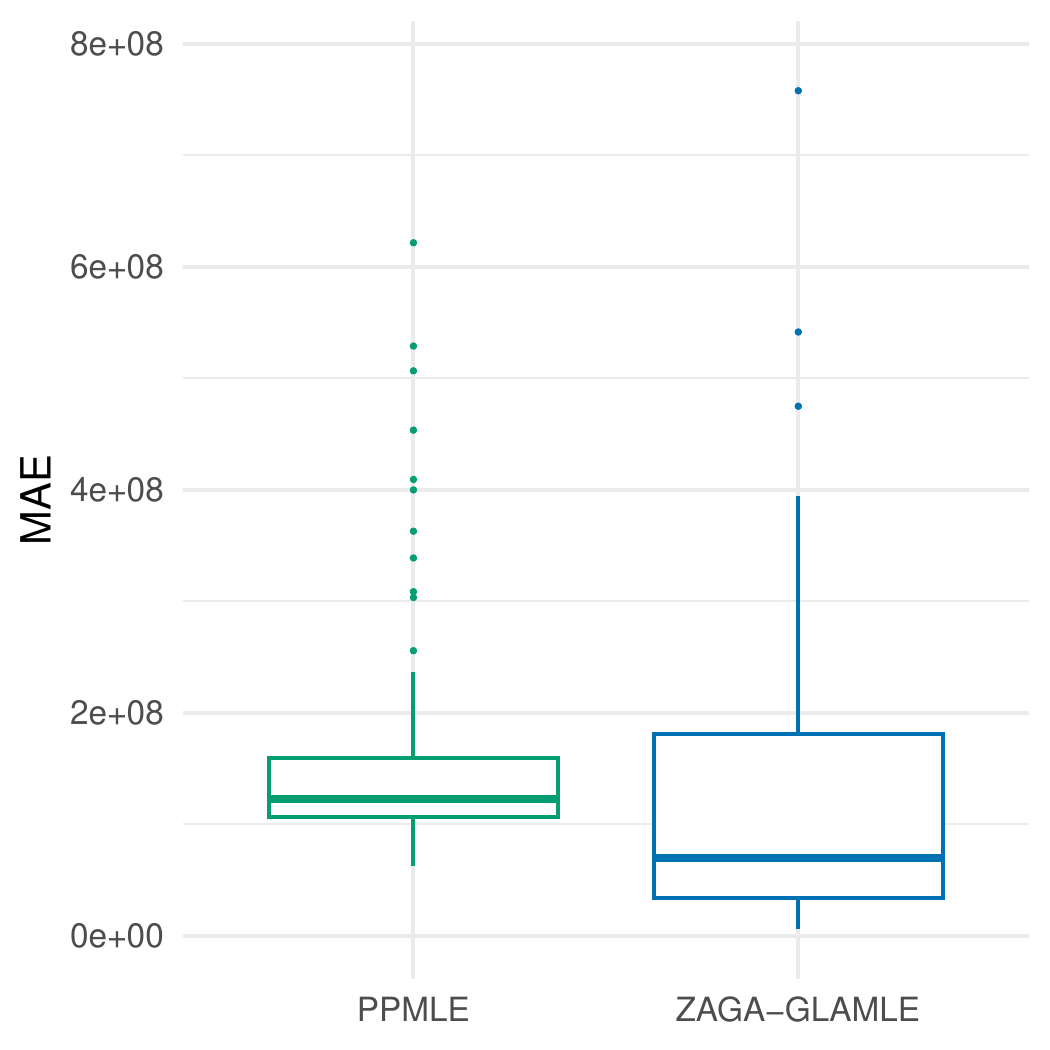}
    \caption{Layer-wise MAE for WTO data comparing ZAGA-GLAMLE with PPMLE.}
    \label{fig:wto-glamle-vs-ppml-mae}
\end{figure}

Finally, we have fitted the model for $q=2$, and compared it to $q=1$. \Autoref{tab:wto-edge-summary-q2} depicts the empirical summaries of the edge-specific parameters, which appear to be quite close to ones in \Autoref{tab:wto-edge-summary-colony-contig} ($q=1$). Since the zero parts are identical under the current setup, we only compare the positive parts. The MAE plots are presented in \Autoref{fig:wto-glamle-q1-vs-q2-mae}. Since the median MAE is very similar to $q=1$ ($69'345'149$ vs $68'873'592$ for $q=1$ and $q=2$, respectively), we chose the most parsimonious model as our final model. 

\begin{figure}
  \centering
    \includegraphics[width=0.5\textwidth]{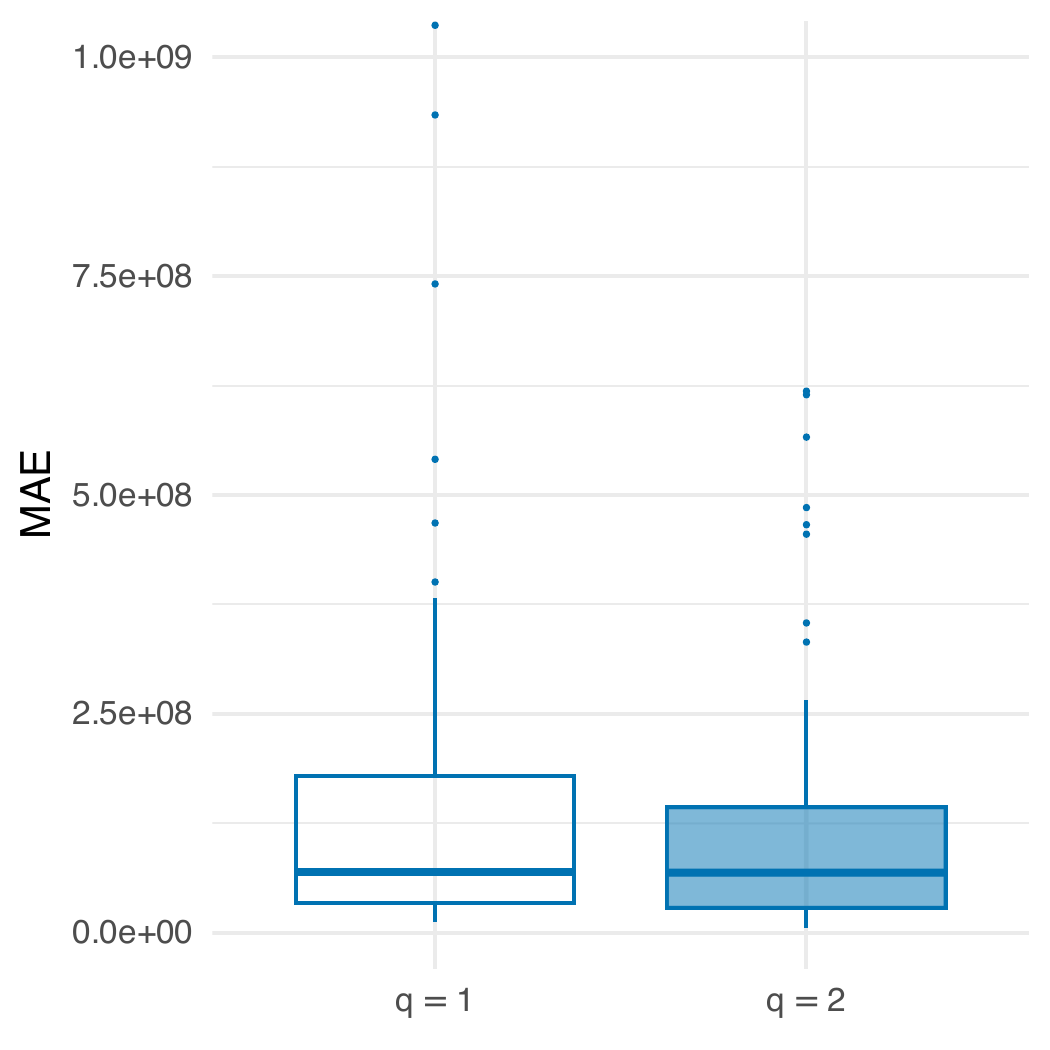}
    \caption{Layer-wise MAE for WTO data comparing ZAGA-GLAMLE with $q=1$ and $q=2$.}
    \label{fig:wto-glamle-q1-vs-q2-mae}
\end{figure}

\clearpage
\newpage

\end{document}